\documentclass[12pt,notitlepage,a4paper]{article}
\usepackage{amsmath}
\usepackage{amssymb}
\usepackage{amsfonts}
\usepackage[onehalfspacing,nodisplayskipstretch]{setspace}
\usepackage[a4paper,nohead,left=1in,right=1in,top=1in,bottom=1in]{geometry}

\newtheorem{theorem}{Theorem}

\newtheorem{lemma}[theorem]{Lemma}

\newtheorem{proposition}[theorem]{Proposition}

\newenvironment{proof}[1][Proof]{\noindent\textbf{#1.} }{\hfill $\square$}
\input{tcilatex}
\begin{document}

\title{\textquotedblleft Calibeating": Beating Forecasters at Their Own Game%
\thanks{%
Previous versions: February 2020; October 2021 (Hebrew University of
Jerusalem, Center for Rationality DP-743), May 2022, September 2022 (\texttt{%
arXiv:2209.04892}), October 2022 (\texttt{arXiv:2209.04892v2}). Journal
(shortened) version: \emph{Theoretical Economics} 18 (2023), 4, 1441-1474.
This version corrects an error in Appendix A.7 (Foster and Hart 2026), and
adds a new Appendix A.10 (Foster and Hart 2024). We thank Drew Fudenberg,
Benjy Weiss, the coeditor, and the referees for useful comments and
suggestions. A presentation is available at \texttt{%
http://www.ma.huji.ac.il/hart/pres.html\#calib-beat-p}}}
\author{Dean P. Foster\thanks{%
Department of Statistics, Wharton, University of Pennsylvania, Philadelphia,
and Amazon, New York. \emph{e-mail}: \texttt{dean@foster.net} \ \emph{web
page}: \texttt{http://deanfoster.net}} \and Sergiu Hart\thanks{%
Institute of Mathematics, Department of Economics, and Federmann Center for
the Study of Rationality, The Hebrew University of Jerusalem. \emph{e-mail}: 
\texttt{hart@huji.ac.il} \ \emph{web page}: \texttt{%
http://www.ma.huji.ac.il/hart}}}
\maketitle

\begin{abstract}
In order to identify expertise, forecasters should not be tested by their
calibration score, which can always be made arbitrarily small, but rather by
their Brier score. The Brier score is the sum of the calibration score and
the refinement score; the latter measures how good the sorting into bins
with the same forecast is, and thus attests to \textquotedblleft expertise."
This raises the question of whether one can gain calibration without losing
expertise, which we refer to as \textquotedblleft calibeating." We provide
an easy way to calibeat any forecast, by a deterministic online procedure.
We moreover show that calibeating can be achieved by a stochastic procedure
that is itself calibrated, and then extend the results to simultaneously
calibeating multiple procedures, and to deterministic procedures that are
continuously calibrated.
\end{abstract}

\tableofcontents

\def\@biblabel#1{#1\hfill}
\def\thebibliography#1{\section*{References}
\addcontentsline{toc}{section}{References}
\list
{}{
\labelwidth 0pt
\leftmargin 1.8em
\itemindent -1.8em
\usecounter{enumi}}
\def\newblock{\hskip .11em plus .33em minus .07em}
\sloppy\clubpenalty4000\widowpenalty4000
\sfcode`\.=1000\relax\def\baselinestretch{1}\large \normalsize}
\let\endthebibliography=\endlist
\newcommand{\shfrac}[2]{\ensuremath{{}^{#1} \hspace{-0.04in}/_{\hspace{-0.03in}#2}}}
\newcommand\T{\rule{0pt}{2.6ex}}
\newcommand\B{\rule[-1.2ex]{0pt}{0pt}}%

\section{Introduction}

Forecasters---whether of weather or of events like elections and
sports---make probabilistic predictions, such as \textquotedblleft the
probability of rain is $p.$\textquotedblright\ What does it mean, and how
does one test whether it is any good? Taking the classic view of probability
as long-run frequency, the above prediction translates to \textquotedblleft
in the days when the forecast is $p$ the frequency of rain is close to $p$
in the long run.\textquotedblright\ If this holds for all values of $p$ used
as forecasts, one says that the forecaster is \emph{calibrated. }There is a
large literature on calibration; see the survey of Olszewski (2015), and the
recent paper of Foster and Hart (2021), which also discusses the economic
utility of calibration (see Section I.A there).

The \emph{calibration score} $\mathcal{K}$ is defined as the average squared
distance between forecasts and realized (relative) frequencies (i.e., the
proportion of, say, rainy days), where each forecast is weighted by how
often it has been used; evaluated after $t$ days, this yields%
\begin{equation*}
\mathcal{K}=\frac{1}{t}\sum_{s=1}^{t}\left( c_{s}-\bar{a}(c_{s})\right) ^{2},
\end{equation*}%
where $c_{s}$ is the forecast at time $s$ and for each $p$ we denote by $%
\bar{a}(p)\equiv \bar{a}_{t}(p)$ the frequency of rain in the days from $1$
to $t$ in which the forecast was $p$ (giving weight $1/t$ to each day is the
same as weighting each forecast by the proportion of days it has been used).
Being calibrated means that $\mathcal{K}$ is (close to) $0.$

A classic and surprising result of Foster and Vohra (1998) is that one can
generate forecasts that are \emph{guaranteed} to be calibrated, no matter
what the weather will be. This immediately casts some doubt on whether
calibration is the appropriate way to test the expertise of forecasters.
(There is an extensive literature on \textquotedblleft experts" that uses
calibration tests to check whether they are indeed experts; see, e.g., the
book of Cesa-Bianchi and Lugosi 2006 and the survey of Olszewski 2015. The
fact that the calibration score is not the right way to identify experts
does \emph{not} imply that calibration should be ignored---on the contrary,
calibration is a useful property for forecasts to satisfy; see Section I.A
in Foster and Hart 2021.)

\begin{figure}[htbp] \centering%
\renewcommand{\arraystretch}{1.5}%
\begin{tabular}{c||c|c|c|c|c|c|c||c|c|c|}
\textbf{Day} & $1$ & $\;\;2\;\;$ & $3$ & $\;\;4\;\;$ & $5$ & $\;\;6\;\;$ & $%
\;\;...\;\;$ & $\;\;\mathcal{K\;\;}$ & $\;\;\mathcal{R\;\;}$ & $\;\;\mathcal{%
B\;\;}$ \\ \hline\hline
\textbf{Rain} & $1$ & $0$ & $1$ & $0$ & $1$ & $0$ &  &  &  &  \\ \hline
\textbf{F1} & $100\%$ & $0\%$ & $100\%$ & $0\%$ & $100\%$ & $0\%$ &  & $0$ & 
$0$ & $0$ \\ \hline
\textbf{F2} & $50\%$ & $50\%$ & $50\%$ & $50\%$ & $50\%$ & $50\%$ &  & $0$ & 
$0.25$ & $0.25$ \\ \hline
\end{tabular}%
\renewcommand{\arraystretch}{1}%

\caption{Two calibrated forecasts\label{fig1}}%
\end{figure}%

Take the following simple and well-known example (see Figure \ref{fig1}).
Suppose that the weather alternates between rain on odd days and no rain on
even days. Consider two rain forecasters: F1 forecasts $100\%$ on odd days
and $0\%$ on even days, and F2 forecasts $50\%$ every day. While both
forecasts are well calibrated (the calibration score $\mathcal{K}$ of F1 is $%
0$ every day, and that of F2 is $0$ on even days and $\approx 0$,
specifically, $1/(4t^{2})$, on odd days), F1 is clearly a much better and
more useful forecaster than F2.

The difference between the two forecasts is underscored by appealing to the
classic \emph{Brier} (1950) \emph{score }$\mathcal{B}$, which measures how
close the forecasts and the realizations are, by the standard mean squared
error formula:%
\begin{equation*}
\mathcal{B}=\frac{1}{t}\sum_{s=1}^{t}\left( c_{s}-a_{s}\right) ^{2},
\end{equation*}%
where $a_{s}$ denotes the weather on day $s$, with $a_{s}=1$ standing for
rain and $a_{s}=0$ for no rain, and $c_{s}$ is, as above, the forecast on
day $s$. For F1 the Brier score $\mathcal{B}$ is $0$ every day (because $%
c_{s}=a_{s}$ for all $s)$, whereas for F2 it is $1/4$ every day (because $%
(0.5-1)^{2}=(0.5-0)^{2}=1/4)$. The Brier score thus distinguishes well
between the two forecasters ($\mathcal{B}=0$ vs. $\mathcal{B}=1/4$), while
the calibration score does not ($\mathcal{K}=0$ for both).

To interpret this difference in the Brier scores, view forecasting as
consisting of two separate ingredients. The first one is the
\textquotedblleft classification" or \textquotedblleft sorting" of days into
\textquotedblleft bins," where all the days with the same forecast $p$ are
assigned to the same bin. The second one is the specific value of the
forecast $p$ that is used to define each bin, which we refer to as the
\textquotedblleft label" of the bin. In the above example, F1 sorts the days
into two bins, a $100\%$-bin, which consists of the odd days, and a $0\%$%
-bin, which consists of the even days, whereas for F2 there is a single bin,
the $50\%$-bin, which contains all days. Both bins of F1 are homogeneous:
there is no variance among the days in the same bin (they are either all
\textquotedblleft rain," or all \textquotedblleft no rain"); by contrast, in
the single bin of F2 there is a high variance among the days (half of them
are \textquotedblleft rain" and half \textquotedblleft no rain"). This
\textquotedblleft within-bin variance" is captured by the \emph{refinement
score} $\mathcal{R}$, which is the average squared distance between the
weather $a_{s}$ and the bin-average weather (which is the average frequency
of rain on the days from $1$ to $t$ that are in the $c_{s}$-bin, i.e., on
those days when the forecast was the same as on day $s$), denoted by $\bar{a}%
(c_{s})$:%
\begin{equation*}
\mathcal{R}=\frac{1}{t}\sum_{s=1}^{t}\left( a_{s}-\bar{a}(c_{s})\right) ^{2}.
\end{equation*}

The Brier score neatly decomposes into the sum of the refinement and the
calibration scores, 
\begin{equation*}
\mathcal{B}=\mathcal{R}+\mathcal{K}
\end{equation*}%
(this easily follows from the equality $\mathbb{E}\left[ X^{2}\right] =%
\mathbb{V}ar\left[ X\right] +\left( \mathbb{E}\left[ X\right] \right) ^{2}$;
see Section \ref{sus:scores}). The refinement score $\mathcal{R}$ yields the
average of the within-bin variances, and the calibration score $\mathcal{K}$
the average squared distance between the bin labels and the bin averages.
Perfect calibration, i.e., $\mathcal{K}=0$, says that all the labels are
correct: the label of each bin, i.e., the value of the forecast that defines
the bin, is equal to the average weather of the bin. In addition, the
refinement score $\mathcal{R}$ and the calibration score $\mathcal{K}$ are
\textquotedblleft orthogonal": changing the labels does not affect $\mathcal{%
R}$ (indeed, $\mathcal{R}$ is the \textquotedblleft relabeling-minimum"
Brier score; see Section \ref{sus:general B}), and changing the distribution
of actions within each bin without changing their average does not affect $%
\mathcal{K}$. Returning to the example, we have $\mathcal{R}=$ $\mathcal{K}=%
\mathcal{B}=0$ for all $t$ for F1, and $\mathcal{R}\approx 1/4$, $\mathcal{K}%
\approx 0$, $\mathcal{B}=1/4$ for all $t$ for F2 (for perfect classification
without calibration, use, for instance, the forecast $75\%$ on odd days and
the forecast $25\%$ on even days: $\mathcal{R}=0$ and $\mathcal{K}=1/16$ for
all $t$).

Thus, our first conclusion is

\begin{quote}
\underline{\emph{Conclusion}}\emph{: Experts should better be tested by the
Brier score and not by calibration alone.}
\end{quote}

Unlike the calibration score, the Brier score cannot in general be brought
down to zero in the long run. Indeed, for an i.i.d.$\ 50\%$ probability of
rain, the refinement score $\mathcal{R}$ is close to $(1/2)\cdot (1/2)=1/4$
for \emph{any} forecasting sequence (because this is the variance of each
bin), and thus the Brier score $\mathcal{B}$ is at least $1/4$. However, if
there are certain \textquotedblleft regularities" or \textquotedblleft
patterns" in the weather, then an expert forecaster who recognizes them can
get a lower refinement score. For example, suppose that it is very likely
that when it rains, it does so for precisely two consecutive days; this
means a high probability, say $90\%$, that $1$ comes after $01$ and also
that $0$ comes after $011$ (where $1$ stands for rain and $0$ for no rain).
For a forecaster that forecasts $p_{1}$ if and only if the last two days
were $01$, and forecasts $p_{2}$ (different from $p_{1}$) if and only if the
last three days were $011$, the $p_{1}$-bin and the $p_{2}$-bin each have a
low variance of $0.9\cdot 0.1=0.09$. Knowledge about the weather, which we
refer to as \emph{expertise}, is thus reflected in sorting the days into
bins that consist of similar days, and in making the binning as refined as
possible (which can only decrease $\mathcal{R}$; see Section \ref{s:self}
and Appendix \ref{s-a:refined})---that is, in having a low refinement score $%
\mathcal{R}.$

Returning to calibration, a forecaster can always guarantee its forecasts to
be calibrated, by the Foster and Vohra (1998) result. However, this would
require it to run one of the calibration procedures (some of which---like
the \textquotedblleft forecast-hedging" one of Section 5 of Foster and Hart
2021---are extremely simple) and ignore whatever expert knowledge he has
about the weather, and whichever patterns he has identified in the data.

Thus, the natural question that arises is

\begin{quote}
\underline{\emph{Question}}\emph{: Can one gain calibration without losing
expertise?}
\end{quote}

\noindent In formal terms, can one decrease $\mathcal{K}$ to zero without
increasing $\mathcal{R}$?

This can of course always be done \emph{in retrospect}: replacing each
forecast $p$ with the corresponding bin average $\bar{a}(p)$ yields
calibration while preserving the binning, and thus the refinement score $%
\mathcal{R}$. For example, if the frequency of rain on the days when the
forecast was $70\%$ turned out to be $40\%$, then each forecast of $70\%$ is
\textquotedblleft corrected" to $40\%$. The new calibration score is then
zero, i.e., $\mathcal{K}^{\prime }=0$, while the refinement score is
unchanged, i.e., $\mathcal{R}^{\prime }=\mathcal{R}$; therefore, the Brier
score is decreased by the calibration score: $\mathcal{B}^{\prime }=\mathcal{%
B-K}$ (because $\mathcal{B}^{\prime }=\mathcal{K}^{\prime }+\mathcal{R}%
^{\prime }=0+\mathcal{R=R}$ and $\mathcal{R=B-K}$). We will call this

\begin{quote}
\emph{\textquotedblleft }\underline{\emph{Calibeating}}\emph{": Beating the
Brier score by an amount equal to the calibration score}.
\end{quote}

\noindent The calibeating described above is however obtained only in
retrospect---\emph{offline}---since the bin averages are known only at the
time $t$ when the testing is done. Moreover, the forecast corrections depend
on the testing horizon $t$, since the average frequency of rain may well
change over time: $\bar{a}_{t}(p)$ and $\bar{a}_{t^{\prime }}(p)$ may be
quite different for $t\neq t^{\prime }$.

The interesting question is then what can be done \emph{online}, by a
procedure where the forecast of each day $s$ may be modified on the basis of
what is known at that time only and nothing beyond it (i.e., neither the
upcoming weather on day $s$, nor the future weather and forecasts on days
after $s$). Our main result is

\begin{quote}
\underline{\emph{Result}}\emph{: One can guarantee online calibeating of
forecasts.}
\end{quote}

The first result (Theorem \ref{th:beat-1} in Section \ref{s:simple}) shows
that this can be achieved by a simple online procedure: replace each
forecast by the average frequency of rain on the \emph{previous} days in
which this forecast was made. This attains---\emph{online}---the same
lowering of the Brier score by the calibration score that is obtained by the
above offline correction. We emphasize that this calibeating is achieved for
weather and forecasts that are arbitrary (and not stationary in any way),
for sorting into bins that may be far from perfect, and for bin averages
that need not converge; moreover, everything is guaranteed uniformly, even
against a so-called \textquotedblleft adversary." The proof uses a neat
online estimation of the variance.

Thus, any forecast that is not calibrated can be beaten, online, by another
forecast with a strictly better (i.e., lower) Brier score. An alternative
interpretation of the result takes a forecasting procedure and announces
every period, instead of the intended forecast, its corresponding
calibeating replacement (as described in the previous paragraph). This
generates a new forecasting procedure, whose Brier score is lower than that
of the original one---a clear improvement. This may apply, for instance, to
\textquotedblleft online regression" or \textquotedblleft online
least-squares" procedures, introduced by Foster (1991)---see also Forster
(1999), Vovk (2001), Azoury and Warmuth (2001), and Cesa-Bianchi and Lugosi
(2006)---which minimize the Brier score directly, and need not be calibrated
in general.

Now the calibeating procedure of our first result need not be calibrated
itself, which means that it may be calibeaten too. To avoid this, our second
result (Theorem \ref{th:calibration} in Section \ref{s:beat-by-calibrated})
provides a calibeating procedure that is guaranteed to be calibrated, by
appealing to a \textquotedblleft stochastic fixed point" result, namely, the
stochastic \textquotedblleft outgoing minimax" tool of Foster and Hart
(2021). The calibeating in this case thus yields $\mathcal{K}^{\prime }=0$
and $\mathcal{B}^{\prime }=\mathcal{R}^{\prime }\leq \mathcal{R}.$

The procedure of this second result is \emph{stochastic}, as it must be in
order to guarantee calibration (cf. Dawid 1982, Oakes 1985, and Foster and
Vohra 1998). However, if the calibration requirement is weakened to \emph{%
continuous} calibration---a concept introduced in Foster and Hart (2021),
which implies smooth and weak calibration as well, and suffices for
equilibrium dynamics---we obtain (Theorem \ref{th:cont} in Section \ref%
{s:beat-by-calibrated} and Theorem \ref{th:cont-calib} in Appendix \ref%
{s:cont-calib}) \emph{deterministic} calibeating procedures that are
continuously calibrated. This requires the use of a fixed point tool,
specifically, the \textquotedblleft outgoing fixed point" result of Foster
and Hart (2021); see Section III.D there, and Appendix \ref{sus:fp-mm
procedures} here, for the distinction between minimax and fixed point
methods.

Next, we show that all the above results can be extended to simultaneously
calibeating multiple forecasters (Theorem \ref{th:multi} in Section \ref%
{s:multi-beat}).

Finally, we comment on the use of the quadratic scores (such as $\left\Vert
a-c\right\Vert ^{2}$). This is standard in statistics (e.g., analysis of
variance and linear regression), as it easily leads to useful
decompositions, such as the Brier score being the sum of the refinement and
calibration scores here. However, it raises the question of how much do our
results depend on using the quadratic scores.\footnote{%
We thank the referee who posed this question.} While a general analysis is
beyond the scope of the present paper, we believe that the ideas and
approach here carry through for other scoring functions; in Appendix \ref%
{s-a:log} we show this for another classic scoring rule, the logarithmic one.

To summarize the contribution of this paper: we address the frequently asked
question of how to get better forecasts when there is some expertise. We
argue that expertise should better be tested by the Brier score and not just
by calibration, and show how to calibeat forecasts that are not calibrated:
lower their Brier score by at least their calibration score, without losing
the expertise embodied in these forecasts.

\section{The Setup\label{s:setup}}

Let $A$ be the set of possible outcomes, which we call \emph{actions}, and
let $C$ be the set of \emph{forecasts} about these actions. We assume that $%
C\subset \mathbb{R}^{m}$ is a nonempty compact convex subset of a Euclidean
space, and that $A\subseteq C$. Some examples: (i) $A=\{0,1\}$, with $a=1$
standing for \textquotedblleft rain\textquotedblright\ and $a=0$ for
\textquotedblleft no rain,\textquotedblright\ and $C=[0,1]$, with $c$ in $C$
standing for \textquotedblleft the chance of rain is $c$"; (ii) more
generally, $C$ is the set of probability distributions $\Delta (A)$ on a
finite set $A$, i.e., a unit simplex (we identify the elements of $A$ with
the unit vectors of $C$); (iii) $C$ is the convex hull $\mathrm{conv}(A)$ of 
$A$. Let $\gamma :=\mathrm{diam}(C)\equiv \max_{c,c^{\prime }\in
C}\left\Vert c-c^{\prime }\right\Vert $ denote the \emph{diameter} of the
set $C$. Let $\delta >0;$ a subset $D$ of $C$ is a $\delta $-\emph{grid} of $%
C$ if for every $c\in C$ there is $d\in D$ at a distance of less than $%
\delta $ from $c$, i.e., $\left\Vert d-c\right\Vert <\delta ;$ a compact set 
$C$ always has a finite $\delta $-grid (obtained from a finite subcover by
open $\delta $-balls).

The time periods are indexed by $t=1,2,..$. . An \emph{action sequence} is $%
\mathbf{a}=(a_{t})_{t\geq 1}$ with $a_{t}\in A$ for all $t$, and we write $%
\mathbf{a}_{t}=(a_{s})_{1\leq s\leq t}$ for its first $t$ elements;
similarly, a \emph{forecasting sequence} is $\mathbf{c}=(c_{t})_{t\geq 1}$
with $c_{t}\in C$ for all $t$, and we put $\mathbf{c}_{t}=(c_{s})_{1\leq
s\leq t}.$

\subsection{The Calibration, Refinement, and Brier Scores\label{sus:scores}}

Fix a time horizon $t$. For each possible forecast $x$ in $C$ let\footnote{%
The number of elements of a finite set $Z$ is denoted by $|Z|.$}%
\begin{eqnarray*}
n_{t}(x) &%
{\;:=\;}%
&|\{1\leq s\leq t:c_{s}=x\}|, \\
\bar{a}_{t}(x) &%
{\;:=\;}%
&\frac{1}{n_{t}(x)}\sum_{1\leq s\leq t:c_{s}=x}^{{}}a_{s},\text{\ and} \\
v_{t}(x) &%
{\;:=\;}%
&\frac{1}{n_{t}(x)}\sum_{1\leq s\leq t:c_{s}=x}^{{}}\left\Vert a_{s}-\bar{a}%
_{t}(x)\right\Vert ^{2}
\end{eqnarray*}%
be, respectively, the \emph{number} of times that the forecast $x$ has been
used up to time $t$, and the action\emph{\ average} and \emph{variance} in
those periods; when $x$ has not been used, i.e., $n_{t}(x)=0$, we put for
convenience $v_{t}(x)%
{\;:=\;}%
0$ and (see below) $e_{t}(x)%
{\;:=\;}%
0$.

The \emph{calibration error} $e_{t}(x)$ of a forecast $x$ is the difference
between the action average and $x$, i.e., 
\begin{equation*}
e_{t}(x)%
{\;:=\;}%
\bar{a}_{t}(x)-x,
\end{equation*}%
and the \emph{calibration score }is the average square calibration error,
i.e.,\footnote{%
The sum is finite as it goes over all $x$ with $n_{t}(x)>0,$ i.e., over $x$
in the set $\{c_{1},...,c_{t}\}$.}%
\begin{equation*}
\mathcal{K}_{t}%
{\;:=\;}%
\sum_{x\in C}\left( \frac{n_{t}(x)}{t}\right) \left\Vert e_{t}(x)\right\Vert
^{2};
\end{equation*}%
thus, the error of each $x$ is weighted in proportion to the number of times 
$n_{t}(x)$ that $x$ has been used (the weights add up to $1$ because $%
\sum_{x}n_{t}(x)=t$). Since from $1$ to $t$ there are exactly $n_{t}(x)$
terms with $x=c_{s}$, this is equivalent to%
\begin{equation*}
\mathcal{K}_{t}=\frac{1}{t}\sum_{s=1}^{t}\left\Vert e_{t}(c_{s})\right\Vert
^{2}=\frac{1}{t}\sum_{s=1}^{t}\left\Vert \bar{a}_{t}(c_{s})-c_{s}\right\Vert
^{2}.
\end{equation*}%
We refer to $\mathcal{K}_{t}$ as the \textquotedblleft $\ell _{2}$%
-calibration score," to distinguish it from $K_{t}$ (note the different
font) that is used in other papers (e.g., Foster and Hart 2021 and Hart
2021), and which is the \textquotedblleft $\ell _{1}$-calibration score,"
i.e., the weighted average of $\left\Vert e_{t}(x)\right\Vert $ rather than $%
\left\Vert e_{t}(x)\right\Vert ^{2}$. The two scores are equivalent, since $%
(K_{t})^{2}\leq \mathcal{K}_{t}\leq \gamma K_{t}$ (the first inequality by
Jensen's inequality, the second by $\left\Vert e_{t}(x)\right\Vert \leq
\gamma $), and so $\mathcal{K}_{t}\rightarrow 0$ if and only if $%
K_{t}\rightarrow 0.$

The \emph{refinement score} is the average over all forecasts of the
corresponding action variances:%
\begin{equation*}
\mathcal{R}_{t}%
{\;:=\;}%
\sum_{x\in C}\left( \frac{n_{t}(x)}{t}\right) v_{t}(x);
\end{equation*}%
again, this is equivalently expressed as 
\begin{equation*}
\mathcal{R}_{t}=\frac{1}{t}\sum_{s=1}^{t}\left\Vert a_{s}-\bar{a}%
_{t}(c_{s})\right\Vert ^{2}.
\end{equation*}

Finally, the \emph{Brier} (1950) \emph{score,} 
\begin{equation*}
\mathcal{B}_{t}%
{\;:=\;}%
\frac{1}{t}\sum_{s=1}^{t}\left\Vert a_{s}-c_{s}\right\Vert ^{2},
\end{equation*}%
measures how close the forecasts $c_{s}$ are to the actions $a_{s}$ by a
standard mean of squared error formula. This is a so-called
\textquotedblleft strictly proper scoring rule," which means that if the
sequence $\mathbf{a}_{t}$ is generated by a probability distribution $%
\mathbb{P}$, then the unique minimizer of the expected Brier score is the
sequence $c_{s}=\mathbb{P}\left[ a_{s}|\mathbf{a}_{s-1}\right] $ of true
conditional probabilities (assume for simplicity that $C$ is the set of
probability distributions $\Delta (A)$ on a finite set $A$).

One may assume for convenience\footnote{%
See Foster and Hart (2021); this matters also when generalizing to
fractional binnings (Section \ref{s:cont-calib}).} that one assigns to the
bins the \emph{differences} \linebreak $z_{s}:=a_{s}-c_{s}$ between actions
and forecasts, instead of the actions $a_{s}$; this amounts to subtracting
the constant $x$ from all the entries in the $x$-bin, and then $e_{t}(x)$
and $v_{t}(x)$ become, respectively, the expectation and variance of the $x$%
-bin. The empirical distribution of the differences $z_{s}$ and of the bin
labels $c_{s}$ yields two ($\mathbb{R}^{m}$-valued) random variables, which
we denote by $Z$ and $U$, respectively; namely, the pair $(Z,U)$ takes the
value $(z_{s},c_{s})\equiv (a_{s}-c_{s},c_{s})$ for $s=1,...,t$ with
probability $1/t$ each. With this representation we have%
\begin{eqnarray*}
e_{t}(x) &=&\mathbb{E}\left[ Z|U=x\right] , \\
v_{t}(x) &=&\mathbb{V}ar\left[ Z|U=x\right] , \\
\mathcal{K}_{t} &=&\mathbb{E}\left[ \left\Vert \mathbb{E}\left[ Z|U\right]
\right\Vert ^{2}\right] , \\
\mathcal{R}_{t} &=&\mathbb{E}\left[ \mathbb{V}ar\left[ Z|U\right] \right]
,\;\;\text{and} \\
\mathcal{B}_{t} &=&\mathbb{E}\left[ \left\Vert Z\right\Vert ^{2}\right] =%
\mathbb{E}\left[ \mathbb{E[}\left\Vert Z\right\Vert ^{2}|U]\right] .
\end{eqnarray*}%
Using the identity $\mathbb{E}\left[ X^{2}\right] =\mathbb{V}ar\left[ X%
\right] +\left( \mathbb{E}\left[ X\right] \right) ^{2}$ for each one of the $%
m$ coordinates of $Z|U$, summing over the coordinates, and then taking
overall expectation yields%
\begin{equation}
\mathcal{B}_{t}=\mathcal{R}_{t}+\mathcal{K}_{t},  \label{eq:B=R+K}
\end{equation}%
which is a useful decomposition of the Brier score (see Sanders 1963 and
Murphy 1972). Appendix \ref{s-a:decompose} generalizes this to
\textquotedblleft fractional" binnings.

For each $x$ the variance $v_{t}(x)$ of the $x$-bin is the minimum over $%
y\in C$ of\linebreak\ $n_{t}(x)^{-1}\sum_{1\leq s\leq
t:c_{s}=x}^{{}}\left\Vert a_{s}-y\right\Vert ^{2},$ which is attained when $%
y $ equals the bin average $\bar{a}_{t}(x)$. Therefore, the refinement score
is the Brier score where each bin label $x$ is replaced by $\bar{a}_{t}(x)$,
and this is the minimal Brier score over all relabelings of the bins: 
\begin{equation}
\mathcal{R}_{t}=\min_{\phi }\mathcal{B}_{t}^{\phi (\mathbf{c})},
\label{eq:R=minB}
\end{equation}%
where the minimum is taken over all functions $\phi :C\rightarrow C$ (from
current labels $x$ to new labels $y$), and we write $\mathcal{B}_{t}^{\phi (%
\mathbf{c})}$ for the Brier score where the sequence $\mathbf{c}$ is
replaced by\footnote{%
Joining two bins that have the same average does not affect the refinement
score.} $\phi (\mathbf{c})=(\phi (c_{s}))_{1\leq s\leq t}$. Thus, starting
from the Brier scoring rule, we could define the refinement score $\mathcal{R%
}$ as the \textquotedblleft relabeling-minimum" Brier score, and the
calibration score $\mathcal{K}$ as the \textquotedblleft residual" score $%
\mathcal{B-R}$. The same holds for the logarithmic scoring rule (see
Appendix \ref{s-a:log}), and may well be used for other scoring rules.

\subsection{Calibration\label{sus:calibration}}

A stochastic \emph{forecasting procedure} $\sigma $ is a mapping $\sigma
:\cup _{t\geq 1}(A^{t-1}\times C^{t-1})\rightarrow \Delta (C);$ i.e., to
each history $(\mathbf{a}_{t-1},\mathbf{c}_{t-1})$ of actions and forecasts
before time $t$ the procedure $\sigma $ assigns a probability distribution $%
\sigma (\mathbf{a}_{t-1},\mathbf{c}_{t-1})$ on $C$, which yields the
forecast $c_{t}\in C$. When these distributions are all pure (i.e., their
support is always a single $c_{t}$ in $C$), the procedure is \emph{%
deterministic}.

Let $\varepsilon \geq 0;$ a (stochastic) procedure $\sigma $ is $\varepsilon 
$\emph{-calibrated} (Foster and Vohra 1998) if\footnote{%
The calibration score $\mathcal{K}_{t}$ depends on the actions and forecasts
up to time $t,$ and is thus a function $\mathcal{K}_{t}\equiv \mathcal{K}%
_{t}(\mathbf{a},\sigma )$ of the action sequence $\mathbf{a}$ and the
forecasting procedure $\sigma $ (in fact, only $\mathbf{a}^{t}$ and $\sigma
^{t}$ matter for $\mathcal{K}_{t}$). The same applies to the other scores
throughout the paper.}$_{\text{'}}$\footnote{%
The reason that we have $\varepsilon ^{2}$ on the right-hand side is that we
are dealing here with the square-calibration score; the same applies to
calibeating. The definition here implies the standard one that uses $K_{t}$
instead of $\mathcal{K}_{t}$ (e.g., Foster and Hart 2021), since, as we have
seen in Section \ref{sus:scores}, $(K_{t})^{2}\leq \mathcal{K}_{t}.$}%
\begin{equation*}
\varlimsup_{t\rightarrow \infty }\left( \sup_{\mathbf{a}_{t}}\mathbb{E}\left[
\mathcal{K}_{t}\right] \right) \leq \varepsilon ^{2}
\end{equation*}%
(the expectation $\mathbb{E}$ is taken over the random forecasts of $\sigma $%
).

\subsection{The Concept of \textquotedblleft Calibeating"\label%
{sus:calibeat-def}}

We come now to the central concept of this paper, \textquotedblleft
calibeating," which stands for \textquotedblleft beating by an amount equal
to the calibration score": a forecasting sequence $\mathbf{c}$
\textquotedblleft calibeats" another forecasting sequence $\mathbf{b}$ if,
fixing the action sequence, $\mathbf{c}$ beats the Brier score of $\mathbf{b}
$ by at least $\mathbf{b}$'s calibration score (i.e., $\mathcal{B}^{\mathbf{c%
}}\leq \mathcal{B}^{\mathbf{b}}-\mathcal{K}^{\mathbf{b}}$ in the long run).
Thus, if $\mathbf{b}$ is not calibrated, and hence its calibration score $%
\mathcal{K}^{\mathbf{b}}$ is positive, then the Brier score $\mathcal{B}^{%
\mathbf{c}}$ of $\mathbf{c}$ is not just better (i.e., lower) than the Brier
score $\mathcal{B}^{\mathbf{b}}$ of $\mathbf{b}$, but it is strictly better,
by at least $\mathcal{K}^{\mathbf{b}}$. The formal definition will require
calibeating to be carried out \emph{online}---i.e., to have access only to
the current forecast of $\mathbf{b\ }$(and the history) and nothing beyond
that---and also to be \emph{guaranteed}---i.e., to hold no matter what the
sequences of actions and forecasts will be; moreover, this should hold
uniformly over all these sequences.

By way of the uniformity requirement, we consider a given set $B\subseteq C$
of possible forecasts; for instance, $B$ may be a finite set. A forecasting
procedure $\sigma $ all of whose forecasts are in $B$ is called a $B$-\emph{%
forecasting procedure }(when $B=C$ we will usually just say a
\textquotedblleft forecasting procedure"). Let $\Sigma _{B}$ denote the set
of all $B$-forecasting procedures $\sigma $, i.e., all mappings $\sigma
:\cup _{t\geq 1}(A^{t-1}\times B^{t-1})\rightarrow \Delta (B)$. For $\sigma
\in \Sigma _{B}$, let $b_{t}\in B$ denote the forecast at time $t$, and put $%
\mathbf{b}_{t}=(b_{s})_{1\leq s\leq t}$ and $\mathbf{b}=(b_{s})_{s\geq 1}.$

Assume that in each period $t$ the forecast $b_{t}$ is announced \emph{before%
} we provide our forecast $c_{t}$. Thus, (the distribution of) $c_{t}$ may
depend on $(\mathbf{a}_{t-1},\mathbf{c}_{t-1},\mathbf{b}_{t})$, i.e., on the
history $h_{t-1}=(\mathbf{a}_{t-1},\mathbf{c}_{t-1},\mathbf{b}_{t-1})$
before time $t$ together with the current $b_{t}$. A $\mathbf{b}$\emph{%
-based forecasting procedure} $\zeta $ is a mapping\footnote{%
One should not confuse \textquotedblleft $B$-forecasting" with
\textquotedblleft $\mathbf{b}$-based"; the former refers to the \emph{outputs%
} of the procedure (all forecasts are in $B$) whereas the latter refers to
the \emph{inputs} of the procedure (the sequence $\mathbf{b}$)$.$}\emph{\ }$%
\zeta :\cup _{t\geq 1}(A^{t-1}\times C^{t-1}\times B^{t})\rightarrow \Delta
(C)$. We will use superscripts $\mathbf{b},\mathbf{c}$ on the scores $%
\mathcal{B},\mathcal{R},\mathcal{K}$ to denote the sequence to which they
apply, and similarly for action averages; for example, $\bar{a}_{t}^{\mathbf{%
b}}(x)$ is the average of the actions in all periods $s\leq t$ where $%
b_{s}=x $, and $\bar{a}_{t}^{\mathbf{c}}(x)$ is the average of the actions
in all periods $s\leq t$ where $c_{s}=x$.

Let $\varepsilon \geq 0$; a $\mathbf{b}$-based procedure $\zeta $ is $%
(\varepsilon ,B)$\emph{-calibeating} if its Brier score beats the Brier
score of \emph{any} $B$-forecasting procedure $\sigma $ (on which it is
based) by that procedure's calibration score; formally,%
\begin{equation}
\varlimsup_{t\rightarrow \infty }\left( \sup_{\sigma \in \Sigma _{B}}\sup_{%
\mathbf{a}_{t}\in A^{t}}\mathbb{E}\left[ \mathcal{B}_{t}^{\mathbf{c}}\mathbf{%
-}\left( \mathcal{B}_{t}^{\mathbf{b}}-\mathcal{K}_{t}^{\mathbf{b}}\right) %
\right] \right) \leq \varepsilon ^{2},  \label{eq:calibeat-def}
\end{equation}%
where the expectation $\mathbb{E}$ is over the random forecasts of $\sigma $
and $\zeta $; when $\varepsilon =0$ we call this $B$-\emph{calibeating}.
Thus, calibeating is guaranteed for any sequence $\mathbf{a}$ of actions and
any sequence $\mathbf{b}$ of resulting forecasts of $\sigma $, \emph{%
uniformly }over all $B$-forecasting procedures $\sigma $ and action
sequences $\mathbf{a}.$

Clearly, condition (\ref{eq:calibeat-def}) is not affected if one allows the
sequences $\mathbf{a}_{t}$ to be random. Moreover, since \emph{all }%
sequences $\mathbf{a}_{t}$ are considered, one may envision an
\textquotedblleft adversary" that chooses the $B$-forecasting procedure $%
\sigma $ as well as the action sequence $\mathbf{a}_{t}$, and so the
sequences $\mathbf{b}_{t}$ and $\mathbf{a}_{t}$ may well be
\textquotedblleft coordinated." Thus, $\sup_{\sigma }\sup_{\mathbf{a}_{t}}$
in (\ref{eq:calibeat-def}) is the same as $\sup_{\mathbf{a}_{t},\mathbf{b}%
_{t}}$, where $\mathbf{b}_{t}$ ranges over $B^{t};$ indeed, the latter
supremum can only be larger, as all sequences $\mathbf{b}_{t}$ are
considered there and not just those generated by $\sigma $; however, it
cannot be strictly larger since all $\sigma $ that forecast a fixed sequence 
$\mathbf{b}_{t}$ (ignoring the history) are included in the former supremum.
Thus, a $\mathbf{b}$-based procedure $\zeta $ is $(\varepsilon ,B)$%
-calibeating if%
\begin{equation}
\varlimsup_{t\rightarrow \infty }\left( \sup_{\mathbf{a}_{t}\in A^{t},%
\mathbf{b}_{t}\in B^{t}}\mathbb{E}\left[ \mathcal{B}_{t}^{\mathbf{c}}\mathbf{%
-}\left( \mathcal{B}_{t}^{\mathbf{b}}-\mathcal{K}_{t}^{\mathbf{b}}\right) %
\right] \right) \leq \varepsilon ^{2},  \label{eq:calibeat-def1}
\end{equation}%
where the expectation is now over the randomizations of $\zeta .$

\subsubsection{Calibeating for General $B$\label{sus:general B}}

Since $\mathcal{B-K=R}$, we can replace $\mathcal{B}_{t}^{\mathbf{b}}-%
\mathcal{K}_{t}^{\mathbf{b}}$ in (\ref{eq:calibeat-def}) and (\ref%
{eq:calibeat-def1}) with the refinement score $\mathcal{R}_{t}^{\mathbf{b}}$
of $\mathbf{b}$: calibeating means that $\mathbf{c}$'s Brier score beats $%
\mathbf{b}$'s refinement score. This allows the notion of calibeating to be
generalized to sequences $\mathbf{b}=(b_{t})_{t\geq 1}$ for which $b_{t}$
need not be an element of $C$. The \textquotedblleft forecast" may thus be
\textquotedblleft a nice day," a \textquotedblleft red day," a
\textquotedblleft $b$-day," or just \textquotedblleft $b$," for some $b$ in
an arbitrary set $B$. What matters for the resulting refinement scores $%
\mathcal{R}_{t}^{\mathbf{b}}$ are the bins into which the days are
classified and the ensuing bin variances; the specific labels $b$ of the
bins do not matter (the labels \emph{do} however matter for the calibration
score, which is \textquotedblleft orthogonal" to the refinement score).
Therefore, we extend our definition to arbitrary sets $B$: a $\mathbf{b}$%
-based forecasting procedure is $(\varepsilon ,B)$\emph{-calibeating} if%
\begin{equation}
\varlimsup_{t\rightarrow \infty }\left( \sup_{\sigma \in \Sigma _{B}}\sup_{%
\mathbf{a}_{t}\in A^{t}}\mathbb{E}\left[ \mathcal{B}_{t}^{\mathbf{c}}\mathbf{%
-}\mathcal{R}_{t}^{\mathbf{b}}\right] \right) \leq \varepsilon ^{2}
\label{eq:calibeat-def-R}
\end{equation}%
or, equivalently,%
\begin{equation}
\varlimsup_{t\rightarrow \infty }\left( \sup_{\mathbf{a}_{t}\in A^{t},%
\mathbf{b}_{t}\in B^{t}}\mathbb{E}\left[ \mathcal{B}_{t}^{\mathbf{c}}\mathbf{%
-}\mathcal{R}_{t}^{\mathbf{b}}\right] \right) \leq \varepsilon ^{2}.
\label{eq:calibeat-def-R1}
\end{equation}

As we will see below, this natural extension will be useful, for instance,
when considering the joint binning generated by several forecasting
procedures.

Finally, calibeating can be formalized in terms of Brier scores only. Since
the refinement score is the minimal Brier score over all relabelings of the
bins (see (\ref{eq:R=minB})), it follows that calibeating amounts to getting
the Brier score of $\mathbf{c}$ down to the \textquotedblleft
relabeling-minimum" Brier score of $\mathbf{b}$, i.e. (ignoring $\varepsilon 
$, $\varlimsup $, and $\sup $), $\mathcal{B}_{t}^{\mathbf{c}}\leq \min_{\phi
}\mathcal{B}_{t}^{\phi (\mathbf{b})}$, where the minimum is taken over all
functions $\phi :B\rightarrow \Delta (A)$ (the minimum is attained when $%
\phi (b)$ equals the average $\bar{a}_{t}^{\mathbf{b}}(b)$ of the $b$-bin;
cf. the correction of forecasts \textquotedblleft in retrospect" in the
Introduction).

\section{The Online Refinement Score\label{s:online R}}

The main tool that we will use is that $\mathcal{R}_{t}$, the refinement
score at time $t$, which is the average variance of the bins and can thus be
computed only at time $t$ when the averages of all bins are known (i.e., 
\emph{offline}), can be approximated by a similar score $\widetilde{\mathcal{%
R}}_{t}$, which is computed period by period (i.e., \emph{online}).

Specifically, we define the \emph{online refinement score }$\widetilde{%
\mathcal{R}}_{t}$ \emph{at time }$t$ by 
\begin{equation*}
\widetilde{\mathcal{R}}_{t}%
{\;:=\;}%
\frac{1}{t}\sum_{s=1}^{t}\left\Vert a_{s}-\bar{a}_{s-1}(c_{s})\right\Vert
^{2},
\end{equation*}%
where for each $c$ in $C$ we take $\bar{a}_{0}(c)$ to be an arbitrary
element of $C$. Comparing this with the refinement score $\mathcal{R}%
_{t}=(1/t)\sum_{s=1}^{t}\left\Vert a_{s}-\bar{a}_{t}(c_{s})\right\Vert ^{2}$%
, we see that what $\widetilde{\mathcal{R}}_{t}$ does is to replace for each 
$s=1,...,t$ the term $\bar{a}_{t}(c_{s})$, the average at time $t$ of the $%
c_{s}$-bin to which $a_{s}$ is assigned\footnote{%
As pointed out in Section \ref{sus:scores}, neither $\mathcal{R}_{t}$ nor $%
\widetilde{\mathcal{R}}_{t}$ is affected whether we assign to the $c_{s}$%
-bin the action $a_{s}$ or the difference $z_{s}=a_{s}-c_{s}$.} (an average
that will be determined only at time $t$, i.e., offline), by the term $\bar{a%
}_{s-1}(c_{s})$, the past average (i.e., before time $s)$ of that same $%
c_{s} $-bin, which \emph{is} known at time $s$ (i.e., online).

The following proposition bounds the difference between $\widetilde{\mathcal{%
R}}_{t}$ and $\mathcal{R}_{t}$.

\begin{proposition}
\label{p:online-R}For any $t\geq 1$ and any sequences $\mathbf{a}_{t}$ and $%
\mathbf{c}_{t}$ we have 
\begin{equation}
\mathcal{R}_{t}\leq \widetilde{\mathcal{R}}_{t}\leq \mathcal{R}_{t}+\gamma
^{2}\frac{N_{t}}{t}\left( \ln \left( \frac{t}{N_{t}}\right) +1\right) ,
\label{eq:Rtilde-R}
\end{equation}%
where $N_{t}:=\left\vert \{c_{s}:1\leq s\leq t\}\right\vert $ is the number
of distinct elements in the sequence $\mathbf{c}_{t}=(c_{1},...,c_{t})$
(i.e., the number of distinct forecasts used).
\end{proposition}

Thus, $\widetilde{\mathcal{R}}_{t}-\mathcal{R}_{t}\rightarrow 0$ as $%
t\rightarrow \infty $ when $N_{t}/t\rightarrow 0$, i.e., the number of
forecasts used up to time $t$ increases at a slower rate than\footnote{%
For a simple example where $N_{t}/t$ does not converge to $0$ and the online
refinement score $\widetilde{\mathcal{R}}_{t}$ does not approach the
refinement score $\mathcal{R}_{t},$ take%
\begin{equation*}
\begin{array}{ccccccccccccc}
\mathbf{a:} & 0 & 1 & 0 & 1 & 0 & 1 & 0 & 1 & ... & 0 & 1 & ... \\ 
\mathbf{c:} & 1 & 1 & \frac{1}{2} & \frac{1}{2} & \frac{1}{3} & \frac{1}{3}
& \frac{1}{4} & \frac{1}{4} & ... & \frac{1}{n} & \frac{1}{n} & ...%
\end{array}%
.
\end{equation*}%
Indeed, for all even periods $t=2n$ (where $N_{t}=n,$ and so $%
N_{t}/t\rightarrow 1/2),$ we have $\mathcal{R}_{t}=1/4$ (since each $(1/i)$%
-bin contains two elements, $a_{2i-1}=0$ and $a_{2i}=1)$ and $\widetilde{%
\mathcal{R}}_{t}\geq 1/2$ (since $(a_{2i}-\bar{a}%
_{2i-1}(c_{2i}))^{2}=(1-0)^{2}=1$ and $(a_{2i-1}-\bar{a}%
_{2i-2}(c_{2i-1}))^{2}\geq 0$).} $t$. When forecasts belong to a \emph{finite%
} set $D\subset C$, and so $N_{t}\leq |D|$ and $\ln (t/N_{t})\leq \ln t$ for
all $t$, we get 
\begin{equation}
0\leq \widetilde{\mathcal{R}}_{t}-\mathcal{R}_{t}\leq |D|\frac{\ln t+1}{t}.
\label{eq:finite D}
\end{equation}

Proposition \ref{p:online-R} follows from the following online formula for
the variance. Let $(x_{n})_{n\geq 1}$ be a sequence of vectors in a
Euclidean space (or, more generally, in a normed vector space).

\begin{proposition}
\label{p:var}For every $n\geq 1$ we have 
\begin{equation}
\sum_{i=1}^{n}\left\Vert x_{i}-\bar{x}_{n}\right\Vert
^{2}=\sum_{i=1}^{n}\left( 1-\frac{1}{i}\right) \left\Vert x_{i}-\bar{x}%
_{i-1}\right\Vert ^{2},  \label{eq:p-var}
\end{equation}%
where $\bar{x}_{m}:=(1/m)\sum_{i=1}^{m}x_{i}$ denotes the average of%
\footnote{%
The sum on the right-hand side of (\ref{eq:p-var}) effectively starts from $%
i=2,$ and so it does not matter how $\bar{x}_{0}$ is defined.} $%
x_{1},...,x_{m}.$
\end{proposition}

\begin{proof}
Put $s_{n}:=\sum_{i=1}^{n}\left\Vert x_{i}-\bar{x}_{n}\right\Vert ^{2}$; we
claim that%
\begin{equation}
s_{n}=s_{n-1}+\left( 1-\frac{1}{n}\right) \left\Vert x_{n}-\bar{x}%
_{n-1}\right\Vert ^{2}.  \label{eq:diff-var}
\end{equation}%
We provide a short proof:\footnote{%
An alternative proof of (\ref{eq:diff-var}) uses $\mathbb{V}ar(X)=\mathbb{E}%
\left[ \mathbb{V}ar(X|Y)\right] +\mathbb{V}ar(\mathbb{E}\left[ X|Y\right] )$%
, where $X=x_{i}$ with probability $1/n$ and $Y$ is the indicator that $i=n.$
Formula (\ref{eq:diff-var}) is known as a \textquotedblleft variance update"
formula; see, e.g., Welford (1962).} let $n\geq 2$ (when $n=1$ both sides
vanish), and assume that $\bar{x}_{n-1}=0$ (this is without loss of
generality, since subtracting a constant from all the $x_{i}$ does not
affect any of the terms); then $\bar{x}_{n}=(1/n)x_{n}$, and so, using $%
s_{n}=\sum_{i=1}^{n}||x_{i}||^{2}-n||\bar{x}_{n}||^{2}$, we get%
\begin{equation*}
s_{n}-s_{n-1}=\left( \sum_{i=1}^{n}\left\Vert x_{i}\right\Vert
^{2}-n\left\Vert \frac{1}{n}x_{n}\right\Vert ^{2}\right)
-\sum_{i=1}^{n-1}\left\Vert x_{i}\right\Vert ^{2}=\left\Vert
x_{n}\right\Vert ^{2}-\frac{1}{n}\left\Vert x_{n}\right\Vert ^{2},
\end{equation*}%
which is $(1-1/n)\left\Vert x_{n}\right\Vert ^{2}=(1-1/n)\left\Vert x_{n}-%
\bar{x}_{n-1}\right\Vert ^{2}$.

Applying (\ref{eq:diff-var}) recursively yields the result.
\end{proof}

\bigskip

Let $v_{n}:=(1/n)\sum_{i=1}^{n}\left\Vert x_{i}-\bar{x}_{n}\right\Vert ^{2}$
denote the variance of $x_{1},...,x_{n}$, and put $\widetilde{v}%
_{n}:=(1/n)\sum_{i=1}^{n}\left\Vert x_{i}-\bar{x}_{i-1}\right\Vert ^{2};$
i.e., $\bar{x}_{n}$, the final (up to $n$) average, is replaced for each $%
i=1,...,n$ with $\bar{x}_{i-1}$, the previous (up to $i-1)$ average (take $%
\bar{x}_{0}$ to be an arbitrary element of the convex hull of the $x_{i}$).
We refer to $\widetilde{v}_{n}$ as the \emph{online variance} of $%
x_{1},...,x_{n}$. Proposition \ref{p:var} gives $\widetilde{v}%
_{n}-v_{n}=(1/n)\sum_{i=1}^{n}(1/i)\left\Vert x_{i}-\bar{x}_{i-1}\right\Vert
^{2}$, and so%
\begin{equation}
0\leq \widetilde{v}_{n}-v_{n}\leq \frac{1}{n}\sum_{i=1}^{n}\frac{1}{i}\xi
^{2}\leq \xi ^{2}\frac{\ln n+1}{n},  \label{eq:v-tilda}
\end{equation}%
where $\xi :=\max_{1\leq i,j\leq n}\left\Vert x_{i}-x_{j}\right\Vert $;
moreover, the\footnote{%
We use standard asymptotic notation as $n\rightarrow \infty $: $%
f(n)=O(g(n)),\;f(n)=o(g(n)),\;$and $f(n)\sim g(n)$ stand for, respectively, $%
\overline{\lim }_{n\rightarrow \infty }f(n)/g(n)<\infty ,$ $%
\lim_{n\rightarrow \infty }f(n)/g(n)=0,$ and $\lim_{n\rightarrow \infty
}f(n)/g(n)=1.$} $O(\log n/n)$ bound is tight (take each $x_{i}$ to be at a
distance of at least some $\delta >0$ from $\bar{x}_{i-1}$), and thus so is
the $O(\log t/t)$ bound in Proposition \ref{p:online-R} and (\ref{eq:finite
D}).

Proposition \ref{p:online-R} now easily follows.

\bigskip

\begin{proof}[Proof of Proposition \protect\ref{p:online-R}]
Let $D\equiv D_{t}:=\{c_{s}:1\leq s\leq t\}\subset C$ be the set of
forecasts used up to time $t$, i.e., the set of nonempty bins. For each $%
d\in D$ we apply (\ref{eq:v-tilda}) to get%
\begin{equation*}
0\leq \frac{1}{n_{t}(d)}\sum_{s\leq t:c_{s}=d}^{{}}\left\Vert a_{s}-\bar{a}%
_{s-1}(d)\right\Vert ^{2}-\frac{1}{n_{t}(d)}\sum_{s\leq
t:c_{s}=d}^{{}}\left\Vert a_{s}-\bar{a}_{t}(d)\right\Vert ^{2}\leq \gamma
^{2}\frac{\ln n_{t}(d)+1}{n_{t}(d)}.
\end{equation*}%
Averaging over $d$ in $D_{t}$ with the weights $n_{t}(d)/t$ then yields%
\begin{equation*}
0\leq \widetilde{\mathcal{R}}_{t}-\mathcal{R}_{t}\leq \gamma ^{2}\frac{1}{t}%
\sum_{d\in D_{t}}(\ln n_{t}(d)+1).
\end{equation*}%
Since the function $\ln $ is concave and $\sum_{d\in D_{t}}n_{t}(d)=t$, the
sum $\sum_{d\in D_{t}}\ln n_{t}(d)$ is maximal when all the $n_{t}(d)$ are
equal, i.e., when $n_{t}(d)=t/N_{t}$ for each $d\in D_{t};$ this yields the
result (\ref{eq:Rtilde-R}).
\end{proof}

\section{A Simple Way to Calibeat\label{s:simple}}

We provide a simple calibeating procedure. The set $B$ is taken for now to
be finite (the restriction on the number of possible forecasts, i.e., on the
number of bins, is needed in order for the resulting classification to be
meaningful; in the extreme case where all forecasts are distinct, and thus
each bin contains a single element, we have $\mathcal{R}_{t}=0$ for all $t$%
). This finiteness assumption may be relaxed; see Remark (d) below.

\begin{theorem}
\label{th:beat-1}Let $B$ be a finite set, and let $\zeta $ be the
deterministic $\mathbf{b}$-based forecasting procedure given by 
\begin{equation}
c_{t}=\bar{a}_{t-1}^{\mathbf{b}}(b_{t})  \label{eq:c=a-bar}
\end{equation}%
for every time $t\geq 1$ (if $t$ is the first time that $b_{t}$ is used,
take $c_{t}$ to be an arbitrary element of $C$). Then $\zeta $ is $B$%
-calibeating; specifically,%
\begin{equation}
0\leq \mathcal{B}_{t}^{\mathbf{c}}-\mathcal{R}_{t}^{\mathbf{b}}\leq \gamma
^{2}|B|\frac{\ln t+1}{t}  \label{eq:r-b}
\end{equation}%
for all $t\geq 1$ and all sequences $\mathbf{a}_{t}\in A^{t}$ and $\mathbf{b}%
_{t}\in B^{t}.$
\end{theorem}

\begin{proof}
Our choice of $c_{t}=\bar{a}_{t-1}^{\mathbf{b}}(b_{t})$ makes $\mathcal{B}%
_{t}^{\mathbf{c}}=\widetilde{\mathcal{R}}_{t}^{\mathbf{b}}$ for any $\mathbf{%
a}_{t}$ and $\mathbf{b}_{t};$ use Proposition \ref{p:online-R} (see (\ref%
{eq:finite D})).\footnote{%
One always has $\widetilde{\mathcal{R}}_{t}^{\mathbf{b}}=\mathcal{B}_{t}^{%
\mathbf{\bar{a}}(\mathbf{b})}$; that is, the online refinement score $%
\widetilde{\mathcal{R}}_{t}^{\mathbf{b}}$ of the sequence $\mathbf{b}%
=(b_{s})_{s\geq 1}$ is the same as the Brier score of the sequence of action
averages $\mathbf{\bar{a}}(\mathbf{b})=(\bar{a}_{s-1}(b_{s}))_{s\geq 1}.$}
\end{proof}

\bigskip

The calibeating forecast $c_{t}$ is thus the average of the actions in those
periods $1\leq s\leq t-1$ in which the forecast $b_{s}$ was equal to the
current forecast $b_{t}$. When $B$ is a subset of $C$ we get by (\ref{eq:r-b}%
) that $\mathcal{B}_{t}^{\mathbf{c}}\leq \mathcal{B}_{t}^{\mathbf{b}}-%
\mathcal{K}_{t}^{\mathbf{b}}+o(1)$; i.e., the Brier score of $\mathbf{c}$ is
lower than that of $\mathbf{b}$ by essentially the calibration score of $%
\mathbf{b}$. Note that the specific values $b_{t}$ of the $B$-forecasts are
not used by the calibeating procedure $\zeta $, and only the binning that
they generate matters (see Section \ref{sus:general B}).

\bigskip

\noindent \textbf{Remarks.} \emph{(a)} The simple calibeating procedure $%
\zeta $\textrm{\ }is \textquotedblleft universal" also in the sense of being
independent of the specific set $B$: the forecast $c_{t}$ is just the past
average of the current bin.

\emph{(b) }The history of one's own forecasts, $\mathbf{c}_{t-1}$, is not
used by the procedure $\zeta $; thus, $c_{t}$ is a function of $\mathbf{a}%
_{t-1}$ and $\mathbf{b}_{t}$ only.

\emph{(c) }One cannot guarantee a Brier score that is lower than the
refinement score of $\mathbf{b}$. Indeed, for every $t\geq 1$ and every $%
\mathbf{b}_{t}\in B^{t}$, we have%
\begin{equation*}
\sup_{\mathbf{a}_{t}}\mathbb{E}\left[ \mathcal{B}_{t}^{\mathbf{c}}-\mathcal{R%
}_{t}^{\mathbf{b}}\right] \geq 0
\end{equation*}%
for \emph{any} sequence $\mathbf{c}_{t}$, because when all $a_{s}$ are equal
to a fixed $a^{0}\in A$ we get $\mathcal{R}_{t}^{\mathbf{b}}=0$ (because all
bins contain only $a^{0}$, and so their variance is zero).

\emph{(d) }If the set $B$ is not finite, the procedure $\zeta $ calibeats
also all sequences$\mathbf{\ b}$ with $N_{t}^{\mathbf{b}}/t\rightarrow 0$ as 
$t\rightarrow \infty $, where $N_{t}^{\mathbf{b}}:=\left\vert \{b_{s}:s\leq
t\}\right\vert $ is the number of distinct forecasts used by $\mathbf{b}$ up
to time $t$ (use Proposition \ref{p:online-R}).

\emph{(e)} From any forecasting procedure, whose forecasts may not be
calibrated, we can generate by Theorem \ref{th:beat-1} another forecasting
procedure that yields lower Brier scores in the long run, as follows. Let
the $\mathbf{b}$-forecasts be generated by a forecasting procedure $\sigma $%
, and let $\sigma ^{\prime }$ replace each $b_{t}$ by the corresponding $%
\bar{a}_{t-1}^{\mathbf{b}}(b_{t})$ (see Remark (b) in Appendix \ref%
{susus-a:simple-additional} for some technical details); then $\sigma
^{\prime }$ yields lower Brier scores than $\sigma $ in the long run: $%
\mathcal{B}_{t}^{\mathbf{c}}\leq \mathcal{B}_{t}^{\mathbf{b}}-\mathcal{K}%
_{t}^{\mathbf{b}}+o(1).$

\emph{(f)} The existence of a calibeating procedure may be proved by a
minimax argument, which extends the 1995 proof of Hart of calibration (see
Section 4 of Foster and Vohra 1998, and Hart 2021); we do so in Appendix \ref%
{s-a:minimax}. The existence proof does not however provide an explicit
calibeating procedure, for sure not the very simple one of Theorem \ref%
{th:beat-1}.

\bigskip

Additional comments are relegated to Appendix \ref{sus-a:simple}. In
particular, we show that one cannot guarantee a calibeating error of an
order of magnitude lower than $\log t/t$ (see Appendix \ref{susus-a:lower
bound}), and that the best that one can do is to decrease the error in
Theorem \ref{th:beat-1} by a factor between $2$ and $4$ (depending on the
dimension $m$), by using a more complex formula for the forecast $c_{t}$
instead of (\ref{eq:c=a-bar}) (see Appendix \ref{susus-a:constant}).

\section{Self-calibeating $=$ Calibrating\label{s:self}}

The construction of Section \ref{s:simple} may be leveraged to obtain
calibration. Indeed, when $\mathbf{b}=\mathbf{c}$ we have $\mathcal{B}_{t}^{%
\mathbf{c}}-\mathcal{R}_{t}^{\mathbf{b}}=\mathcal{B}_{t}^{\mathbf{c}}-%
\mathcal{R}_{t}^{\mathbf{c}}=\mathcal{K}_{t}^{\mathbf{c}}$, and so
\textquotedblleft self-calibeating," i.e., $\mathbf{c}$ calibeating $\mathbf{%
c}$, is equivalent to calibration, i.e., $\mathcal{K}_{t}^{\mathbf{c}%
}\rightarrow 0$. To achieve this by the construct of Theorem \ref{th:beat-1}
we would need to choose $c_{t}$ so that $c_{t}=\bar{a}_{t-1}^{\mathbf{c}%
}(c_{t})$. However, this requires a fixed point of the function $\bar{a}%
_{t-1}^{\mathbf{c}}(\cdot )$, which of course need not exist in general. We
circumvent this by using a \textquotedblleft stochastic expected fixed
point" result, i.e., by appealing to the corresponding \textquotedblleft
outgoing" theorems of Foster and Hart (2021)---see Appendix \ref%
{s-a:outgoing} for details---and thereby obtain the classic calibration
results (see Theorem 11(S) and (AD) in Foster and Hart 2021).\footnote{%
While the proof here may look different from the one in Foster and Hart
(2021), the two proofs are in fact identical. The approach here with the
online refinement score makes the proof more transparent.}

\begin{theorem}
\label{th:calibration}Let $\delta >0$ and let $D\subset C$ be a finite $%
\delta $-grid of $C$. Then there exists a stochastic $D$-forecasting
procedure $\sigma $ that is $\delta $-calibrated; specifically,\footnote{%
Since we are dealing here with only one forecasting sequence $\mathbf{c,}$
we will drop the superscript $\mathbf{c}$ from $\mathcal{K}$ and $\bar{a}.$} 
\begin{equation*}
\mathbb{E}\left[ \mathcal{K}_{t}\right] \leq \delta ^{2}+\gamma ^{2}|D|\frac{%
\ln t+1}{t}
\end{equation*}%
for all $t\geq 1$ and all sequences $\mathbf{a}_{t}\in A^{t}$. Moreover, $%
\sigma $ may be taken to be $\delta $-almost deterministic (i.e., all
randomizations are $\delta $-local).
\end{theorem}

\begin{proof}
For every $t$ and history $h_{t-1}=(\mathbf{a}_{t-1},\mathbf{c}_{t-1})$, the
outgoing Theorem \ref{th:outgoing} (S) of Appendix \ref{s-a:outgoing}
applied to the function $\bar{a}_{t-1}(\cdot )$ yields a distribution $\eta
_{t}$ on $D$ such that, using it as the distribution $\sigma (h_{t-1})$ of
the forecast $c_{t}$, we have 
\begin{equation}
\mathbb{E}_{t-1}\left[ \left\Vert a_{t}-c_{t}\right\Vert ^{2}-\left\Vert
a_{t}-\bar{a}_{t-1}(c_{t})\right\Vert ^{2}\right] \leq \delta ^{2}
\label{eq:Ediff}
\end{equation}%
for every $a_{t}\in A$, where $\mathbb{E}_{t-1}$ denotes expectation with
respect to $\sigma (h_{t-1})$. Taking overall expectation and averaging over 
$t=1,2,..$. yields $\mathbb{E}\left[ \mathcal{B}_{t}-\widetilde{\mathcal{R}}%
_{t}\right] \leq \delta ^{2};$ Proposition \ref{p:online-R} completes the
proof. For the \textquotedblleft moreover" part, use part (AD) of Theorem %
\ref{th:outgoing}.
\end{proof}

\bigskip

The proof is quite instructive: what we would like to get is $\lambda
_{t}:=\left\Vert a_{t}-c_{t}\right\Vert ^{2}-\left\Vert a_{t}-\bar{a}%
_{t-1}(c_{t})\right\Vert ^{2}\leq 0$ no matter what $a_{t}$ will be, which
can be guaranteed only by choosing $c_{t}=\bar{a}_{t-1}(c_{t})$. This means
that $c_{t}$ should be a fixed point of the function $\bar{a}_{t-1}(\cdot )$%
, a function that is defined only on the finite set $\{c_{s}:1\leq s\leq t\}$
and is far from being continuous, and so need not in general have a fixed
point. We thus use a distribution $\eta _{t}$ instead---obtained by the
minimax theorem---that guarantees that, in expectation, $\lambda _{t}$
cannot exceed $0$ by much (as in the simple illustration in Section 1.2 in
Foster and Hart 2021).

\bigskip

\noindent \textbf{Remarks. }\emph{(a) }From inequality (\ref{eq:Ediff}) for
every history we get, by the Strong Law of Large Numbers for Dependent
Random Variables (Lo\'{e}ve 1978, Theorem 32.1.E), that $\overline{\lim }%
_{t\rightarrow \infty }\left( \mathcal{B}_{t}-\widetilde{\mathcal{R}}%
_{t}\right) \leq \delta ^{2}$ (a.s.), and thus $\overline{\lim }%
_{t\rightarrow \infty }\mathcal{K}_{t}\leq \delta ^{2}$ (a.s.); see Appendix
A5 in Foster and Hart (2021).

\emph{(b)} Let $D_{t}$ be an increasing sequence (i.e., $D_{t}\subseteq
D_{t+1})$ of $\delta _{t}$-grids of $C$ such that $\delta _{t}\rightarrow 0$
and $|D_{t}|/t\rightarrow 0$ as $t\rightarrow \infty ;$ using $D_{t}$ at
time $t$ guarantees that $\mathbb{E}\left[ \mathcal{K}_{t}\right] =\mathbb{E}%
\left[ \mathcal{B}_{t}-\mathcal{R}_{t}\right] =\mathbb{E}\left[ \mathcal{B}%
_{t}-\widetilde{\mathcal{R}}_{t}\right] +\mathbb{E}\left[ \widetilde{%
\mathcal{R}}_{t}-\mathcal{R}_{t}\right] \leq \delta _{t}^{2}+O\left(
(|D_{t}|/t)\ln \left( t/|D_{t}|\right) \right) \rightarrow 0$ (by
Proposition \ref{p:online-R}), and thus we obtain $0$-calibration.

\section{Calibeating by a Calibrated Forecast\label{s:beat-by-calibrated}}

While the procedure $\zeta $ of Section \ref{s:simple} calibeats any $B$%
-forecasting procedure, $\zeta $ itself need not yield calibrated forecasts
(for example, if all its forecasts $c_{t}=\bar{a}_{t-1}^{\mathbf{b}}(b_{t})$
are distinct, then all its bins are singletons and its calibration score is
high), and so $\zeta $ itself may be calibeaten by yet another procedure.
This suggests requiring our calibeating procedure to be calibrated, which is
what we provide in this section.

Given two sequences $\mathbf{b}^{1}=(b_{t}^{1})_{t\geq 1}$ and $\mathbf{b}%
^{2}=(b_{t}^{2})_{t\geq 1}$ with values in sets $B^{1}$ and $B^{2}$,
respectively, the resulting \emph{joint binning} has $U=B^{1}\times B^{2}$
as the set of bins; i.e., there is a $(b^{1},b^{2})$-bin for each pair $%
(b^{1},b^{2})\in B^{1}\times B^{2}=U$, and $a_{t}$ is assigned to the $u_{t}$%
-bin where $u_{t}=(b_{t}^{1},b_{t}^{2})$. The bin averages are%
\begin{equation*}
\bar{a}_{t}^{\mathbf{u}}(u)\equiv \bar{a}_{t}^{\mathbf{b}^{1},\mathbf{b}%
^{2}}(b^{1},b^{2}):=\frac{\sum_{1\leq s\leq t:u_{s}=x}a_{s}}{\left\vert
\{1\leq s\leq t:u_{s}=u\}\right\vert }
\end{equation*}%
for every $u\in U$, and the refinement score is $\mathcal{R}_{t}^{\mathbf{u}%
}\equiv \mathcal{R}_{t}^{\mathbf{b}^{1}\mathbf{,b}^{2}}=(1/t)\sum_{s=1}^{t}%
\left\Vert a_{s}-\bar{a}_{t}^{\mathbf{b}^{1},\mathbf{b}%
^{2}}(b_{t}^{1},b_{t}^{2})\right\Vert ^{2}\equiv
(1/t)\sum_{s=1}^{t}\left\Vert a_{s}-\bar{a}_{t}^{\mathbf{u}%
}(v_{t})\right\Vert ^{2}$. Since $\mathcal{R}_{t}$ is the average internal
variance of the bins, refining a binning---i.e., splitting bins into several
new bins---can only decrease the refinement score; see Appendix \ref%
{s-a:refined} for a formal proof (informally, consider splitting a bin $b$
with average $\bar{x}$ into two new bins $b^{\prime }$ and $b^{\prime \prime
}$, with averages $\bar{x}^{\prime }$ and $\bar{x}^{\prime \prime }$,
respectively; writing $\sum^{\prime }$ and $\sum^{\prime \prime }$ for the
sums over $b^{\prime }$ and $b^{\prime \prime }$, respectively, we have $%
\sum^{\prime }(x_{j}-\bar{x}^{\prime })^{2}\leq \sum^{\prime }(x_{j}-\bar{x}%
)^{2}$ [this holds for any $y$ in place of $\bar{x}$], and similarly for $%
\sum^{\prime \prime }$, which added together yields $\sum^{\prime }(x_{j}-%
\bar{x}^{\prime })^{2}+\sum^{\prime \prime }(x_{j}-\bar{x}^{\prime \prime
})^{2}\leq \sum (x_{j}-\bar{x})^{2}$). Therefore%
\begin{equation}
\mathcal{R}_{t}^{\mathbf{b}^{1}\mathbf{,b}^{2}}\leq \mathcal{R}_{t}^{\mathbf{%
b}^{1}}\text{\ \ and\ \ }\mathcal{R}_{t}^{\mathbf{b}^{1}\mathbf{,b}^{2}}\leq 
\mathcal{R}_{t}^{\mathbf{b}^{2}}.  \label{eq:join-R}
\end{equation}

By using the joint binning of the given sequence $\mathbf{b}$ together with
our forecast $\mathbf{c}$, and appealing to the stochastic outgoing result,
we obtain:

\begin{theorem}
\label{th:beat-by-calib}Let $B$ be a finite set, and let $D\subset C$ be a
finite $\delta $-grid of $C$ for some $\delta >0$. Then there exists a
stochastic $\mathbf{b}$-based $D$-forecasting procedure $\zeta $ that is $%
(\delta ,B)$-calibeating and $\delta $-calibrated; specifically,%
\begin{equation*}
\mathbb{E}\left[ \mathcal{B}_{t}^{\mathbf{c}}-\mathcal{R}_{t}^{\mathbf{b,c}}%
\right] \leq \delta ^{2}+\gamma ^{2}|B|\,|D|\frac{\ln t+1}{t},
\end{equation*}%
and thus, by (\ref{eq:join-R}),%
\begin{eqnarray*}
\mathbb{E}\left[ \mathcal{B}_{t}^{\mathbf{c}}-\mathcal{R}_{t}^{\mathbf{b}}%
\right] &\leq &\delta ^{2}+\gamma ^{2}|B|\,|D|\frac{\ln t+1}{t}\text{\ \ and}
\\
\mathbb{E}\left[ \mathcal{K}_{t}^{\mathbf{c}}\right] &\leq &\delta
^{2}+\gamma ^{2}|B|\,|D|\frac{\ln t+1}{t}
\end{eqnarray*}%
for all $t\geq 1$ and all sequences $\mathbf{a}_{t}\in A^{t}$ and $\mathbf{b}%
_{t}\in B^{t}$. Moreover, $\zeta $ may be taken to be $\delta $-almost
deterministic.
\end{theorem}

Thus, if we ignore the $\delta ^{2}$ term, in the long run the refinement
score of $\zeta $ is no worse than that of any $B$-forecasting procedure,
and its calibration score is zero. When $|B|=1$ (and thus $B$-forecasting
has no content), it reduces to the calibration result, Theorem \ref%
{th:calibration}, of Section \ref{s:self}.

\begin{proof}
At time $t$, given the history $(\mathbf{a}_{t-1},\mathbf{c}_{t-1},\mathbf{b}%
_{t-1})$ together with $b_{t}$, apply Theorem \ref{th:outgoing} (S),
respectively (AD), to the function $c\longmapsto \bar{a}_{t-1}^{\mathbf{b,c}%
}(b_{t},c)$ (for $c\in D$) to get $\eta _{t}\in \Delta (D)$ such that, by
using it as the distribution of $c_{t}$ given $(\mathbf{a}_{t-1},\mathbf{c}%
_{t-1},\mathbf{b}_{t})$ (which makes it a $\mathbf{b}$-based procedure), we
have 
\begin{equation*}
\mathbb{E}_{t-1}\left[ \left\Vert a_{t}-c_{t}\right\Vert ^{2}-\left\Vert
a_{t}-\bar{a}_{t-1}^{\mathbf{b,c}}(c_{t},b_{t})\right\Vert ^{2}\right] \leq
\delta ^{2}
\end{equation*}%
for every $a_{t}\in A$, where $\mathbb{E}_{t-1}$ denotes the expectation
conditional on $(\mathbf{a}_{t-1},\mathbf{c}_{t-1},\mathbf{b}_{t})$. Taking
overall expectation and averaging over $t$ yields 
\begin{equation*}
\mathbb{E}\left[ \mathcal{B}_{t}-\widetilde{\mathcal{R}}_{t}^{\mathbf{b,c}}%
\right] \leq \delta ^{2}.
\end{equation*}%
Proposition \ref{p:online-R} completes the proof.
\end{proof}

\bigskip

\noindent \textbf{Remarks. }\emph{(a) }Again, we can allow $N_{t}^{\mathbf{b}%
}$, the number of distinct forecasts used up to time $t$, to increase with $%
t $, provided that $N_{t}^{\mathbf{b}}/t\rightarrow 0;$ see Remark (c) in
Section \ref{s:simple} and Remark (b) in Section \ref{s:self}.

\emph{(b)} The procedure in the above proof amounts to using a stochastic
forecast-hedging calibration procedure separately for each bin in $B.$

\emph{(c) }If the calibeating procedure of Theorem \ref{th:beat-1} is not
calibrated, then one can construct another procedure that calibeats it, and
then another one that calibeats that, and so on. A calibeating procedure
that is calibrated, as obtained here, stops this infinite regress, which may
well be quickly overwhelmed by the accumulating errors of calibeating, as
well as those due to rounding up to a finite grid (see Remark (a) in
Appendix \ref{susus-a:simple-additional}).

\emph{(d) }A proof that is directly based on the minimax theorem is provided
in Appendix \ref{s-a:minimax}.

\bigskip

Calibration, and thus calibeating by a calibrated forecast, requires the
procedure to be stochastic. However, if we replace calibration with \emph{%
continuous calibration}, a weakening defined in Foster and Hart
(2021)---useful, in particular, for game dynamics that yield Nash
equilibria---we get a \emph{deterministic} procedure instead.

\begin{theorem}
\label{th:cont}Let $B$ be a finite set. Then there exists a \emph{%
deterministic} $\mathbf{b}$-based forecasting procedure $\zeta $ that is $B$%
-calibeating and is continuously calibrated.
\end{theorem}

We relegate the details to Appendix \ref{s:cont-calib}.

\section{Multi-calibeating\label{s:multi-beat}}

Suppose that there are $N\geq 1$ forecasting sequences, $\mathbf{b}%
^{n}=(b_{t}^{n})_{t\geq 1}$ for $n=1,2,...,N$. We assume that each $\mathbf{b%
}^{n}$ uses only finitely many forecasts: there is a finite set $B^{n}$ such
that $b_{t}^{n}\in B^{n}$ for all $t\geq 1$ (and, as in Section \ref%
{sus:general B}, while $B^{n}$ could be a subset of $C$, it may well be an
arbitrary set). Put $\mathbf{b=}(\mathbf{b}^{1},...,\mathbf{b}^{N});$ we are
looking for a $\mathbf{b}$-based forecasting procedure---i.e., $c_{t}$ is
determined after all the $b_{t}^{1},...,b_{t}^{N}$ are announced (and hence
is a function of $\mathbf{a}_{t-1},\mathbf{c}_{t-1},\mathbf{b}_{t}^{1},...,%
\mathbf{b}_{t}^{N}$---that simultaneously calibeats all the $\mathbf{b}^{n}$
sequences. We have:

\begin{theorem}
\label{th:multi}\emph{\textbf{(i)}} There exists a simple deterministic $(%
\mathbf{b}^{1},...,\mathbf{b}^{N})$-based forecasting procedure $\zeta $
that is $B^{n}$-calibeating for all $n=1,...,N;$ specifically, the forecast
of $\zeta $ in period $t$ is $c_{t}=\bar{a}_{t-1}^{\mathbf{b}^{1},...,%
\mathbf{b}^{N}}(b_{t}^{1},...,b_{t}^{N})$, the average of the actions in all
past periods $s\leq t-1$ where the combination $(b_{t}^{1},...,b_{t}^{N})$
was used (if $t$ is the first period in which $(b_{t}^{1},...,b_{t}^{N})$ is
used, take $c_{t}\in C$ to be arbitrary).

\emph{\textbf{(ii)}} For every finite $\delta $-grid $D$ of $C$ there exists
a stochastic $(\mathbf{b}^{1},...,\mathbf{b}^{N})$-based $D$-forecasting
procedure $\zeta $ that is $(\delta ,B^{n})$-calibeating for all $n=1,...,N$
and is $\delta $-calibrated. Moreover, $\zeta $ may be taken to be $\delta $%
-almost deterministic.

\emph{\textbf{(iii)}} There exists a deterministic $(\mathbf{b}^{1},...,%
\mathbf{b}^{N})$-based $C$-forecasting procedure $\zeta $ that is $B^{n}$%
-calibeating for all $n=1,...,N$ and is continuously calibrated.
\end{theorem}

\begin{proof}
This is immediate from the results of the previous sections by taking $(%
\mathbf{b}^{1},...,\mathbf{b}^{N})$ as $\mathbf{b}$ and using inequalities
such as $\mathcal{R}_{t}^{\mathbf{b}^{1},...,\mathbf{b}^{N}}\leq \mathcal{R}%
_{t}^{\mathbf{b}^{n}}$ for each $n$ by Appendix \ref{s-a:refined}.
\end{proof}

\bigskip

\textbf{\noindent Remarks. }\emph{(a) }The error term is%
\begin{equation*}
\gamma ^{2}\prod_{n=1}^{N}\left\vert B^{n}\right\vert \frac{\ln t+1}{t};
\end{equation*}%
thus, in (i) we have%
\begin{equation}
\mathcal{B}_{t}^{\mathbf{c}}\leq \mathcal{R}_{t}^{\mathbf{b}^{n}}+\gamma
^{2}\prod_{n=1}^{N}\left\vert B^{n}\right\vert \frac{\ln t+1}{t},
\label{eq:beat-n}
\end{equation}%
and in (ii) we have 
\begin{eqnarray*}
\mathbb{E}\left[ \mathcal{K}_{t}^{\mathbf{c}}\right] &\leq &\delta
^{2}+\gamma ^{2}\left\vert D\right\vert \prod_{n=1}^{N}\left\vert
B^{n}\right\vert \frac{\ln t+1}{t}\text{\ \ \ and} \\
\mathbb{E}\left[ \mathcal{B}_{t}^{\mathbf{c}}\right] &\leq &\mathbb{E}\left[ 
\mathcal{R}_{t}^{\mathbf{b}^{j}}\right] +\delta ^{2}+\gamma ^{2}\left\vert
D\right\vert \prod_{n=1}^{N}\left\vert B^{n}\right\vert \frac{\ln t+1}{t}
\end{eqnarray*}%
for all $n=1,...,N$, all $t\geq 1$, and all sequences $\mathbf{a,b}^{1},...,%
\mathbf{b}^{N}\mathbf{.}$

\emph{(b) }Since the constant $\prod_{n=1}^{N}\left\vert B^{n}\right\vert $
in the above error terms increases exponentially with $N$, we provide in
Appendix \ref{s-a:multi} multi-calibeating procedures that are more complex
but yield smaller error terms.

\emph{(c)} One may again allow the $B^{n}$ to be infinite, provided that $%
N_{t}^{\mathbf{b}^{n}}/t\rightarrow 0$ as $t\rightarrow \infty .$

\emph{(d)} The result in (ii) holds with probability one, and not only in
expectation; see Remark (a) in Section \ref{s:beat-by-calibrated} and
Appendix A5 in Foster and Hart (2021).\emph{\ }

\appendix{}

\section{Appendix}

The appendix contains additional results, proofs, and remarks.

\subsection{A Simple Way to Calibeat\label{sus-a:simple}}

We consider here the minimal calibeating error that can be guaranteed.
First, we show in Appendix \ref{susus-a:lower bound} that it must be at
least of the order of $\log t/t$, the same order obtained by Theorem \ref%
{th:beat-1}; second, we pin down the constant in Appendix \ref%
{susus-a:constant}: it is within a factor between $2$ and $4$ (depending on
the dimension $m$ and the geometric shape of the set $C)$ of the constant of
Theorem \ref{th:beat-1}. Additional comments on Section \ref{s:simple} are
provided in \ref{susus-a:simple-additional}.

\subsubsection{The Calibeating Error\label{susus-a:lower bound}}

We prove here that one cannot guarantee a calibeating error of an order of
magnitude lower than $\log t/t$. We show that this is so already in the
simplest one-dimensional case; see Remark (a) below for the extension to the
multidimensional case.

\begin{proposition}
\label{p:lb=1/4}Let $A=\{0,1\}$ and $C=[0,1]$, and let $\mathbf{b}$ be a
constant sequence (e.g., $b_{t}=1/2$ for all $t$). Then for every $\mathbf{b}
$-based forecasting procedure $\zeta $ we have\footnote{%
For a constant sequence $\mathbf{b,}$ a $\mathbf{b}$-based forecasting
procedure is simply a forecasting procedure. The expectation in (\ref%
{eq:ge-logt/t}) is over the stochastic choices of $\zeta $ (and it applies
only to $\mathcal{B}_{t}^{\mathbf{c}},$ since $\mathcal{R}_{t}^{\mathbf{b}}$
is determined by $\mathbf{a}_{t}$ alone when $\mathbf{b}_{t}$ is a constant
sequence).}%
\begin{equation}
\sup_{\mathbf{a}_{t}\in A^{t}}\mathbb{E}\left[ \mathcal{B}_{t}^{\mathbf{c}}%
\mathbf{-}\mathcal{R}_{t}^{\mathbf{b}}\right] \geq \left( \frac{1}{4}%
-o(1)\right) \frac{\ln t}{t}  \label{eq:ge-logt/t}
\end{equation}%
as $t\rightarrow \infty .$
\end{proposition}

\begin{proof}
Consider the game between the \textquotedblleft action
player\textquotedblright\ who chooses the actions $a_{t}$ and the
\textquotedblleft calibrating player" who chooses the sequence of forecasts $%
c_{t}$, with payoff $\mathcal{B}_{t}^{\mathbf{c}}\mathbf{-}\mathcal{R}_{t}^{%
\mathbf{b}}$ (cf. the \textquotedblleft calibration game" in Foster and Hart
2018, where the payoff was $\mathcal{K}_{t}^{\mathbf{c}}$). We will provide
a mixed strategy of the action player that guarantees that 
\begin{equation}
\inf_{\zeta }\mathbb{E}\left[ \mathcal{B}_{t}^{\mathbf{c}}\mathbf{-}\mathcal{%
R}_{t}^{\mathbf{b}}\right] \geq \left( \frac{1}{4}-o(1)\right) \frac{\ln t}{t%
},  \label{eq:ge lnt/t}
\end{equation}%
where the infimum is taken over all forecasting procedures $\zeta $ (and the
expectation is over the randomizations of both actions and forecasts). This
implies that for every such $\zeta $ there is for each $t\geq 1$ at least
one sequence $\mathbf{a}_{t}$ in $A^{t}$ for which the same inequality
holds; this is (\ref{eq:ge-logt/t}).

The mixed strategy of the action player that we provide consists of
conditionally i.i.d. actions, specifically, $a_{t}|\theta \sim \mathrm{%
Bernoulli}(\theta )$ where $\theta \sim \mathrm{Beta}(\alpha ,\alpha )$ for
a fixed $\alpha >0$ (this is the so-called \textquotedblleft beta-binomial"
distribution with parameters $\alpha =\beta $). The following formulas are
well known, and easy to see (e.g., Johnson, Kemp, and Kotz 2005): 
\begin{equation}
\mathbb{E}\left[ \bar{a}_{t}\right] =\frac{1}{2}\text{\ \ \ and\ \ \ }%
\mathbb{V}ar\left[ \bar{a}_{t}\right] =\mathbb{V}ar\left[ \bar{a}_{t}\right]
=\frac{t+2\alpha }{4(2\alpha +1)t};  \label{eq:EV(a-bar)}
\end{equation}%
the Bayesian estimate of $\theta $ given the history $h_{t}$ is%
\begin{equation*}
\hat{\theta}_{t}%
{\;:=\;}%
\mathbb{E}\left[ \theta |h_{t}\right] =\frac{t\bar{a}_{t}+\alpha }{t+2\alpha 
};
\end{equation*}%
thus, by (\ref{eq:EV(a-bar)}),%
\begin{equation}
\mathbb{E}\left[ \hat{\theta}_{t}\right] =\frac{1}{2}\text{\ \ \ and\ \ \ }%
\mathbb{V}ar\left[ \hat{\theta}_{t}\right] =\frac{t}{4(2\alpha +1)(t+2\alpha
)}.  \label{eq:EV(theta-hat)}
\end{equation}

The sequence $\mathbf{b}$ yields a single bin, and so $\mathcal{R}_{t}^{%
\mathbf{b}}$ is the variance of $a_{1},...,a_{t}$, i.e., $\bar{a}_{t}(1-\bar{%
a}_{t})$, which, by (\ref{eq:EV(a-bar)}), gives%
\begin{eqnarray}
\mathbb{E}\left[ \mathcal{R}_{t}^{\mathbf{b}}\right] &=&\mathbb{E}\left[ 
\bar{a}_{t}(1-\bar{a}_{t})\right] =\mathbb{E}\left[ \bar{a}_{t}\right] -%
\mathbb{E}^{2}\left[ \bar{a}_{t}\right] -\mathbb{V}ar\left[ \bar{a}_{t}%
\right]  \notag \\
&=&\frac{1}{2}-\frac{1}{4}-\frac{t+2\alpha }{4(2\alpha +1)t}=\lambda -\frac{%
\lambda }{t},  \label{eq:E(R)}
\end{eqnarray}%
where we used $\mathbb{E}\left[ \bar{a}_{t}^{2}\right] =\mathbb{E}^{2}\left[ 
\bar{a}_{t}\right] +\mathbb{V}ar\left[ \bar{a}_{t}\right] $, and put 
\begin{equation*}
\lambda 
{\;:=\;}%
\frac{\alpha }{2(2\alpha +1)}.
\end{equation*}%
Next, we have%
\begin{eqnarray*}
\mathbb{E}\left[ \mathcal{B}_{t}^{\mathbf{c}}\right] &=&\frac{1}{t}%
\sum_{s=1}^{t}\mathbb{E}\left[ (a_{s}-c_{s})^{2}\right] =\frac{1}{t}%
\sum_{s=1}^{t}\mathbb{E}\left[ \mathbb{E}\left[ (a_{s}-c_{s})^{2}|h_{s-1}%
\right] \right] \\
&\geq &\frac{1}{t}\sum_{s=1}^{t}\mathbb{E}\left[ \mathbb{V}ar\left[
a_{s}|h_{s-1}\right] \right] =\frac{1}{t}\sum_{s=1}^{t}\mathbb{E}\left[ \hat{%
\theta}_{s-1}(1-\hat{\theta}_{s-1})\right] ,
\end{eqnarray*}%
where the inequality is by $\mathbb{E}\left[ (X-Y)^{2}\right] \geq \mathbb{V}%
ar[X]$ for any $Y$ that is independent of\footnote{%
Use $\mathbb{E}\left[ (X-y)^{2}\right] \geq \mathbb{V}ar\left[ X\right] $
for each value $y$ of $Y.$ The inequality holds more generally for
nonpositively correlated $X$ and $Y,$ since $\mathbb{E}\left[ (X-Y)^{2}%
\right] \geq \mathbb{V}ar\left[ X-Y\right] =\mathbb{V}ar\left[ X\right] -2\,%
\mathbb{C}ov\left[ X,Y\right] +\mathbb{V}ar\left[ Y\right] $, which is $\geq 
\mathbb{V}ar\left[ X\right] $ when $\mathbb{C}ov\left[ X,Y\right] \leq 0$.} $%
X$, and the equality following it is by $a_{s}|h_{s-1}\sim \mathrm{Bernoulli}%
(\theta |h_{s-1})=\mathrm{Bernoulli}(\hat{\theta}_{s-1})$. Using (\ref%
{eq:EV(theta-hat)}) we get%
\begin{eqnarray*}
\mathbb{E}\left[ \hat{\theta}_{s-1}(1-\hat{\theta}_{s-1})\right] &=&\mathbb{E%
}\left[ \hat{\theta}_{s-1}\right] -\mathbb{E}^{2}\left[ \hat{\theta}_{s-1}%
\right] -\mathbb{V}ar\left[ \hat{\theta}_{s-1}\right] \\
&=&\frac{1}{2}-\frac{1}{4}-\frac{s-1}{4(2\alpha +1)(s-1+2\alpha )}=\lambda +%
\frac{\lambda }{s+2\alpha -1},
\end{eqnarray*}%
and thus%
\begin{equation*}
\mathbb{E}\left[ \mathcal{B}_{t}^{\mathbf{c}}\right] \geq \lambda +\frac{%
\lambda }{t}\sum_{s=1}^{t}\frac{1}{s+2\alpha -1}.
\end{equation*}%
Together with (\ref{eq:E(R)}) this yields%
\begin{equation*}
\mathbb{E}\left[ \mathcal{B}_{t}^{\mathbf{c}}-\mathcal{R}_{t}^{\mathbf{b}}%
\right] \geq \frac{\lambda }{t}\sum_{s=1}^{t}\frac{1}{s+2\alpha -1}+\frac{%
\lambda }{t}\sim \lambda \frac{\ln t}{t}
\end{equation*}%
as $t\rightarrow \infty $. Since $\lambda $ can be made arbitrarily close to%
\footnote{%
As $\alpha \rightarrow \infty $ the beta-binomial distribution converges to
the binomial distribution with $\theta =1/2,$ for which $\mathcal{R}%
_{t}\approx 1/4.$ We cannot however use this limit distribution, since $%
\theta $ being fixed yields a much smaller error, of the order of $1/t$
instead of $\log t/t.$} $1/4$ by taking large enough $\alpha $ we get (\ref%
{eq:ge lnt/t}), which completes the proof.
\end{proof}

\bigskip

\noindent \textbf{Remarks. }\emph{(a) }In the multidimensional case with $%
A=\{0,1\}^{m}$ and $C=[0,1]^{m}$ (for any $m\geq 1$), let $\mathbf{b}$ be a
constant sequence; applying the above result to each one of the $m$
coordinates separately and then summing up yields%
\begin{equation}
\sup_{\mathbf{a}_{t}\in A^{t}}\mathbb{E}\left[ \mathcal{B}_{t}^{\mathbf{c}}%
\mathbf{-}\mathcal{R}_{t}^{\mathbf{b}}\right] \geq \left( \frac{m}{4}%
-o(1)\right) \frac{\ln t}{t}  \label{eq:m/4}
\end{equation}%
as $t\rightarrow \infty .$

\emph{(b)} Given a finite set $B$, let the sequence $\mathbf{b}$ use all $b$
in $B$ with equal frequencies (for example, let the $b_{t}$ alternate in a
round-robin manner between the elements of $B$); applying Proposition \ref%
{p:lb=1/4} to the subsequence where $b_{t}=b$ for each $b\in B$ separately
and then summing up yields%
\begin{equation*}
\sup_{\mathbf{a}_{t}\in A^{t}}\mathbb{E}\left[ \mathcal{B}_{t}^{\mathbf{c}}%
\mathbf{-}\mathcal{R}_{t}^{\mathbf{b}}\right] \geq \left( \frac{1}{4}%
-o(1)\right) \left\vert B\right\vert \frac{\ln (t/\left\vert B\right\vert )}{%
t}.
\end{equation*}

\subsubsection{Improving the Constant\label{susus-a:constant}}

Here we will show that one can improve the calibeating error of Theorem \ref%
{th:beat-1} by a factor between $2$ and $4$ (depending on the dimension and
the shape of the set $C$), and this essentially matches the lower bound of
the previous Section \ref{susus-a:lower bound}.

Assuming that $C$ is a full-dimensional set in\footnote{%
If the affine space spanned by $C\subset \mathbb{R}^{m}$ has a lower
dimension $m^{\prime }<m,$ project everything to $\mathbb{R}^{m^{\prime }}.$}
$\mathbb{R}^{m}$, let $r$ be the \emph{radius of the minimal bounding sphere
of} $C;$ thus, $r$ is minimal such that $C\subseteq \bar{B}(c^{0},r)$ for
some $c^{0}\in C$. The relation of $r$ to the diameter $\gamma $ of $C$ is,
by Jung's (1901) theorem, 
\begin{equation}
r^{2}\leq \gamma ^{2}\frac{m}{2(m+1)}  \label{eq:Jung}
\end{equation}%
(and, of course, $\gamma \leq 2r).$

\begin{proposition}
\label{p:zeta-prime}Let $B$ be a finite set, and let $\zeta ^{\prime }$ be
the deterministic $\mathbf{b}$-based forecasting procedure given by 
\begin{equation*}
c_{t}^{\prime }%
{\;:=\;}%
\left( 1-\frac{1}{n_{t}^{\mathbf{b}}(b_{t})}\right) \bar{a}_{t-1}^{\mathbf{b}%
}(b_{t})+\frac{1}{n_{t}^{\mathbf{b}}(b_{t})}c^{0}
\end{equation*}%
for every time $t\geq 1$. Then $\zeta $ is $B$-calibeating, and%
\begin{equation}
\mathcal{B}_{t}^{\mathbf{c}^{\prime }}-\mathcal{R}_{t}^{\mathbf{b}}\leq
r^{2}|B|\frac{\ln t+1}{t}  \label{eq:Bc'-R}
\end{equation}%
for all $t\geq 1$ and all sequences $\mathbf{a}_{t}\in A^{t}$ and $\mathbf{b}%
_{t}\in B^{t}.$
\end{proposition}

Thus, the forecast $c_{t}^{\prime }$ of $\zeta ^{\prime }$ is an
appropriately weighted average of the forecast $c_{t}=\bar{a}_{t-1}^{\mathbf{%
b}}(b_{t})$ of the procedure $\zeta $ of Theorem \ref{th:beat-1} and the
fixed \textquotedblleft center" point $c^{0}$ of $C$. Compared with (\ref%
{eq:r-b}), the upper bound of (\ref{eq:Bc'-R}) on the calibeating error has $%
r^{2}$ instead of $\gamma ^{2}$, which, by (\ref{eq:Jung}), is an
improvement by a factor of at least $2;$ when $m=1$, by a factor of $4$. Of
course, $\zeta ^{\prime }$ gives up somewhat on the extreme simplicity of $%
\zeta $, i.e., (\ref{eq:c=a-bar}).

\bigskip

\begin{proof}
For any vectors $x,y\in \mathbb{R}^{m}$ and any scalar $\nu \in \lbrack 0,1]$%
, we have $\left\Vert x-(1-\nu )y\right\Vert ^{2}-(1-\nu )\left\Vert
x-y\right\Vert ^{2}=\nu \left\Vert x\right\Vert ^{2}-$ $\nu (1-\nu
)\left\Vert y\right\Vert ^{2}\leq \nu \left\Vert x\right\Vert ^{2}$.
Applying this to $x=a_{s}-c^{0}$, $y=\bar{a}_{s-1}^{\mathbf{b}}(b_{s})-c^{0}$%
, and $\nu =1/n_{s}^{\mathbf{b}}(b_{s})$ yields%
\begin{equation*}
\left\Vert a_{s}-c_{s}^{\prime }\right\Vert ^{2}-\left( 1-\frac{1}{n_{s}^{%
\mathbf{b}}(b_{s})}\right) \left\Vert a_{s}-\bar{a}_{s-1}^{\mathbf{b}%
}(b_{s})\right\Vert ^{2}\leq \frac{1}{n_{s}^{\mathbf{b}}(b_{s})}\left\Vert
a_{s}-c^{0}\right\Vert ^{2}\leq \frac{r^{2}}{n_{s}^{\mathbf{b}}(b_{s})}.
\end{equation*}%
Averaging the left-hand side for $s=1,...,t$ yields $\mathcal{B}_{t}^{%
\mathbf{c}^{\prime }}-\mathcal{R}_{t}^{\mathbf{b}}$, and so, putting $%
B_{t}:=\{b\in B:n_{t}^{\mathbf{b}}(b)>0\}$, we get%
\begin{equation*}
\mathcal{B}_{t}^{\mathbf{c}^{\prime }}-\mathcal{R}_{t}^{\mathbf{b}}\leq 
\frac{1}{t}\sum_{b\in B_{t}}\sum_{i=1}^{n_{t}^{\mathbf{b}}(b)}\frac{r^{2}}{i}%
\leq \frac{1}{t}\left\vert B\right\vert r^{2}(\ln t+1)
\end{equation*}%
(because $\left\vert B_{t}\right\vert \leq \left\vert B\right\vert $ and $%
n_{t}^{\mathbf{b}}(b)\leq t$), which is (\ref{eq:Bc'-R}).
\end{proof}

\bigskip

\noindent \textbf{Remark. }When $C=[0,1]^{m}$ for some $m\geq 1$, we have $%
r^{2}=m/4$ (take $c^{0}=(1/2,...,1/2)$), and thus%
\begin{equation*}
\mathcal{B}_{t}^{\mathbf{c}^{\prime }}-\mathcal{R}_{t}^{\mathbf{b}}\leq 
\frac{m|B|(\ln t+1)}{4t}.
\end{equation*}%
When $B$ is a singleton and $\mathbf{b}$ is a constant sequence, this upper
bound is $(m/4+o(1))\ln t/t$, which is asymptotically the same as the lower
bound of (\ref{eq:m/4}).

\subsubsection{Additional Comments\label{susus-a:simple-additional}}

We provide here a number of additional remarks on the result of Theorem \ref%
{th:beat-1} on simple calibeating.

\bigskip

\noindent \textbf{Remarks. }\emph{(a) }If we want our forecasts $c_{t}$ to
lie in a given $\delta $-grid $D\subset C$ (which may or may not be the same
as the grid $B$ used for $\mathbf{b}$), then taking $c_{t}\in D$ to be
within $\delta $ of $\bar{a}_{t-1}^{\mathbf{b}}(b_{t})$ introduces an
additional error of $2\gamma \delta $ in (\ref{eq:r-b}) (because $\left\vert
\left\Vert a-c\right\Vert ^{2}-\left\Vert a-d\right\Vert ^{2}\right\vert
\leq 2\gamma \left\Vert d-c\right\Vert $ for all $a,c,d\in C$), and thus it
yields $(\sqrt{2\gamma \delta },B)$-calibeating.

\emph{(b) }Let the $\mathbf{b}$-forecasts be generated by a forecasting
procedure $\sigma ;$ the procedure $\sigma ^{\prime }$ of Remark (e) in
Section \ref{s:simple}, whereby each $b_{t}$ is replaced by the
corresponding $\bar{a}_{t-1}^{\mathbf{b}}(b_{t})$, and which guarantees a
lower Brier score than $\sigma $ in the long run, is implemented as follows.
In each period $t$ one computes the forecast $b_{t}$ according to $\sigma $,
and then announces $c_{t}=\bar{a}_{t-1}^{\mathbf{b}}(b_{t})$ (the $b_{t}$ is
not announced). To carry this out one needs to recall the history $\mathbf{b}%
_{t-1}$, which in general need not be deducible from the history $(\mathbf{a}%
_{t-1},\mathbf{c}_{t-1})$ of $\sigma ^{\prime }$ (because different $b$-bins
may have had the same average, and so different $b_{s}$ may have yielded the
same $c_{s}=\bar{a}_{s-1}^{\mathbf{b}}(b_{s})$). In game-theoretic terms,
the resulting $\sigma ^{\prime }$ is \emph{not} a \emph{behavior} strategy
(which is what we have defined a forecasting procedure to be, in Section \ref%
{sus:calibration}), but rather a \emph{mixed} strategy (i.e., a
probabilistic mixture of pure, deterministic, strategies). However, since
the game between the \textquotedblleft forecasting player" and the
\textquotedblleft action player" is a game of perfect recall, by Kuhn's
(1953) theorem the mixed strategy $\sigma ^{\prime }$ induces an equivalent
behavior strategy $\sigma ^{\prime \prime }$, which is thus a forecasting
procedure (this \textquotedblleft equivalence" means that no matter what the
action player does, the probability of any outcome is the same under the
mixed strategy and the induced behavior strategy). The construction of $%
\sigma ^{\prime \prime }$ is straightforward (see, e.g., Hart 1992): for
every $t\geq 1$, history $(\mathbf{a}_{t-1},\mathbf{c}_{t-1})\in
A^{t-1}\times C^{t-1}$, and forecast $c_{t}\in C$, let $\Gamma _{t-1}:=\{%
\mathbf{b}_{t-1}:\bar{a}_{s-1}^{\mathbf{b}}(b_{s})=c_{s}$ for every $1\leq
s\leq t-1\}$ and $\Gamma _{t}:=\{\mathbf{b}_{t}:\bar{a}_{s-1}^{\mathbf{b}%
}(b_{s})=c_{s}$ for every $1\leq s\leq t\}$ be the sets of $\mathbf{b}_{t-1}$
and $\mathbf{b}_{t}$ that, together with the given $\mathbf{a}_{t-1}$, yield 
$\mathbf{c}_{t-1}$ and $\mathbf{c}_{t}$, respectively; then the probability
that $\sigma ^{\prime \prime }$ forecasts $c_{t}$ after $(\mathbf{a}_{t-1},%
\mathbf{c}_{t-1})$ is given by%
\begin{equation}
\sigma ^{\prime \prime }(\mathbf{a}_{t-1},\mathbf{c}_{t-1})(c_{t})%
{\;:=\;}%
\frac{\sum_{\mathbf{b}_{t}\in \Gamma _{t}}\prod_{s=1}^{t}\sigma (\mathbf{a}%
_{s-1},\mathbf{b}_{s-1})(b_{s})}{\sum_{\mathbf{b}_{t-1}\in \Gamma
_{t-1}}\prod_{s=1}^{t-1}\sigma (\mathbf{a}_{s-1},\mathbf{b}_{s-1})(b_{s})}.
\label{eq:kuhn}
\end{equation}

\emph{(c) }Can one do better than by choosing $c_{t}=\bar{a}_{t-1}^{\mathbf{b%
}}(b_{t})$ at each time $t$? Since $\widetilde{\mathcal{R}}_{t}^{\mathbf{b}}-%
\mathcal{R}_{t}^{\mathbf{b}}=O(\log t/t)\rightarrow 0$, consider the game
where our $\mathbf{c}$-forecaster wants to minimize $\mathcal{B}_{t}^{%
\mathbf{c}}-\widetilde{\mathcal{R}}_{t}^{\mathbf{b}}$ (instead of $\mathcal{B%
}_{t}^{\mathbf{c}}-\mathcal{R}_{t}^{\mathbf{b}})$ against an opponent that
controls the sequences $\mathbf{a}_{t}$ and $\mathbf{b}_{t};$ alternatively,
the opponent controls the sequence $\mathbf{a}_{t}$, whereas the sequence $%
\mathbf{b}_{t}$ is exogenous or is determined by history (i.e., by a
forecasting procedure). We claim that the strategy $\zeta $ of Theorem \ref%
{th:beat-1} is the \emph{unique} subgame-perfect optimal strategy of our
forecaster. To see this, consider 
\begin{equation}
\mathcal{B}_{t}^{\mathbf{c}}-\widetilde{\mathcal{R}}_{t}^{\mathbf{b}}=\frac{1%
}{t}\sum_{s=1}^{t}\left[ ||a_{s}-c_{s}||^{2}-\left\Vert a_{s}-\bar{a}_{s-1}^{%
\mathbf{b}}(b_{s})\right\Vert ^{2}\right] .  \label{eq:B-Rtilde}
\end{equation}%
Suppose that we are in a period $r\leq t$, and so the terms $s<r$ of the sum
(\ref{eq:B-Rtilde}) are all given. To \emph{guarantee}---no matter what the
future $a_{s}$ and $b_{s}$ will be---that the sum of the remaining terms,
i.e., $s\geq r$, is as small as possible, one \emph{must} now choose $c_{r}=%
\bar{a}_{r-1}^{\mathbf{b}}(b_{r})$. This follows since for every $\bar{a}\in 
\mathrm{conv}A$ and $c\neq \bar{a}$ we have\footnote{%
One way to see this is as follows. Let $\bar{a}=\sum_{i}\lambda _{i}a_{i}$
be a convex combination of elements $a_{i}$ in $A;$ then $\sum_{i}\lambda
_{i}\left\Vert a_{i}-c\right\Vert ^{2}=\sum_{i}\lambda _{i}\left\Vert a_{i}-%
\bar{a}\right\Vert ^{2}+\left\Vert c-\bar{a}\right\Vert ^{2}$ (because $\bar{%
a}$ is the weighted average of the $a_{i}),$ and so for some $i$ we must
have $\left\Vert a_{i}-c\right\Vert ^{2}\geq \left\Vert a_{i}-\bar{a}%
\right\Vert ^{2}+\left\Vert c-\bar{a}\right\Vert ^{2}.$} 
\begin{equation*}
\sup_{a\in A}[\left\Vert a-c\right\Vert ^{2}-\left\Vert a-\bar{a}\right\Vert
^{2}]\geq \left\Vert c-\bar{a}\right\Vert ^{2}>0,
\end{equation*}%
whereas $\left\Vert a-c\right\Vert ^{2}-\left\Vert a-\bar{a}\right\Vert
^{2}=0$ for every $a$ when $c=\bar{a}$. Thus 
\begin{equation*}
\min_{c\in C}\,\sup_{a\in A}[\left\Vert a-c\right\Vert ^{2}-\left\Vert a-%
\bar{a}\right\Vert ^{2}]=0,
\end{equation*}%
with a \emph{unique} minimizer at $c=\bar{a}$. Therefore, for any sequence $%
\mathbf{b}_{t}$ we have%
\begin{equation*}
\min_{c_{r},...,c_{t}\in C}\,\sup_{a_{r},...,a_{t}\in A}\sum_{s=r}^{t}\left[
||a_{s}-c_{s}||^{2}-\left\Vert a_{s}-\bar{a}_{s-1}^{\mathbf{b}%
}(b_{s})\right\Vert ^{2}\right] =0,
\end{equation*}%
with the minimum \emph{uniquely }attained by choosing $c_{s}=\bar{a}_{s-1}^{%
\mathbf{b}}(b_{s})$ for each $s=r,...,t.$

\subsection{A Minimax Proof of Calibeating by a Calibrated Forecast\label%
{s-a:minimax}}

The simplest proof of the existence of forecasts that are guaranteed to be
calibrated consists of an application of the von Neumann's (1928) minimax
theorem for finite two-person zero-sum games. See Hart (2021), which is a
writeup of a proof provided in 1995 (see Section 4, \textquotedblleft An
Argument of Sergiu Hart," in Foster and Vohra 1998); for generalizations,
see Sandroni (2003), Olszewski and Sandroni (2008), and Shmaya (2008). We
provide here a minimax proof of calibeating as well. Note that these minimax
proofs do not yield actual constructions of the corresponding procedures;
they just give simple proofs of their existence.

The \emph{calibeating game} consists of two players, which we call the
\textquotedblleft $AB$-player" and the \textquotedblleft $C$-player." In
each period $s=1,2,...,t$ the $AB$-player chooses a pair $(a_{s},b_{s})\in
A\times B$, then the chosen $b_{s}$ is revealed, and finally the $C$-player
chooses a forecast $c_{s}\in C$ (equivalently: first $b_{s}$ is chosen by
the $AB$-player and is publicly announced, and then $a_{s}$ and $c_{s}$ are
chosen by the two players independently). In period $t$, when the game ends,
the payoff of the $C$-player is his calibration score $\mathcal{K}_{t}^{%
\mathbf{c}}$. The joint choice of the action $a_{s}$ and the forecast $b_{s}$
is equivalent to allowing an arbitrary dependence between them; i.e., the $%
\mathbf{b}$-forecaster may have any degree of \textquotedblleft knowledge"
or \textquotedblleft expertise" about the action sequence $\mathbf{a}$ (from
\textquotedblleft no knowledge,\textquotedblright\ where $a_{s}$ and $b_{s}$
are chosen independently, all the way to \textquotedblleft complete
knowledge," where they are fully correlated, e.g., $b_{s}=a_{s}$).

The sets $A$ and $B$ are assumed to be finite (with $A\subset C$), and we
will restrict the $C$-player to use a finite set $D\subset C$, which makes
the calibeating game a finite game. Moreover, we will make the $\mathbf{c}$%
-binning a refinement of the $\mathbf{b}$-binning, and so $\mathcal{R}_{t}^{%
\mathbf{c}}\leq \mathcal{R}_{t}^{\mathbf{b}}$ (by Proposition \ref%
{p:refine-R} in Appendix \ref{s-a:refined} below). Therefore, a strategy $%
\sigma $ of the $C$-player that guarantees $\mathbb{E}\left[ \mathcal{K}%
_{t}^{\mathbf{c}}\right] \leq \varepsilon $ against any strategy $\tau $ of
the $AB$-player gives $\mathbb{E}\left[ \mathcal{B}_{t}^{\mathbf{c}}-%
\mathcal{R}_{t}^{\mathbf{b}}\right] =\mathbb{E}\left[ \mathcal{K}_{t}^{%
\mathbf{c}}+\mathcal{R}_{t}^{\mathbf{c}}-\mathcal{R}_{t}^{\mathbf{b}}\right]
\leq \mathbb{E}\left[ \mathcal{K}_{t}^{\mathbf{c}}\right] \leq \varepsilon $%
, and thus it yields $B$-calibeating by a calibrated procedure (cf. Theorem %
\ref{th:beat-by-calib}).

We proceed as follows. For every $b\in B\mathbf{\ }$we take $D_{b}$ to be a
finite $\delta $-grid of $C$, such that these grids are disjoint, i.e., $%
D_{b}\cap D_{b^{\prime }}=\emptyset $ for all $b\neq b^{\prime };$ put $%
D:=\cup _{b\in B}D_{b}$. We then restrict the $C$-player to use only
forecasts in $D_{b}$ after $b$ is announced, i.e., $c_{s}\in D_{b_{s}}$ for
every $s;$ this indeed makes the $\mathbf{c}$-binning a refinement of the $%
\mathbf{b}$-binning.

We claim that for every mixed strategy $\tau $ of the $AB$-player there is a
strategy $\sigma $ of the $C$-player such that $\mathbb{E}\left[ \mathcal{K}%
_{t}^{\mathbf{c}}\right] \leq \varepsilon :=O(\delta +1/\sqrt{t})$ (we do
not try optimize the bounds here). Indeed, in each period $s\leq t$ choose $%
c_{s}\in D_{b_{s}}$ to be such that $\left\Vert c_{s}-\mathbb{E}\left[
a_{s}|h_{s-1},b_{s}\right] \right\Vert \leq \delta $ (i.e., take the
conditional probability---given the history $h_{s-1}=(\mathbf{a}_{s-1},%
\mathbf{b}_{s-1},\mathbf{c}_{s-1})$ and the announced $b_{s}$---that is
generated by the mixed strategy $\tau $ of the $AB$-player, and round it up
to the $\delta $-grid $D_{b_{s}}$ that is used for $b_{s}$). A standard
computation, as in Hart (2021), shows that $\mathbb{E}\left[ K_{t}\right]
\leq \delta +\gamma \sqrt{\left\vert D\right\vert }/\sqrt{t}$, where $%
K_{t}:=(1/t)\sum_{d\in D}n_{t}(d)\left\Vert e_{t}(d)\right\Vert $ is the $%
\ell _{1}$-calibration score, for which we trivially have $\mathcal{K}%
_{t}\leq \gamma K_{t}$ (see Section \ref{sus:scores}); this proves the claim.

The minimax theorem therefore yields a strategy $\sigma $ of the $C$-player
that guarantees $\mathbb{E}\left[ \mathcal{K}_{t}^{\mathbf{c}}\right] \leq
\varepsilon $ against all strategies $\tau $ of the $AB$-player; we have
thus obtained calibration and, as shown above, $B$-calibeating.

\subsection{Complexity of Procedures: Minimax (MM) and Fixed Point (FP)
Procedures\label{sus:fp-mm procedures}}

The basic calibeating procedure of Theorem \ref{th:beat-1} is very simple,
as it requires just the computation of averages; the same holds for the
multi-calibeating procedure of Theorem \ref{th:multi} (i). The other
procedures that we provide are more complex, and require solving at each
step a certain multidimensional problem. These problems turn out to be of
two distinct kinds: for stochastic procedures, they are finite minimax (or
linear programming) problems, and for deterministic and almost-deterministic
procedures, they are continuous fixed point problems. (The existence of the
corresponding solutions is proven by the von Neumann 1928 minimax theorem
and the Brouwer 1912 fixed point theorem, respectively; see Appendix \ref%
{s-a:outgoing} below.) Following Foster and Hart (2021, Section III.D), we
refer to these as being of \emph{type MM}\ (minmax) and \emph{type FP}
(fixed point), respectively.

This distinction, which is significant in the multidimensional case (i.e.,
for $m>1$) and is of the polynomial vs. nonpolynomial variety, is not just a
matter of proof technique; see Foster and Hart (2021), Sections III.D, VI,
and VII (with a summary in Table I there). Theorem \ref{th:outgoing} in
Appendix \ref{s-a:outgoing}, which provides the \textquotedblleft outgoing"
tools that we use, makes the distinction clear: part (S) gives procedures of
type MM, and parts (D) and (AD) give procedures of type FP.

Specifically, the results that yield procedures of type MM are: calibration
(Theorem \ref{th:calibration}), calibeating with calibration (Theorem \ref%
{th:beat-by-calib}), and multi-calibeating with calibration (Theorem \ref%
{th:multi} (ii))---all of them without the \textquotedblleft moreover"
almost-deterministic statement. The results that yield FP procedures are all
the above \textquotedblleft moreover" statements, calibeating with
continuous calibration, and multi-calibeating with continuous calibration
(Theorems \ref{th:cont}, \ref{th:cont-calib}, and \ref{th:multi} (iii)).

\subsection{\textquotedblleft Outgoing" Results\label{s-a:outgoing}}

We provide the results of the \textquotedblleft outgoing" theorems of Foster
and Hart (2021), restating them in a convenient manner for our use. The
seemingly slightly weaker formulations here are still equivalent to
Brouwer's fixed point theorem and von Neumann's minimax theorem,
respectively; see Remarks (c) and (d) below. The FP vs. MM distinction is
discussed in Appendix \ref{sus:fp-mm procedures} above. A probability
distribution $\eta $ is called \textquotedblleft $\delta $-local" if its
support is included in a ball of radius $\delta ;$ i.e., there is $y$ such
that $\eta (\bar{B}(y;\delta ))=1.$

\begin{theorem}
\label{th:outgoing}Let $C\subset \mathbb{R}^{m}$ be a nonempty compact
convex set.

\emph{\textbf{(D)}} Let $g:C\rightarrow \mathbb{R}^{m}$ be a continuous
function. Then there exists a point $y$ in $C$ that is of type FP, such that%
\begin{equation}
\left\Vert x-y\right\Vert \leq \left\Vert x-g(y)\right\Vert \;\;\text{for
all }x\in C.  \label{eq:out_fp-fp}
\end{equation}

\emph{\textbf{(S)}} Let $D\subset C$ be a finite $\delta $-grid of $C$ for
some $\delta >0$, and let $g:D\rightarrow \mathbb{R}^{m}$ be a function.
Then there exists a probability distribution $\eta $ on $D$ that is of type
MM and has support of size at most $m+3$, such that%
\begin{equation}
\mathbb{E}_{y\sim \eta }\left[ \left\Vert x-y\right\Vert ^{2}\right] \leq 
\mathbb{E}_{y\sim \eta }\left[ \left\Vert x-g(y)\right\Vert ^{2}\right]
+\delta ^{2}\;\;\text{for all }x\in C.  \label{eq:out-mm-fp}
\end{equation}

\emph{\textbf{(AD)}} Let $D\subset C$ be a finite $\delta $-grid of $C$ for
some $\delta >0$, and let $g:D\rightarrow \mathbb{R}^{m}$ be a function.
Then there exists a probability distribution $\eta $ on $D$ that is $\delta $%
-local, of type FP, has support of size at most $m+1$, and satisfies (\ref%
{eq:out-mm-fp}).
\end{theorem}

\begin{proof}
We will use the following easy-to-verify identity%
\begin{equation}
\left\Vert x-y\right\Vert ^{2}-\left\Vert x-z\right\Vert ^{2}=2(z-y)\cdot
(x-y)-\left\Vert z-y\right\Vert ^{2},  \label{eq:identity}
\end{equation}%
with $z=g(y),$ to get from the statements in Foster and Hart (2021) to the
present ones.

(D) The fixed point outgoing theorem 4 of Foster and Hart (2021) applied to
the function $f(x)=g(x)-x$ yields a point $y\in C$ such that for all $x\in C$
we have 
\begin{equation*}
(g(y)-y)\cdot (x-y)\leq 0,
\end{equation*}%
and thus $\left\Vert x-y\right\Vert ^{2}-\left\Vert x-g(y)\right\Vert
^{2}\leq 0$, by (\ref{eq:identity}) with $z=g(y).$

(S) The minimax outgoing theorem 5 of Foster and Hart (2021) applied to the
function $f(x)=g(x)-x$ yields a distribution $\eta \in \Delta (D)$ such that
for all $x\in C$ we have%
\begin{equation}
\mathbb{E}_{y\sim \eta }\left[ (g(y)-y)\cdot (x-y)\right] \leq \delta \,%
\mathbb{E}_{y\sim \eta }\left[ \left\Vert g(y)-y\right\Vert \right] ,
\label{eq:out-mm}
\end{equation}%
and thus, by (\ref{eq:identity}) with $z=g(y)$,%
\begin{equation*}
\mathbb{E}_{y\sim \eta }\left[ \left\Vert x-y\right\Vert ^{2}-\left\Vert
x-g(y)\right\Vert ^{2}\right] \leq \mathbb{E}_{y\sim \eta }\left[ 2\delta
\left\Vert g(y)-y\right\Vert \right] -\mathbb{E}_{y\sim \eta }\left[
\left\Vert g(y)-y\right\Vert ^{2}\right] ,
\end{equation*}%
which gives (\ref{eq:out-mm-fp}) since $2\delta \left\Vert g(y)-y\right\Vert
\leq \delta ^{2}+\left\Vert g(y)-y\right\Vert ^{2}.$

(AD) This is the same proof as for (S), except that it uses the almost
deterministic outgoing theorem 7 of Foster and Hart (2021).
\end{proof}

\bigskip

What (\ref{eq:out_fp-fp}) says is that $y$ is closer than $g(y)$ to each
point $x$ in $C;$ similarly, (\ref{eq:out-mm-fp}) says that the random $y$
with distribution $\eta $ is closer on average than $g(y)$ (within a $\delta 
$-tolerance) to each point $x$ in $C$. To get some intuition, let $\lambda
\equiv \lambda (x):=\left\Vert x-y\right\Vert ^{2}-\left\Vert
x-g(y)\right\Vert ^{2};$ if $g:C\rightarrow C$ (as in Brouwer's fixed point
theorem) then condition (\ref{eq:out_fp-fp}), which says that $\lambda $ $%
\leq 0$ for every $x\in C$, is equivalent to $g(y)=y$, i.e., to $y$ being a
fixed point of $g$ (indeed, for a fixed point $y$ we have $\lambda =0$ for
all $x;$ conversely, for $x=g(y)$, which is a point in $C$, we get $\lambda
=\left\Vert g(y)-y\right\Vert ^{2}\leq 0$, and thus $g(y)=y$). Condition (%
\ref{eq:out-mm-fp}) extends this by requiring that $\lambda \leq 0$ hold
approximately on average, i.e., $\mathbb{E}\left[ \lambda \right] \leq
\delta ^{2}$, for every $x\in C$. This suggests (\ref{eq:out-mm-fp}) as a
suitable concept of a \textquotedblleft \emph{stochastic approximate fixed
point}" (note that a point $y$ such that $y$ and $g(y)$ are close---a
natural attempt to define an approximate fixed point concept---need not
exist in general: take, for example, the function $g:[0,1]\rightarrow
\lbrack 0,1]$ given by $g(x)=1$ for $x\leq 1/2$, and $g(x)=0$ for $x>1/2$,
for which $|g(x)-x|\geq 1/2$ for all $x$; this example also shows that one
cannot strengthen $\mathbb{E}\left[ \lambda \right] \leq \delta ^{2}$ to $%
\mathbb{E}\left[ \left\vert \lambda \right\vert \right] \leq \delta ^{2}$
---i.e., \textquotedblleft $\lambda =0$" instead of \textquotedblleft $%
\lambda \leq 0$"---because for $x=0$ we have $\left\vert \lambda \right\vert
=\left\vert y^{2}-g(y)^{2}\right\vert \geq 1/4$ for all $y\in C$).

\bigskip

\noindent \textbf{Remarks. }\emph{(a)} It suffices to consider functions $g$
whose range is included in\footnote{%
In the application of these results to calibration (both in Foster and Hart
2021 and in the present paper) the functions are always into $C$ (for
instance, $g(c)=\bar{a}_{t-1}(c)).$} $C$. Indeed, replacing $g$ with the
function $\hat{g}$ given by $\hat{g}(x):=\mathrm{proj}_{C}g(x)$ for every $x$
(which is continuous when $g$ is continuous) can only decrease the
right-hand side of inequalities (\ref{eq:out_fp-fp}) and (\ref{eq:out-mm-fp}%
) (because $\left\Vert x-\mathrm{proj}_{C}z\right\Vert \leq \left\Vert
x-z\right\Vert $ for every $x\in C$ and $z\in \mathbb{R}^{m}$), and so if
they hold for $\hat{g}$ then they hold for $g$ as well. For a direct proof
of Theorem \ref{th:outgoing} (D) by Brouwer's fixed point theorem, let $y$
be a fixed point of $\hat{g};$ then $\left\Vert x-y\right\Vert =\left\Vert x-%
\hat{g}(y)\right\Vert \leq \left\Vert x-g(y)\right\Vert $ for all $x\in C$.

\emph{(b) }A direct proof of Theorem \ref{th:outgoing} (S) by von Neumann's
minimax theorem for finite games is as follows. Let $\delta _{0}:=\max_{x\in
C}\mathrm{dist}(x,D)<\delta $ (see the proof of Theorem 5 in Foster and Hart
2021), put $\delta _{1}:=(\delta ^{2}-\delta _{0}^{2})/4\gamma >0$, and take 
$D_{1}\subset C$ to be a finite $\delta _{1}$-grid of $C$. Consider the
finite two-person zero-sum game where the maximizer chooses $x\in D_{1}$,
the minimizer chooses $y\in D$, and the payoff is\ $\left\Vert
x-y\right\Vert ^{2}-\left\Vert x-g(y)\right\Vert ^{2}$, where $%
g:D\rightarrow C$ is the given function (use Remark (a) above). For every
mixed strategy $\xi \in \Delta (D_{1})$ of the maximizer, let $\bar{x}:=%
\mathbb{E}_{x\sim \xi }\left[ x\right] \in C$ be its expectation; the
minimizer can make the payoff $\leq \delta _{0}^{2}$ by choosing a point $y$
on the grid $D$ that is within $\delta _{0}$ of $\bar{x}$, since%
\begin{eqnarray*}
\mathbb{E}_{x\sim \xi }\left[ \left\Vert x-y\right\Vert ^{2}\right] &=&%
\mathbb{E}_{x\sim \xi }\left[ \left\Vert x-\bar{x}\right\Vert ^{2}\right]
+\left\Vert \bar{x}-y\right\Vert ^{2}\leq \mathbb{E}_{x\sim \xi }\left[
\left\Vert x-\bar{x}\right\Vert ^{2}\right] +\delta _{0}^{2} \\
\mathbb{E}_{x\sim \xi }\left[ \left\Vert x-g(y)\right\Vert ^{2}\right] &=&%
\mathbb{E}_{x\sim \xi }\left[ \left\Vert x-\bar{x}\right\Vert ^{2}\right]
+\left\Vert \bar{x}-g(y)\right\Vert ^{2}\geq \mathbb{E}_{x\sim \xi }\left[
\left\Vert x-\bar{x}\right\Vert ^{2}\right] .
\end{eqnarray*}%
Therefore, by the minimax theorem, the minimizer can guarantee that the
payoff is $\leq \delta _{0}^{2}$; i.e., there is a mixed strategy $\eta \in
\Delta (D)$ such that 
\begin{equation}
\mathbb{E}_{y\sim \eta }\left[ \left\Vert x-y\right\Vert ^{2}-\left\Vert
x-g(y)\right\Vert ^{2}\right] \leq \delta _{0}^{2}  \label{eq:delta-M}
\end{equation}%
for every $x\in D_{1}$. Now for every $x\in C$ there is $x^{\prime }\in
D_{1} $ with $\left\Vert x-x^{\prime }\right\Vert <\delta _{1}$, and so $%
\left\vert \left\Vert x-z\right\Vert ^{2}-\left\Vert x^{\prime
}-z\right\Vert ^{2}\right\vert <2\gamma \delta _{1}$ for any $z$ in $C;$
adding this inequality with $z=y$ and with $z=g(y)$ to (\ref{eq:delta-M})
yields, by the definition of $\delta _{1}$, the claimed result (\ref%
{eq:out-mm-fp}).

\emph{(c)} Theorem \ref{th:outgoing} (D) is equivalent to Brouwer's fixed
point theorem. Indeed, the former has been proved by using the latter (see
Foster and Hart 2021 or Remark (a) above); conversely, given a continuous
function $g:C\rightarrow C$, inequality (\ref{eq:out_fp-fp}) for the point $%
x=g(y)$ (which is in $C)$ yields $g(y)=y.$

\emph{(d)} Theorem \ref{th:outgoing} (S) is equivalent to von Neumann's
minimax theorem for finite games. Indeed, the former has been proved by
using the latter (see Foster and Hart 2021 or Remark (b) above); conversely,
we will show that the former yields Corollary 6 of Foster and Hart (2021),
from which the minimax theorem follows by Remark 4 of Corollary 6 in
Appendix A3.3 of Foster and Hart (2021). For this Corollary 6, let $%
f:C\rightarrow \mathbb{R}^{m}$ be a bounded function, say, $\left\Vert
f(x)\right\Vert \leq M$ for all $x\in C$, and let $\varepsilon >0$. Applying
(\ref{eq:out-mm-fp}) to the function $g(x)=x+\delta f(x)$ and a finite $%
\delta $-grid $D$ of $C$ yields a distribution $\eta \in \Delta (C)$ such
that $\mathbb{E}\left[ 2\delta f(y)\cdot (x-y)\right] \leq \mathbb{E}\left[
\delta ^{2}\left\Vert f(y)\right\Vert ^{2}\right] +\delta ^{2}\leq \delta
^{2}(M^{2}+1)$ (use the identity (\ref{eq:identity})), and thus, by choosing 
$\delta $ so that $\delta (M^{2}+1)/2\leq \varepsilon $, we get $\mathbb{E}%
\left[ f(y)\cdot (x-y)\right] \leq \varepsilon $, which is the result of
Corollary 6 of Foster and Hart (2021).

\emph{(e)} As in Remark 1 of Theorem 5 in Appendix A3.2 of Foster and Hart
(2021), one may use a limit argument (which, however, no longer yields the
distribution $\eta $ by a finite minimax construct) to replace $\delta $ in (%
\ref{eq:out-mm}), and thus in (\ref{eq:out-mm-fp}), with $\delta _{0}\equiv
\delta _{0}(D):=\max_{x\in C}\mathrm{dist}(x,D)<\delta $. Thus,%
\begin{eqnarray*}
&&\inf_{\eta \in \Delta (D)}\sup_{x\in C}\mathbb{E}_{y\sim \eta }\left[
\left\Vert x-y\right\Vert ^{2}-\left\Vert x-g(y)\right\Vert ^{2}\right] \\
&=&\sup_{\xi \in \Delta (C)}\inf_{y\in D}\mathbb{E}_{x\sim \xi }\left[
\left\Vert x-y\right\Vert ^{2}-\left\Vert x-g(y)\right\Vert ^{2}\right] \leq
\delta _{0}^{2}.
\end{eqnarray*}%
Moreover, the $\delta _{0}^{2}$ bound is tight (i.e., it cannot be replaced
by any smaller constant): let $x_{0}\in C$ be such that $\mathrm{dist}%
(x_{0},D)=\delta _{0};$ then, for the constant function $g\equiv x_{0}$
(i.e., $g(y)=x_{0}$ for all $y\in D$), we have $\left\Vert
x_{0}-y\right\Vert \geq \delta _{0}$ and $\left\Vert x_{0}-g(y)\right\Vert
=0 $ for all $y\in D.$

\bigskip

Finally, we point out that Theorem \ref{th:log-outgoing-MM} in Appendix \ref%
{s-a:log} provides a result for logarithmic scores that is parallel to
Theorem \ref{th:outgoing} (S).

\subsection{Refined Refinement\label{s-a:refined}}

In this appendix we prove formally that the refinement score is
monotonically decreasing with respect to refining the binning; this yields
in particular $\mathcal{R}_{t}^{\mathbf{b}^{1},...,\mathbf{b}^{N}}\leq 
\mathcal{R}_{t}^{\mathbf{b}^{n}}$ for each $n=1,....,N$ (Section \ref%
{s:multi-beat}) and also $\mathcal{R}_{t}^{\mathbf{b},\Pi }\leq \mathcal{R}%
_{t}^{\mathbf{b}}$ and $\mathcal{R}_{t}^{\mathbf{b},\Pi }\leq \mathcal{R}%
_{t}^{\Pi }$ (Appendix \ref{s:cont-calib}).

We consider general fractional binnings. Let $I$ be a finite or countably
infinite collection of bins, and consider a sequence $(z_{s})_{s\geq 1}$
(namely, $z_{s}=a_{s}-c_{s}$) such that at time $s$ the fraction $\lambda
_{s}(i)\geq 0$ of $z_{s}$ is assigned to bin $i$ for each $i\in I$, where $%
\sum_{i\in I}\lambda _{s}(i)=1$ (the specific way in which these weights are
determined will not matter). Fix the horizon $t\geq 1$ (we will thus drop
the subscript $t$); the refinement score is%
\begin{equation*}
\mathcal{R}=\frac{1}{t}\sum_{i\in I}\sum_{s=1}^{t}\lambda _{s}(i)(z_{s}-\bar{%
z}(i))^{2},
\end{equation*}%
where, for each $i$ in $I,$%
\begin{equation*}
\bar{z}(i)=\frac{\sum_{s=1}^{t}\lambda _{s}(i)z_{s}}{\sum_{s=1}^{t}\lambda
_{s}(i)}
\end{equation*}%
is the average of bin $i$ (when $\sum_{s\leq t}\lambda _{s}(i)>0$).

As in Section \ref{sus:scores}, let the two-dimensional random variable $%
(Z,U)$ take the value $(z_{s},i)$ with probability $\lambda _{s}(i)/t$ for
each $s=1,...,t$ and $i\in I$ (note that $\sum_{s\leq t}\sum_{i\in I}\lambda
_{s}(i)/t=1$); thus, $\mathbb{P}\left[ (Z,U)=(z,i)\right] =(1/t)\sum_{s\leq
t:z_{s}=z}\lambda _{s}(i)$, which is the average, over all periods $%
s=1,...,t $, of the probability that the value $z$ goes into bin $i$. We
then have 
\begin{eqnarray}
\mathbb{P}\left[ U=i\right] &=&\sum_{s=1}^{t}\frac{\lambda _{s}(i)}{t}, 
\notag \\
\mathbb{E}\left[ Z|U=i\right] &=&\frac{1}{\mathbb{P}\left[ U=i\right] }%
\sum_{s=1}^{t}\left( \frac{\lambda _{s}(i)}{t}\right) z_{s}\;\;=\;\;\bar{z}%
(i),  \notag \\
\mathbb{V}ar\left[ Z|U=i\right] &=&\frac{1}{\mathbb{P}\left[ U=i\right] }%
\sum_{s=1}^{t}\left( \frac{\lambda _{s}(i)}{t}\right) (z_{s}-\bar{z}(i))^{2},%
\text{and}  \notag \\
\mathbb{E}\left[ \mathbb{V}ar\left[ Z|U\right] \right] &=&\sum_{i\in I}%
\mathbb{P}\left[ U=i\right] \,\mathbb{V}ar\left[ Z|U=i\right] \;\;=\;\;%
\mathcal{R}.  \label{eq:EV=R}
\end{eqnarray}

Now assume that we are given another collection of bins $J$ together with
binning weights $\mu _{s}(j)\geq 0$, where $\sum_{j\in J}\mu _{s}(j)=1$ for
each $s$. The $J$-binning is a \emph{coarsening} of the $I$-binning
(equivalently, the $I$-binning is a \emph{refinement} of the\emph{\ }$J$%
-binning) if there is a function $\phi :I\rightarrow J$ such that $\mu
_{s}(j)=\sum_{i:\phi (i)=j}\lambda _{s}(i);$ that is, for each $j$ in $J$
the $j$-bin is the union of the set $\phi ^{-1}(j)=\{i\in I:\phi (i)=j\}$ of 
$i$-bins in $I$. Letting $U_{I}$ and $U_{J}$ be the random variables
corresponding to the $I$-binning and the $J$-binning, respectively, we have $%
U_{J}=\phi (U_{I})$, because being assigned to an $i$-bin for $i\in I$
translates to being assigned to the $j$-bin for $j=\phi (i)\in J$. Let $%
\mathcal{R}_{I}$ and $\mathcal{R}_{J}$ be the refinement scores
corresponding to the $I$-binning and the $J$-binning, respectively.

\begin{proposition}
\label{p:refine-R}If the $J$-binning is a coarsening of the $I$-binning then%
\begin{equation*}
\mathcal{R}_{J}=\mathbb{E}\left[ \mathbb{V}ar\left[ Z|U_{J}\right] \right]
\geq \mathbb{E}\left[ \mathbb{V}ar\left[ Z|U_{I}\right] \right] =\mathcal{R}%
_{I}.
\end{equation*}
\end{proposition}

\begin{proof}
Let $\mathcal{F}_{1},\mathcal{F}_{2}$ be two $\sigma $-fields such that $%
\mathcal{F}_{1}\subseteq \mathcal{F}_{2}$, i.e., $\mathcal{F}_{1}$ is a
coarsening of $\mathcal{F}_{2}$, and let $Z$ be a random variable. We will
show that%
\begin{equation}
\mathbb{E}\left[ \mathbb{V}ar\left[ Z|\mathcal{F}_{1}\right] \right] \geq 
\mathbb{E}\left[ \mathbb{V}ar\left[ Z|\mathcal{F}_{2}\right] \right] ,
\label{eq:E[V]}
\end{equation}%
which yields the result by (\ref{eq:EV=R}).

Applying the classic inequality $\mathbb{V}ar\left[ X\right] =\mathbb{E}%
\left[ \left\Vert X-\mathbb{E}\left[ X\right] \right\Vert ^{2}\right] \leq 
\mathbb{E}\left[ \left\Vert X-x\right\Vert ^{2}\right] $ for any random
variable $X$ and any constant $x$ (i.e., the expected square deviation from
a constant is minimized when the constant equals the expectation) to $Z|%
\mathcal{F}_{2}$ we get (a.s.) 
\begin{equation*}
\mathbb{V}ar\left[ Z|\mathcal{F}_{2}\right] \leq \mathbb{E}\left[ \left\Vert
Z-\mathbb{E}\left[ Z|\mathcal{F}_{1}\right] \right\Vert ^{2}|\mathcal{F}_{2}%
\right] ,
\end{equation*}%
because $\mathbb{E}\left[ Z|\mathcal{F}_{1}\right] $ is constant given $%
\mathcal{F}_{2}$ (since $\mathcal{F}_{1}$ is a coarsening of $\mathcal{F}%
_{2} $). Taking expectation conditional on $\mathcal{F}_{1}$ yields on the
right-hand side $\mathbb{E}\left[ \left\Vert Z-\mathbb{E}\left[ Z|\mathcal{F}%
_{1}\right] \right\Vert ^{2}|\mathcal{F}_{1}\right] $ (again, by $\mathcal{F}%
_{1}\subseteq \mathcal{F}_{2}$), which is the conditional variance $\mathbb{V%
}ar\left[ Z|\mathcal{F}_{1}\right] $, and so we have (a.s.)%
\begin{equation*}
\mathbb{E}\left[ \mathbb{V}ar\left[ Z|\mathcal{F}_{2}\right] |\mathcal{F}_{1}%
\right] \leq \mathbb{V}ar\left[ Z|\mathcal{F}_{1}\right] .
\end{equation*}%
Taking overall expectation yields (\ref{eq:E[V]}), and thus completes the
proof.
\end{proof}

\bigskip

Applying Proposition \ref{p:refine-R} with $\phi $ being a projection, such
as $\phi (b^{1},...,b^{N})=b^{n}$, yields the needed inequalities.

\subsection{General Brier Score Decomposition\label{s-a:decompose}}

We show here that the decomposition of the Brier score (\ref{eq:B=R+K}) into
the refinement and calibration scores holds for any fractional binning $\Pi
=(w_{i})_{i=1}^{I}$, i.e.,%
\begin{equation*}
\mathcal{B}_{t}=\mathcal{R}_{t}^{\Pi }+\mathcal{K}_{t}^{\Pi }.
\end{equation*}%
Indeed, in the notation of the previous Section \ref{s-a:refined}, this is%
\begin{equation*}
\mathbb{E}\left[ \left\Vert Z\right\Vert ^{2}\right] =\mathbb{E}\left[ 
\mathbb{E}[\left\Vert Z\right\Vert ^{2}|U]\right] =\mathbb{E}\left[ \mathbb{V%
}ar\left[ Z|U\right] \right] +\mathbb{E}\left[ \left\Vert \mathbb{E}\left[
Z|U\right] \right\Vert ^{2}\right] ,
\end{equation*}%
which follows from applying the identity $\mathbb{E}\left[ X^{2}\right] =%
\mathbb{V}ar\left[ X\right] +\mathbb{E}\left[ X\right] ^{2}$ to each one of
the $m$ coordinates of $Z|U$, summing up, and then taking overall
expectation.

\subsection{Calibeating by a Deterministic Continuously Calibrated Forecast 
\label{s:cont-calib}}

In this appendix we prove Theorem \ref{th:cont} in Section \ref%
{s:beat-by-calibrated}: one can guarantee calibeating by a \emph{%
deterministic} procedure that is \emph{continuously calibrated}, a useful
weakening of calibration (see Foster and Hart 2021).

We start by recalling the definition of continuous calibration. A \emph{%
(fractional) binning} is a collection $\Pi =(w_{i})_{i\in I}$ of weight
functions $w_{i}:C\rightarrow \lbrack 0,1]$ for $i\in I$ such that $%
\sum_{i\in I}w_{i}(c)=1$ for all $c\in C$, where $I$ is a finite or
countably infinite set; the binning $\Pi $ is \emph{continuous} if all the $%
w_{i}$ are continuous functions on $C$. The interpretation is that at each
period $s$ the fraction $w_{i}(c_{s})$ of $z_{s}=a_{s}-c_{s}$ is assigned to
each bin $i$ in $I$. A deterministic forecasting procedure $\sigma $ is 
\emph{continuously calibrated} if%
\begin{equation}
\lim_{t\rightarrow \infty }\left( \sup_{\mathbf{a}_{t}}\mathcal{K}_{t}^{\Pi
}\right) =0  \label{eq:cont-calib}
\end{equation}%
for every continuous binning $\Pi $, where the $\Pi $\emph{-calibration}
score $\mathcal{K}_{t}^{\Pi }$ is\footnote{%
A more precise, but cumbersome, notation would be $\mathcal{K}^{\Pi (\mathbf{%
c})},$ since at each time $t$ the binning is given by $\Pi (c_{t}\mathbf{)}%
=(w_{i}(c_{t}))_{i\in I}.$}%
\begin{equation*}
\mathcal{K}_{t}^{\Pi }:=\sum_{i\in I}\left( \frac{n_{t}^{i}}{t}\right)
\left\Vert \bar{a}_{t}^{i}-\bar{c}_{t}^{i}\right\Vert ^{2},
\end{equation*}%
where, for each bin $i\in I,$ we let $n_{t}^{i}%
{\;:=\;}%
\sum_{s=1}^{t}w_{i}(c_{s})$ be the total weight of bin $i$, and $\bar{a}%
_{t}^{i}%
{\;:=\;}%
\sum_{s=1}^{t}(w_{i}(c_{s})/n_{t}^{i})a_{s}$ and $\bar{c}_{t}^{i}%
{\;:=\;}%
\sum_{s=1}^{t}(w_{i}(c_{s})/n_{t}^{i})c_{s}$ be the average action and
forecast there. Proposition 3 in Foster and Hart (2021) shows that it
suffices to require (\ref{eq:cont-calib}) for one specific continuous
binning $\Pi _{0}$; i.e., $\sigma $ is continuously calibrated if and only
if (\ref{eq:cont-calib}) holds for $\Pi =\Pi _{0}.$

To avoid confusion,\footnote{%
This confusion led to an error in previous versions of the paper; see the
errata Foster and Hart (2026).} when we deal with the variance of the
differences $a_{t}-c_{t}$ rather than the variance of the actions $a_{t}$,
we will from now on add a superscript $\#$ (see $v_{t}^{\#}$ and $\mathcal{R}%
_{t}^{\#}$ below).

Let $B$ be an arbitrary finite set\footnote{%
One may easily generalize to fractional $B$ binnings; also, $B$ could be
infinite when the binning is continuous (or, more generally, when the
binning is uniformly approximable by finite fractional binnings, as in (9)
in Foster and Hart 2021).} and let $\Pi =(w_{i})_{i\in I}$ be a fractional
binning. Consider the joint fractional binning with bins $U:=B\times I$,
where at each time $t$ the fractions $w_{i}(c_{t})$ of $a_{t}-c_{t}$ are
assigned to bins $(b_{t},i)$ for all $i\in I;$ that is, each bin $(b,i)\in
B\times I$ is assigned the fraction%
\begin{equation*}
\lambda _{t}(b,i):=\mathbf{1}_{b}(b_{t})w_{i}(c_{t}),
\end{equation*}%
where $\mathbf{1}_{x}$ stands for the $x$-indicator function (i.e., $\mathbf{%
1}_{x}(y)=1$ for $y=x$ and $\mathbf{1}_{x}(y)=0$ for $y\neq x$). Consider
bin $(b,i)$ at time $t;$ its weight, averages, and variance are,
respectively,%
\begin{eqnarray*}
n_{t}(b,i) &%
{\;:=\;}%
&\sum_{s=1}^{t}\lambda _{s}(b,i), \\
\bar{a}_{t}(b,i) &%
{\;:=\;}%
&\sum_{s=1}^{t}\left( \frac{\lambda _{s}(b,i)}{n_{t}(b,i)}\right) a_{s}, \\
\bar{c}_{t}(b,i) &%
{\;:=\;}%
&\sum_{s=1}^{t}\left( \frac{\lambda _{s}(b,i)}{n_{t}(b,i)}\right) c_{s},%
\text{\ \ \ and} \\
v_{t}^{\#}(b,i) &%
{\;:=\;}%
&\sum_{s=1}^{t}\left( \frac{\lambda _{s}(b,i)}{n_{t}(b,i)}\right) \left\Vert
a_{s}-c_{s}-\left( \bar{a}_{t}(b,i)-\bar{c}_{t}(b,i)\right) \right\Vert ^{2}.
\end{eqnarray*}%
The calibration and refinement scores are then 
\begin{eqnarray*}
\mathcal{K}_{t}^{\mathbf{b,}\Pi } &%
{\;:=\;}%
&\sum_{(b,i)\in B\times I}\left( \frac{n_{t}(b,i)}{t}\right) \left\Vert \bar{%
a}_{t}(b,i)-\bar{c}_{t}(b,i)\right\Vert ^{2}\;\text{\ \ and} \\
\mathcal{R}_{t}^{\#;\mathbf{b,}\Pi } &%
{\;:=\;}%
&\sum_{(b,i)\in B\times I}\left( \frac{n_{t}(b,i)}{t}\right) v_{t}^{\#}(b,i).
\end{eqnarray*}

Note that the result of Appendix \ref{s-a:refined}, namely, the refinement
score being monotonically decreasing with respect to refining the binning,
is stated for an arbitrary sequence $z_{t}$, and so holds for the two kinds
of refinement scores; thus, for instance, $\mathcal{R}_{t}^{\#;\mathbf{b,}%
\Pi }\leq \mathcal{R}_{t}^{\#;\Pi }$ and $\mathcal{R}_{t}^{\mathbf{b,}\Pi
}\leq \mathcal{R}_{t}^{\mathbf{b}}$. The general Brier score decomposition
of Appendix \ref{s-a:decompose} for a fractional binning $\Pi $ is now
written as $\mathcal{B}_{t}^{\mathbf{c}}=\mathcal{R}_{t}^{\#;\Pi }+\mathcal{K%
}_{t}^{\Pi }$.

We now state a more detailed version of Theorem \ref{th:cont} (see Section %
\ref{s:beat-by-calibrated}) on calibeating by a \emph{deterministic}
continuously calibrated procedure.

\begin{theorem}
\label{th:cont-calib}Let $B$ be a finite set. Then there exists a \emph{%
deterministic} $\mathbf{b}$-based forecasting procedure $\zeta $ that is $B$%
-calibeating and is continuously calibrated. Specifically: first, for every
continuous binning $\Pi $ there is a deterministic $\mathbf{b}$-based
forecasting procedure $\zeta \equiv \zeta _{\Pi }$ such that%
\begin{equation}
\mathcal{K}_{t}^{\mathbf{b,}\Pi }=\mathcal{B}_{t}^{\mathbf{c}}-\mathcal{R}%
_{t}^{\#;\mathbf{b,}\Pi }\leq o(1);  \label{eq:cont-calibeat}
\end{equation}%
and second, there is a continuous binning $\Pi ^{\ast }$ for which (\ref%
{eq:cont-calibeat}) implies that the corresponding procedure $\zeta _{\Pi
^{\ast }}$ is $B$-calibeating, i.e., 
\begin{equation*}
\mathcal{B}_{t}^{\mathbf{c}}\leq \mathcal{R}_{t}^{\mathbf{b}}+o(1),
\end{equation*}%
and is continuously calibrated. All the above inequalities hold as $%
t\rightarrow \infty $ uniformly over all sequences $\mathbf{a}$ and $\mathbf{%
b}.$
\end{theorem}

To prove this we use the corresponding \emph{online refinement} score $%
\widetilde{\mathcal{R}}_{t}^{\#;\mathbf{b},\Pi }$, in which the offline
averages $\bar{a}_{t}-\bar{c}_{t}$ are replaced with the online averages $%
\bar{a}_{s-1}-\bar{c}_{s-1}$; namely, 
\begin{equation*}
\widetilde{\mathcal{R}}_{t}^{\#;\mathbf{b},\Pi }%
{\;:=\;}%
\sum_{(b,i)\in B\times I}\sum_{s=1}^{t}\left( \frac{\lambda _{s}(b,i)}{t}%
\right) \left\Vert a_{s}-c_{s}-(\bar{a}_{s-1}(b,i)-\bar{c}%
_{s-1}(b,i))\right\Vert ^{2}.
\end{equation*}%
The parallel result to Proposition \ref{p:online-R} is

\begin{proposition}
\label{p:Rtilde-R-cont}For every finite set $B$ and every continuous binning 
$\Pi =(w_{i})_{i=1}^{I}$ on $C$, as $t\rightarrow \infty $ we have%
\begin{equation*}
0\leq \widetilde{\mathcal{R}}_{t}^{\#;\mathbf{b},\Pi }-\mathcal{R}_{t}^{\#;%
\mathbf{b},\Pi }\leq o(1)
\end{equation*}%
uniformly over all sequences $\mathbf{a}$, $\mathbf{b}$, and $\mathbf{c.}$
\end{proposition}

The proof is an adaptation of the proof of Proposition \ref{p:online-R} to
fractional binnings. We start by generalizing Proposition \ref{p:var} to
weighted variances. Let $(x_{n})_{n\geq 1}$ be a sequence of vectors in a
Euclidean space (or, more generally, in a normed vector space), let $%
(\lambda _{n})_{n\geq 1}$ be a sequence of weights in $[0,1]$, and let $%
\Lambda _{n}:=\sum_{i=1}^{n}\lambda _{i}$. Let $\bar{x}_{n}:=\sum_{i=1}^{n}(%
\lambda _{i}/\Lambda _{n})x_{i}$ denote the weighted average of $%
x_{1},...,x_{n}$ (when $\Lambda _{n}=0$, and hence $\lambda _{i}=0$ for all $%
i=1,...,n$, let $\lambda _{i}/\Lambda _{n}=0/0=0$).

\begin{proposition}
\label{p:var-lambda}For every $n\geq 1$ we have\footnote{%
The sum on the right-hand side of (\ref{eq:p-var-lambda}) effectively starts
from $i=2,$ and so, as in (\ref{eq:p-var}), it does not matter how $\bar{x}%
_{0}$ is defined.} 
\begin{equation}
\sum_{i=1}^{n}\lambda _{i}\left\Vert x_{i}-\bar{x}_{n}\right\Vert
^{2}=\sum_{i=1}^{n}\lambda _{i}\left( 1-\frac{\lambda _{i}}{\Lambda _{i}}%
\right) \left\Vert x_{i}-\bar{x}_{i-1}\right\Vert ^{2}.
\label{eq:p-var-lambda}
\end{equation}
\end{proposition}

\begin{proof}
Let $s_{n}:=\sum_{i=1}^{n}\lambda _{i}\left\Vert x_{i}-\bar{x}%
_{n}\right\Vert ^{2}$; we claim that%
\begin{equation}
s_{n}=s_{n-1}+\lambda _{n}\left( 1-\frac{\lambda _{n}}{\Lambda _{n}}\right)
\left\Vert x_{n}-\bar{x}_{n-1}\right\Vert ^{2}.  \label{eq:diff-var-lambda}
\end{equation}%
Indeed, assume that $\Lambda _{n}>0$ (otherwise both sides vanish) and $\bar{%
x}_{n-1}=0$ (without loss of generality, since subtracting a constant from
all the $x_{i}$ does not affect any term); then $\bar{x}_{n}=(\lambda
_{n}/\Lambda _{n})x_{n}$, and so, using $s_{n}=\sum_{i=1}^{n}\lambda
_{i}||x_{i}||^{2}-\Lambda _{n}||\bar{x}_{n}||^{2}$, we get%
\begin{equation*}
s_{n}-s_{n-1}=\left( \sum_{i=1}^{n}\lambda _{i}\left\Vert x_{i}\right\Vert
^{2}-\Lambda _{n}\left\Vert \frac{\lambda _{n}}{\Lambda _{n}}%
x_{n}\right\Vert ^{2}\right) -\sum_{i=1}^{n-1}\lambda _{i}\left\Vert
x_{i}\right\Vert ^{2}=\lambda _{n}\left\Vert x_{n}\right\Vert ^{2}-\frac{%
\lambda _{n}^{2}}{\Lambda _{n}}\left\Vert x_{n}\right\Vert ^{2},
\end{equation*}%
which is precisely $\lambda _{n}\left( 1-\lambda _{n}/\Lambda _{n}\right)
\left\Vert x_{n}-\bar{x}_{n-1}\right\Vert ^{2}$.

Applying (\ref{eq:diff-var-lambda}) recursively yields the result.
\end{proof}

\bigskip

Let $v_{n}:=(1/\Lambda _{n})\sum_{i=1}^{n}\lambda _{i}\left\Vert x_{i}-\bar{x%
}_{n}\right\Vert ^{2}$ denote the weighted variance of $x_{1},...,x_{n}$,
and let $\widetilde{v}_{n}:=(1/\Lambda _{n})\sum_{i=1}^{n}\lambda
_{i}\left\Vert x_{i}-\bar{x}_{i-1}\right\Vert ^{2}$ be the corresponding 
\emph{online weighted variance} of $x_{1},...,x_{n}$ (again, take $\bar{x}%
_{0}$ to be an arbitrary element of the convex hull of the $x_{i}$).
Proposition \ref{p:var-lambda} gives $\widetilde{v}_{n}-v_{n}=(1/\Lambda
_{n})\sum_{i=1}^{n}(\lambda _{i}^{2}/\Lambda _{i})\left\Vert x_{i}-\bar{x}%
_{i-1}\right\Vert ^{2}$, and so, by inequality (22) in Foster and Hart
(2021),%
\begin{equation}
0\leq \widetilde{v}_{n}-v_{n}\leq \xi ^{2}\frac{1}{\Lambda _{n}}%
\sum_{i=1}^{n}\frac{\lambda _{i}^{2}}{\Lambda _{i}}\leq \xi ^{2}\frac{\ln
\Lambda _{n}+2}{\Lambda _{n}},  \label{eq:v-tilda-lambda}
\end{equation}%
where $\xi :=\max_{1\leq i,j\leq n}\left\Vert x_{i}-x_{j}\right\Vert .$

\bigskip

We now prove Proposition \ref{p:Rtilde-R-cont}, which shows that the online
refinement score $\widetilde{\mathcal{R}}_{t}^{\#;\mathbf{b},\Pi }$ is close
to the (offline) refinement score.

\bigskip

\begin{proof}[Proof of Proposition \protect\ref{p:Rtilde-R-cont}]
We have $\widetilde{\mathcal{R}}_{t}^{\#;\mathbf{b},\Pi }-\mathcal{R}%
_{t}^{\#;\mathbf{b},\Pi }=(1/t)\sum_{b\in B}\sum_{i\in I}\mu _{t}(b,i)$,
where%
\begin{equation*}
\mu _{t}(b,i):=\sum_{s=1}^{t}\lambda _{s}(b,i)\left\Vert
a_{s}-c_{s}-e_{s-1}(b,i)\right\Vert ^{2}-\sum_{s=1}^{t}\lambda
_{s}(b,i)\left\Vert a_{s}-c_{s}-e_{t}(b,i)\right\Vert ^{2}
\end{equation*}%
with $e_{t}(b,i):=\bar{a}_{t}(b,i)-\bar{c}_{t}(b,i)$, for each $(b,i)\in
B\times I$. Proposition \ref{p:var-lambda}, specifically, (\ref%
{eq:v-tilda-lambda}), yields%
\begin{equation}
0\leq \mu _{t}(b,i)\leq 4\gamma ^{2}(\ln n_{t}(b,i)+2)\leq 4\gamma ^{2}(\ln
t+2)  \label{eq:i<k}
\end{equation}%
(because $\left\Vert a-c-e\right\Vert \leq 2\gamma $---since $\left\Vert
a-c\right\Vert \leq \gamma $ and so $\left\Vert e\right\Vert \leq \gamma $%
---and $n_{t}(b,i)\leq t$). For each finite $J\subseteq I$, summing over all 
$(b,i)$ in $B\times J$ yields%
\begin{equation}
0\leq \frac{1}{t}\sum_{b\in B}\sum_{i\in J}\mu _{t}(b,i)\leq 4\gamma
^{2}|B|\,\left\vert J\right\vert \,\frac{\ln t+2}{t}.  \label{eq:J}
\end{equation}

When $I$ is finite we are thus done. When $I$ is infinite, for every $%
\varepsilon >0$ there is a finite $J\subset I$ such that $\sum_{i\in
I\backslash J}w_{i}(c)\leq \varepsilon $ for all $c\in C;$ such a finite $J$
exists by Dini's theorem (see (9) in Foster and Hart 2021). For $i\in
I\backslash J$ we get%
\begin{eqnarray*}
0 &\leq &\frac{1}{t}\sum_{b\in B}\sum_{i\in I\backslash J}\mu _{t}(b,i)\leq 
\frac{1}{t}\sum_{b\in B}\sum_{i\in I\backslash J}\sum_{s=1}^{t}\lambda
_{s}(b,i)\left\Vert a_{s}-c_{s}-e_{s-1}(b,i)\right\Vert ^{2} \\
&\leq &4\gamma ^{2}\frac{1}{t}\sum_{s=1}^{t}\sum_{i\in I\backslash J}\lambda
_{s}(b_{s},i)\leq 4\gamma ^{2}\frac{1}{t}\sum_{s=1}^{t}\varepsilon =4\gamma
^{2}\varepsilon .
\end{eqnarray*}%
Adding this to (\ref{eq:J}) yields%
\begin{equation*}
0\leq \widetilde{\mathcal{R}}_{t}^{\#;\mathbf{b},\Pi }-\mathcal{R}_{t}^{\#;%
\mathbf{b},\Pi }\leq 4\gamma ^{2}|B|\,\left\vert J\right\vert \,\frac{\ln t+2%
}{t}+4\gamma ^{2}\varepsilon ,
\end{equation*}%
which is less than, say, $5\gamma ^{2}\varepsilon $ for all large enough $t$%
. The result follows since $\varepsilon $ was arbitrary; moreover, all the
above inequalities are uniform over all sequences $\mathbf{a}$, $\mathbf{b}$%
, and $\mathbf{c}$.
\end{proof}

\bigskip 

Next, we need to show that the two kinds of refinements---$\mathcal{R}^{\#}$%
, the average within-bin variance of $a_{t}-c_{t}$, and $\mathcal{R}$, the
average within-bin variance of $a_{t}$---are close when the binning is
\textquotedblleft local," in the sense that the forecasts within each bin
are almost constant.

A weight function $w:C\rightarrow \lbrack 0,1]$ is $\delta $\emph{-local}
(for $\delta >0$) if all points with positive weight (the \emph{support} of $%
w$) lie in an open ball of radius $\delta $; i.e., $\{c\in
C:w(c)>0\}\subseteq B(y;\delta )$ for some $y\in C$. A fractional binning $%
\Pi =(w_{i})_{i\in I}$ is $\delta $\emph{-local} if all the weight functions 
$w_{i}$ are $\delta $-local. Such a fractional binning $\Pi \equiv \Pi
_{\delta }$ can be obtained, for instance, by using the so-called
\textquotedblleft $\delta $-tent functions" based on a $\delta $-grid: let $%
C_{\delta }=\{y^{i}\}_{i\in I}$ be a finite $\delta $-grid of $C$, and for
each $i\in I$ let $w_{i}(c):=\Lambda (c,y^{i})/\sum_{j\in I}\Lambda
(c,y_{j}),$ where $\Lambda (c,y):=[\delta -\left\Vert c-y\right\Vert ]_{+}$;
moreover, this construction yields a continuous binning (because the
functions $\Lambda (\cdot ,y^{i})$ are continuous, and their sum is
positive, and thus bounded away from $0$ on the compact set $C$).

\begin{lemma}
\label{l:Rsharp-R}\emph{Let }$\Pi $\emph{\ be a }$\delta $-local \emph{%
fractional binning. Then}%
\begin{equation*}
\left\vert \mathcal{R}_{t}^{\#;\Pi }-\mathcal{R}_{t}^{\Pi }\right\vert
<2\delta +\delta ^{2}.
\end{equation*}
\end{lemma}

\begin{proof}
Let $\Pi =(w_{i})_{i\in I},$ and for each $i\in I$ let $y^{i}\in C$ be such
that $B(y^{i};\delta )$ contains the support of $w_{i}$. Thus, all forecasts 
$c_{s}$ counted in bin $i$ (i.e., with $\lambda _{s}(i)=w_{i}(c_{s})>0$)
satisfy $\left\Vert c_{s}-y^{i}\right\Vert <\delta $.

Given random variables $X$ and $Y$ such that $\left\Vert X\right\Vert \leq 1$
and $\left\Vert Y\right\Vert <\delta $, and hence $\mathbb{V}ar[X]\leq 1$
and $\mathbb{V}ar[Y]<\delta ^{2}$, we have%
\begin{equation*}
\left\vert \mathbb{V}ar[X-Y]-\mathbb{V}ar[X]\right\vert \leq 2\sqrt{\mathbb{V%
}ar[X]\,\mathbb{V}ar[Y]}+\mathbb{V}ar[Y]<2\delta +\delta ^{2}.
\end{equation*}%
Applying this to $X=a_{s}$ and $Y=c_{s}-y^{i}$ (since $y^{i}$ is a constant
it does not affect the variance) yields%
\begin{equation*}
\left\vert v_{t}^{\#}(i)-v_{t}(i)\right\vert <2\delta +\delta ^{2},
\end{equation*}%
where 
\begin{eqnarray*}
v_{t}^{\#}(i) &=&\sum_{s=1}^{t}\left( \frac{\lambda _{s}(i)}{n_{t}(i)}%
\right) \left\Vert a_{s}-c_{s}-(\bar{a}_{t}(i)-\bar{c}_{t}(i))\right\Vert
^{2}\;\;\;\text{and} \\
v_{t}(i) &=&\sum_{s=1}^{t}\left( \frac{\lambda _{s}(i)}{n_{t}(i)}\right)
\left\Vert a_{s}-\bar{a}_{t}(i)\right\Vert ^{2}
\end{eqnarray*}%
are the variances in bin $i$ of the differences $a_{s}-c_{s}$ and of the
actions $a_{s}$, respectively. Averaging over $i\in I$, with weights $%
n_{t}(i)/t$, gives the result.
\end{proof}

\bigskip

We can now prove Theorem \ref{th:cont-calib}.

\bigskip 

\begin{proof}[Proof of Theorem \protect\ref{th:cont-calib}]
(i) Let $\Pi =(w_{i})_{i\in I}$ be a continuous binning. At time $t$, given $%
\mathbf{a}_{t-1},\mathbf{c}_{t-1}$, and $\mathbf{b}_{t}$, applying the
outgoing fixed point result, specifically, Theorem \ref{th:outgoing} (D), to
the continuous function $c\longmapsto c+\sum_{i\in I}w_{i}(c)e_{t-1}(b_{t},i)
$, yields $c_{t}\in C$ such that 
\begin{eqnarray*}
\left\Vert a_{t}-c_{t}\right\Vert ^{2} &\leq &\left\Vert
a_{t}-c_{t}-\sum_{i\in I}w_{i}(c_{t})e_{t-1}(b_{t},i)\right\Vert ^{2} \\
&\leq &\sum_{i\in I}w_{i}(c_{t})\left\Vert
a_{t}-c_{t}-e_{t-1}(b_{t},i)\right\Vert ^{2}
\end{eqnarray*}%
for every $a_{t}\in A$ (the second inequality is by the convexity of $%
\left\Vert \cdot \right\Vert ^{2}$). Averaging over $t$ gives $\mathcal{B}%
_{t}^{\mathbf{c}}\leq \widetilde{\mathcal{R}}_{t}^{\#;\mathbf{b},\Pi }$, and
thus (\ref{eq:cont-calibeat}) by Proposition \ref{p:Rtilde-R-cont} together
with the decomposition of Appendix A.5 for the fractional binning $(\mathbf{b%
},\Pi )$.

(ii) We construct $\Pi ^{\ast }$ as follows. Let $\Pi _{0}=(w_{i}^{0})_{i\in
I_{0}}$ be the continuous binning given by Proposition 3 in Foster and Hart
(2021). For each $n\geq 1$ let $\Pi _{n}=(w_{i}^{n})_{i\in I_{n}}$ be a $%
\delta _{n}$-local continuous binning, where $0<\delta _{n}<1$ and $\delta
_{n}\rightarrow 0$ as $n\rightarrow \infty $; assume that the indexing sets $%
I_{n}$ for all $n\geq 0$ are taken to be disjoint. The collection of all
these weight functions, with $w_{i}^{n}$ rescaled by a factor of $1/2^{n+1}$%
, yields a continuous binning, which we denote by $\Pi ^{\ast }$; i.e., $\Pi
^{\ast }=(w_{i}^{\ast })_{i\in I^{\ast }}$ with $I^{\ast }:=\cup _{n\geq
0}I_{n}$ and $w_{i}^{\ast }:=(1/2^{n+1})w_{i}^{n}$ for each $i\in I_{n}$ and 
$n\geq 0$ (indeed, each $w_{i}^{\ast }$ is continuous, and $\sum_{i\in
I}w_{i}^{\ast }=\sum_{n\geq 0}\sum_{i\in
I_{n}}(1/2^{n+1})w_{i}^{n}=\sum_{n\geq 0}\left( 1/2^{n+1}\right) \mathbf{1}=%
\mathbf{1},$ because $\sum_{i\in I_{n}}w_{i}^{n}=\mathbf{1}$ for each $n$).

Consider now the joint fractional binning $(\mathbf{b},\Pi ^{\ast })$, and
the corresponding refinement score $\mathcal{R}_{t}^{\#;\mathbf{b},\Pi
^{\ast }}$. Separating the sum over all $i\in I$ into sums over $i\in I_{n}$
for all $n\geq 0$, and noting that rescaling a weight function does not
affect the within-bin \emph{relative} weights $\lambda _{s}(b,i)/n_{t}(b,i)=%
\mathbf{1}_{b}(b_{s})\cdot w_{i}(c_{s})/\sum_{r\leq t}w_{i}(c_{r})$ that are
used to compute the bin averages and variances, we get%
\begin{equation*}
\mathcal{R}_{t}^{\#;\mathbf{b},\Pi ^{\ast }}=\sum_{n\geq 0}\frac{1}{2^{n+1}}%
\mathcal{R}_{t}^{\#;\mathbf{b},\Pi _{n}}.
\end{equation*}%
The procedure $\zeta _{\Pi ^{\ast }}$ constructed in (i) for $\Pi ^{\ast }$
yields, by (\ref{eq:cont-calibeat}),%
\begin{equation*}
\mathcal{B}_{t}^{\mathbf{c}}-\mathcal{R}_{t}^{\#;\mathbf{b},\Pi ^{\ast
}}=\sum_{n\geq 0}\frac{1}{2^{n+1}}\left( \mathcal{B}_{t}^{\mathbf{c}}-%
\mathcal{R}_{t}^{\#;\mathbf{b},\Pi _{n}}\right) \leq o(1).
\end{equation*}%
By the decomposition of Appendix \ref{s-a:decompose} for $(\mathbf{b},\Pi
_{n})$, each term in the above sum equals $\mathcal{K}_{t}^{\mathbf{b},\Pi
_{n}}$, and is thus nonnegative, and so%
\begin{equation}
\mathcal{K}_{t}^{\mathbf{b},\Pi _{n}}=\mathcal{B}_{t}^{\mathbf{c}}-\mathcal{R%
}_{t}^{\#;\mathbf{b,}\Pi _{n}}\leq o(1)\text{\ \ \ for every }n\geq 0.
\label{eq:Pi-n}
\end{equation}

For $n=0$, (\ref{eq:Pi-n}) implies that 
\begin{equation*}
\mathcal{K}_{t}^{\Pi _{0}}=\mathcal{B}_{t}^{\mathbf{c}}-\mathcal{R}%
_{t}^{\#;\Pi _{0}}\leq \mathcal{B}_{t}^{\mathbf{c}}-\mathcal{R}_{t}^{\#;%
\mathbf{b,}\Pi _{0}}\leq o(1),
\end{equation*}%
where we have used the decomposition of Appendix \ref{s-a:decompose} for $%
\Pi _{0}$, and $\mathcal{R}_{t}^{\#;\mathbf{b,}\Pi _{0}}\leq \mathcal{R}%
_{t}^{\#;\Pi _{0}}$ by the refining monotonicity of the refinement score
(Appendix \ref{s-a:refined}). Thus, by Proposition 3 in Foster and Hart
(2021), the procedure is continuously calibrated.

Next, for each $n\geq 1$, (\ref{eq:Pi-n}) implies that%
\begin{equation*}
\mathcal{B}_{t}^{\mathbf{c}}\leq \mathcal{R}_{t}^{\#;\mathbf{b},\Pi
_{n}}+o(1)\leq \mathcal{R}_{t}^{\mathbf{b,}\Pi _{n}}+3\delta _{n}+o(1)\leq 
\mathcal{R}_{t}^{\mathbf{b}}+3\delta _{n}+o(1),
\end{equation*}%
where we have used $\mathcal{R}_{t}^{\#;\mathbf{b},\Pi _{n}}\leq \mathcal{R}%
_{t}^{\mathbf{b},\Pi _{n}}+2\delta _{n}+\delta _{n}^{2}<\mathcal{R}_{t}^{%
\mathbf{b},\Pi _{n}}+3\delta _{n}$ by Lemma \ref{l:Rsharp-R} applied to the
fractional binning $(\mathbf{b},\Pi _{n})$, which is $\delta _{n}$-local,
and $\mathcal{R}_{t}^{\mathbf{b},\Pi _{n}}\leq \mathcal{R}_{t}^{\mathbf{b}}$
(again by Appendix \ref{s-a:refined}). Therefore $\mathcal{B}_{t}^{\mathbf{c}%
}\leq \mathcal{R}_{t}^{\mathbf{b}}+4\delta _{n}$ for all $t$ large enough;
since $\delta _{n}\rightarrow 0$, this proves that $\mathcal{B}_{t}^{\mathbf{%
c}}-\mathcal{R}_{t}^{\mathbf{b}}\leq o(1).$
\end{proof}

\subsection{Multi-calibeating: Improved Error Terms\label{s-a:multi}}

The multi-calibeating procedure of Section \ref{s:multi-beat} yields an
error that is proportional to the product of the sizes of the sets $B^{n}$,
i.e., $\prod_{n=1}^{N}\left\vert B^{n}\right\vert $, which increases
exponentially with $N$. We provide in this appendix two approaches that
yield better errors terms: $O(\sqrt{N}/\sqrt{t})$ by a Blackwell
approachability construct in Section \ref{sus-a:blackwell} (see (\ref%
{eq:blackwell})), and $O((\max_{n}|B^{n}|+N)\log t/t)$ by an online linear
regression construct in Section \ref{sus-a:online LR} (see (\ref{eq:aw-error}%
)). For large $t$, an $O(1/\sqrt{t})$ error is of course worse than an $%
O(\log t/t)$ error; however, for small $t$ the former may well be smaller
than the latter.

We will use throughout the superscript $n$ instead of the more cumbersome $%
\mathbf{b}^{n}$, e.g., $\mathcal{R}_{t}^{n}$ for $\mathcal{R}_{t}^{\mathbf{b}%
^{n}}$, and $\bar{a}_{t-1}^{n}(b_{t}^{n})$ for $\bar{a}_{t-1}^{\mathbf{b}%
^{n}}(b_{t}^{n}).$

\subsubsection{A Blackwell Approachability Approach\label{sus-a:blackwell}}

We use here a construct along the lines of the vector approachability of
Blackwell (1950), with continuous actions taking the place of mixed actions
(which dispenses with the use of probabilities and laws of large numbers).

Put $x_{t}^{n}:=\left\Vert a_{t}-c_{t}\right\Vert ^{2}-\left\Vert a_{t}-\bar{%
a}_{t-1}^{n}(b_{t}^{n})\right\Vert ^{2};$ then%
\begin{equation}
\bar{x}_{t}^{n}%
{\;:=\;}%
\frac{1}{t}\sum_{s=1}^{t}x_{s}^{n}=\mathcal{B}_{t}^{\mathbf{c}}-\widetilde{%
\mathcal{R}}_{t}^{n}.  \label{eq:ui}
\end{equation}%
Let $x_{t}:=(x_{t}^{n})_{n=1,...,N}\in \mathbb{R}^{N}$ and $\bar{x}_{t}:=(%
\bar{x}_{t}^{i})_{n=1,...,N}\in \mathbb{R}^{N}$ be the corresponding $N$%
-dimensional vectors; given the history $(\mathbf{a}_{t-1},\mathbf{c}_{t-1},%
\mathbf{b}_{t})$, the vector $x_{t}$ is determined by $c_{t}$ and $a_{t}$.
We will show that the negative orthant\footnote{%
Notation: $\mathbb{R}_{+}^{N}=\{x\in \mathbb{R}^{N}:x\geq 0\}$ and $\mathbb{R%
}_{-}^{N}=\{x\in \mathbb{R}^{N}:x\leq 0\}$; for real $x\in \mathbb{R},$ put $%
\left[ x\right] _{+}=\max \{x,0\}$ and $\left[ x\right] _{-}=\min \{x,0\};$
for a vector $x\in \mathbb{R}^{N},$ put $[x]_{+}=\left(
[x_{1}]_{+},...,[x_{N}]_{+}\right) $, and $[x]_{-}=\left( \left[ x_{1}\right]
_{-},...,\left[ x_{N}\right] _{-}\right) $.} $\mathbb{R}_{-}^{N}$ of $%
\mathbb{R}^{N}$ is approachable by the $c$-player; i.e., there is a $\mathbf{%
b}$-based forecasting procedure $\zeta $ such that $\sup_{\mathbf{a}_{t},%
\mathbf{b}_{t}}\mathrm{dist}(\bar{x}_{t},\mathbb{R}_{-}^{N})\rightarrow 0$
as $t\rightarrow \infty .$

To this end we claim that for every $\lambda \in \mathbb{R}_{+}^{N}$ there
is $c_{t}\in C$ such that%
\begin{equation}
\lambda \cdot x_{t}\leq 0\;\;\text{for every }a_{t}\in A
\label{eq:blackwell-cond}
\end{equation}%
(this is the Blackwell condition here). This of course holds when $\lambda
=0;$ otherwise, assuming without loss of generality that $%
\sum_{n=1}^{N}\lambda _{n}=1$ (rescale $\lambda $ as needed), we have 
\begin{eqnarray*}
\lambda \cdot x_{t} &=&\sum_{n=1}^{N}\lambda _{n}\left( \left\Vert
a_{t}-c_{t}\right\Vert ^{2}-\left\Vert a_{t}-\bar{a}_{t-1}^{n}(b_{t}^{n})%
\right\Vert ^{2}\right) \\
&=&\left\Vert a_{t}-c_{t}\right\Vert ^{2}-\sum_{n=1}^{N}\lambda
_{n}\left\Vert a_{t}-\bar{a}_{t-1}^{n}(b_{t}^{n})\right\Vert ^{2} \\
&\leq &\left\Vert a_{t}-c_{t}\right\Vert ^{2}-\left\Vert
a_{t}-\sum_{n=1}^{N}\lambda _{n}\bar{a}_{t-1}^{n}(b_{t}^{n})\right\Vert ^{2}
\end{eqnarray*}%
(the inequality is by the convexity of $\left\Vert \cdot \right\Vert ^{2}$),
and so, by taking $c_{t}=\sum_{n=1}^{N}\lambda _{n}\bar{a}%
_{t-1}^{n}(b_{t}^{n})$, one guarantees that the final expression vanishes,
and thus $\lambda \cdot x_{t}\leq 0$, for any $a_{t}.$

Let $\zeta $ be the procedure whereby at time $t$ one chooses $c_{t}\in C$
so as to guarantee $[\bar{x}_{t-1}]_{+}\cdot x_{t}\leq 0$ for all $a_{t}\in
A $ (i.e., condition (\ref{eq:blackwell-cond}) for $\lambda =[\bar{x}%
_{t-1}]_{+}$); thus, $c_{t}$ is arbitrary when $\bar{x}_{t-1}\leq 0$ (i.e., $%
\left[ \bar{x}_{t-1}\right] _{+}=0$), and is otherwise given by%
\begin{equation*}
c_{t}=\frac{\sum_{n=1}^{N}\left[ \bar{x}_{t-1}^{n}\right] _{+}\bar{a}%
_{t-1}^{n}(b_{t}^{n})}{\sum_{n=1}^{N}\left[ \bar{x}_{t-1}^{n}\right] _{+}}.
\end{equation*}

Putting $X_{t}:=t\bar{x}_{t}$, we have%
\begin{eqnarray*}
\mathrm{dist}^{2}(X_{t},\mathbb{R}_{-}^{N}) &\leq &\left\Vert
(x_{t}+X_{t-1})-[X_{t-1}]_{-}\right\Vert ^{2}=\left\Vert
x_{t}+[X_{t-1}]_{+}\right\Vert ^{2} \\
&=&\left\Vert [X_{t-1}]_{+}\right\Vert ^{2}+2[X_{t-1}]_{+}\cdot
x_{t}+\left\Vert x_{t}\right\Vert ^{2} \\
&\leq &\mathrm{dist}^{2}(X_{t-1},\mathbb{R}_{-}^{N})+\gamma ^{4}N,
\end{eqnarray*}%
where the first inequality is by $\mathrm{dist}(X_{t},\mathbb{R}%
_{-}^{N})\leq \left\Vert X_{t}-[X_{t-1}]_{-}\right\Vert $ (since $%
[X_{t-1}]_{-}\in \mathbb{R}_{-}^{N}$), and the second inequality is by the
choice of $c_{t}$ (since $[X_{t-1}]_{+}=(t-1)[\bar{x}_{t-1}]_{+})$ for the
middle term, and $\left\vert x_{t}^{n}\right\vert \leq \gamma ^{2}$ (since $%
a_{t},c_{t},\bar{a}_{t-1}^{n}(b_{t})$ are all in $C)$ for all $n$ for the
third term. Applying this recursively yields%
\begin{equation*}
t^{2}\mathrm{dist}^{2}(\bar{x}_{t},\mathbb{R}_{-}^{N})=\mathrm{dist}%
^{2}(X_{t},\mathbb{R}_{-}^{N})\leq (\gamma ^{4}N)t,
\end{equation*}%
and so $\mathrm{dist}^{2}(\bar{x}_{t},\mathbb{R}_{-}^{N})\leq (\gamma
^{4}N)/t$, which gives%
\begin{equation*}
\max_{1\leq n\leq N}\bar{x}_{t}^{n}\leq \max_{1\leq n\leq N}\left[ \bar{x}%
_{t}^{n}\right] _{+}\leq \left\Vert \left[ \bar{x}_{t}\right]
_{+}\right\Vert =\mathrm{dist}(\bar{x}_{t},\mathbb{R}_{-}^{N})\leq \gamma
^{2}\sqrt{N}\frac{1}{\sqrt{t}}.
\end{equation*}%
By (\ref{eq:ui}) and Proposition \ref{p:online-R} we get%
\begin{equation}
\max_{1\leq n\leq N}\left( \mathcal{B}_{t}^{\mathbf{c}}-\mathcal{R}%
_{t}^{n}\right) \leq \gamma ^{2}\sqrt{N}\frac{1}{\sqrt{t}}+\gamma
^{2}\max_{1\leq n\leq N}|B^{n}|\frac{\ln t+1}{t},  \label{eq:blackwell}
\end{equation}%
which is $\sim \gamma ^{2}\sqrt{N}/\sqrt{t}$ as $t\rightarrow \infty .$

\subsubsection{An Online Linear Regression Approach\label{sus-a:online LR}}

Assume that $C\subseteq \lbrack -\gamma _{0},\gamma _{0}]^{m}$ (see the
remark below on the relation between $\gamma $ and $\gamma _{0}$), and put $%
c_{s}=(c_{i,s})_{i=1,...,m}\in C$ and $a_{s}=(a_{i,s})_{i=1,...,m.}\in A$.
For each $t\geq 1$ let $x_{i,t}^{n}:=\bar{a}_{i,t-1}^{n}(b_{t}^{n})$ (this
is the average of the $i$th coordinates of $a_{s}$ over all periods $s\leq t$
in which $b_{s}^{n}=b_{t}^{n}$), and put $x_{i,t}:=(x_{i,t}^{n})_{n=1,...,N}%
\in \mathbb{R}^{N}.$

For each coordinate $i=1,...,m$, consider the linear regression problem,
regularized by adding the strictly convex term $\alpha \left\Vert \theta
\right\Vert ^{2}$ for some $\alpha >0$, of minimizing%
\begin{equation*}
\mathcal{F}_{i,t}^{{}}(\theta ):=\frac{1}{t}\left(
\sum_{s=1}^{t}(a_{i,s}-\theta \cdot x_{i,s})^{2}+\alpha \left\Vert \theta
\right\Vert ^{2}\right)
\end{equation*}%
over $\theta \in \mathbb{R}^{N};$ let $\mathcal{F}_{i,t}^{\ast }$ denote
this minimum. For each $n=1,...,N$, when $\theta $ equals the $n$th unit
vector $e^{n}\in \mathbb{R}^{N}$, we have 
\begin{equation*}
\mathcal{F}_{i,t}(e^{n})=\frac{1}{t}\sum_{s\leq t}(a_{i,s}^{{}}-\bar{a}%
_{i,s-1}^{n}(b_{s}^{n}))^{2}+\frac{\alpha }{t};
\end{equation*}%
summing over $i=1,...,m$ we get%
\begin{equation}
\sum_{i=1}^{m}\mathcal{F}_{i,t}^{\ast }\leq \sum_{i=1}^{m}\mathcal{F}%
_{i,t}(e^{n})\leq \widetilde{\mathcal{R}}_{t}^{n}+\frac{m\alpha }{t}.
\label{eq:L<=R}
\end{equation}

The \textquotedblleft forward algorithm" of Azoury and Warmuth (2001)
applied to each coordinate $i$ separately yields an online procedure that
generates at each time $t$ a vector $\theta _{i,t}\in \mathbb{R}^{N}$ (that
depends on the history $a_{i,1},...,a_{i,t-1}$ and $x_{i,1},...,x_{i,t-1}$
as well as on $x_{i,t})$ such that%
\begin{equation}
\sum_{s=1}^{t}(a_{i,s}-\theta _{i,s}\cdot x_{i,s})^{2}\leq t\mathcal{F}%
_{i,t}^{\ast }+\gamma _{0}N\ln \left( \frac{\gamma _{0}}{\alpha }t+1\right)
\label{eq:aw}
\end{equation}%
is guaranteed for any sequence\footnote{%
This is Theorem 5.6 of Azoury and Warmuth (2001); in the notation there, $%
X=\max_{n,s}\left\vert x_{i,s}^{n}\right\vert \leq \gamma _{0}$ and $%
Y=\max_{s}\left\vert a_{i,s}\right\vert \leq \gamma _{0}$. Note that there
is a misprinted sign in the first line of the formula (5.17) there.} $%
(a_{i,s})_{s\geq 1}$.

Combining these $m$ algorithms yields an online $\mathbf{b}$-based procedure
(because $x_{t}$ is determined by $b_{t}$ and the history); the vectors $%
\theta _{i,s}$ for $i=1,...,m$ together yield a point $\hat{c}_{s}:=(\theta
_{i,s}\cdot x_{i,s})_{i=1,...,m}\in \mathbb{R}^{m}$. Let $c_{s}:=\mathrm{proj%
}_{C}(\hat{c}_{s})$ be the closest point to $\hat{c}_{s}$ in $C$ (it is well
defined since $C$ is a nonempty convex compact set); then any point in $C$,
in particular $a_{s}$, is closer to $c_{s}$ than to $\hat{c}_{s}$, which
yields 
\begin{equation*}
\left\Vert a_{s}-c_{s}\right\Vert ^{2}\leq \left\Vert a_{s}-\hat{c}%
_{s}\right\Vert ^{2}=\sum_{i=1}^{m}(a_{i,s}-\theta _{i,s}\cdot x_{i,s})^{2}.
\end{equation*}%
Averaging over $s\leq t$ gives%
\begin{eqnarray*}
\mathcal{B}_{t}^{\mathbf{c}} &=&\frac{1}{t}\sum_{s=1}^{t}\left\Vert
a_{s}-c_{s}\right\Vert ^{2}\leq \frac{1}{t}\sum_{i=1}^{m}%
\sum_{s=1}^{t}(a_{i,s}-\theta _{i,s}\cdot x_{i,s})^{2} \\
&\leq &\sum_{i=1}^{m}\mathcal{F}_{i,t}^{\ast }+\frac{m\gamma _{0}N}{t}\ln
\left( \frac{\gamma _{0}}{\alpha }t+1\right)
\end{eqnarray*}%
by (\ref{eq:aw}). Recalling (\ref{eq:L<=R}) and Proposition \ref{p:online-R}
yields%
\begin{equation}
\mathcal{B}_{t}^{\mathbf{c}}-\mathcal{R}_{t}^{n}\leq \frac{m\gamma _{0}N}{t}%
\ln \left( \frac{\gamma _{0}}{\alpha }t+1\right) +\frac{m\alpha }{t}+\gamma
^{2}|B^{n}|\frac{\ln t+1}{t}=O\left( \frac{\log t}{t}\right) .
\label{eq:aw-error}
\end{equation}

\bigskip

\noindent \textbf{Remarks. }\emph{(a) }A connection between $\gamma _{0}$
and $\gamma $ is as follows. The set $C\subset \mathbb{R}^{m}$, whose
diameter is $\gamma $, can be enclosed in a ball of radius $r$, where $%
\gamma /2\leq r\leq \gamma \sqrt{m/(2m+2)}$ by Jung's (1901) theorem. Since
translating the set $C$ does not matter (only differences $a-c$ do), we can
assume without loss of generality that $C\subseteq \bar{B}(0;r)\subseteq
\lbrack -r,r]^{m}$, and so we can take $\gamma _{0}=r.$

\emph{(b)} For a fixed horizon $t$ one may optimize $\alpha .$

\emph{(c)} The forecast $c_{t}$ at time $t$ of the above construction is
given by the formula%
\begin{equation*}
c_{i,t}=\sum_{n=1}^{N}\theta _{i,t}^{n}\bar{a}_{i,t-1}^{n}(b_{t}^{n}),
\end{equation*}%
where $\theta _{i,t}$ is the minimizer of $\mathcal{F}_{i,t}(\theta )$ as if
we have $a_{i,t}=0$ (the actual $a_{t}$ is not known at this point); see
Azoury and Warmuth (2001) for details and more explicit formulas.

\emph{(d) }Since we use the inequalities $\mathcal{F}_{i,t}^{\ast }\leq 
\mathcal{F}_{i,t}(\theta )$ only for $\theta $ equal to the unit vectors $%
e^{n}$ in $\mathbb{R}^{N}$, it suffices to minimize $\mathcal{F}%
_{i,t}(\theta )$ over the convex hull of these vectors, that is, over the
unit simplex $\Delta (N)$ of $\mathbb{R}^{N}$, as in Foster (1991) (whose
result would need to be generalized from the one-dimensional case of $%
A=\{0,1\}$ and $C=[0,1]$ to a general $C$); note that multi-calibeating is
equivalent to being, in terms of the Brier scores, \textquotedblleft as
strong as" each one of the $N$ sequences $(\bar{a}_{t-1}^{n}(b_{t}^{n}))_{t%
\geq 1}$ (for $n=1,...,N$).

\emph{(e)} An alternative approach is to first calibeat each forecaster
separately (by Theorem \ref{th:beat-1}) and then to combine these $N$
calibeating forecasters (by a method such as Azoury and Warmuth's 2001).

\subsection{Log-calibeating\label{s-a:log}}

In this appendix we show that our approach to calibeating works for another
classic scoring rule, namely, the logarithmic one; like the quadratic Brier
score, it is also a strictly proper scoring rule. While the analysis here is
parallel to that of the quadratic scores, some of the technical details
require a little more work here (one reason being the infinite slope of the
logarithmic scores on the boundary).

Let the action space $A$ and the forecast space $C$ both be the unit simplex
in $\mathbb{R}^{m}$; i.e., $A=C=\Delta :=\{x\in \mathbb{R}%
_{+}^{m}:\sum_{i}x_{i}=1\}$. For $x$ and $y$ in $\Delta ,$ the \emph{%
relative entropy} (or \emph{Kullback--Leibler divergence}) of $y$ with
respect to $x$ is defined as%
\begin{equation*}
\mathcal{D}(x\parallel y)%
{\;:=\;}%
\sum_{i=1}^{m}x_{i}\log \left( \frac{x_{i}}{y_{i}}\right)
=\sum_{i=1}^{m}x_{i}\log x_{i}-\sum_{i=1}^{m}x_{i}\log y_{i}
\end{equation*}%
(with $0\log 0=0$ by continuity); $\mathcal{D}(x\parallel y)$ is always
nonnegative, and it equals zero if and only if $y=x$ (like the quadratic $%
\left\Vert x-y\right\Vert ^{2}$). Given sequences $\mathbf{a}_{t}$ and $%
\mathbf{c}_{t}$ in $\Delta ^{t}$, we define the \emph{logarithmic score} $%
\mathcal{L}$, the \emph{log-refinement score} $\mathcal{R}^{\mathcal{L}}$,%
\emph{\ }and the \emph{log-calibration score} $\mathcal{K}^{\mathcal{L}}$, by%
\begin{eqnarray}
\mathcal{L}_{t} &%
{\;:=\;}%
&\frac{1}{t}\sum_{s=1}^{t}\mathcal{D}(a_{s}\parallel c_{s}),  \notag \\
\mathcal{R}_{t}^{\mathcal{L}} &%
{\;:=\;}%
&\frac{1}{t}\sum_{s=1}^{t}\mathcal{D}(a_{s}\parallel \bar{a}_{t}(c_{s})),%
\text{\ \ \ and}  \label{eq:log-scores} \\
\mathcal{K}_{t}^{\mathcal{L}} &%
{\;:=\;}%
&\frac{1}{t}\sum_{s=1}^{t}\mathcal{D}(\bar{a}_{t}(c_{s})\parallel c_{s}). 
\notag
\end{eqnarray}%
This amounts to replacing all quadratic terms, such as $\left\Vert
a_{s}-c_{s}\right\Vert ^{2}$ in the Brier score $\mathcal{B}$, with the
corresponding logarithmic terms, such as $\mathcal{D}(a_{s}\parallel c_{s})$
in $\mathcal{L}$. Now%
\begin{equation*}
\mathcal{D}(a\parallel c)=L(a,c)-L(a,a)=L(a,c)-H(a),
\end{equation*}%
where 
\begin{equation*}
L(a,c)%
{\;:=\;}%
-\sum_{i=1}^{m}a_{i}\log c_{i}
\end{equation*}%
is the \emph{cross entropy }of $c$\emph{\ }with respect to\emph{\ }$a$, and 
\begin{equation*}
H(a)%
{\;:=\;}%
L(a,a)=-\sum_{i=1}^{m}a_{i}\log a_{i}
\end{equation*}%
is the \emph{entropy} of $a$; thus $L(a,c)\geq L(a,a)=H(a)$, with equality
if and only if $c=a.$ Summing by bins and using the linearity of $L$ in its
first argument, the log scores (\ref{eq:log-scores}) can thus be rewritten as%
\begin{eqnarray}
\mathcal{L}_{t} &=&\sum_{x\in \Delta }\left( \frac{n_{t}(x)}{t}\right) L(%
\bar{a}_{t}(x),x)-\mathcal{H}_{t},  \notag \\
\mathcal{R}_{t}^{\mathcal{L}} &=&\sum_{x\in \Delta }\left( \frac{n_{t}(x)}{t}%
\right) H(\bar{a}_{t}(x))-\mathcal{H}_{t},\text{\ \ \ and}
\label{eq:log-scores-2} \\
\mathcal{K}_{t}^{\mathcal{L}} &=&\sum_{x\in \Delta }\left( \frac{n_{t}(x)}{t}%
\right) \mathcal{D}(\bar{a}_{t}(x)\parallel x),  \notag
\end{eqnarray}%
where%
\begin{equation*}
\mathcal{H}_{t}%
{\;:=\;}%
\frac{1}{t}\sum_{s=1}^{t}H(a_{s}),
\end{equation*}%
the average entropy of the actions, depends only on the actions $a_{s}$ and
not on the forecasts $c_{s}$ (and so one could work throughout with $%
\mathcal{L}_{t}^{\prime }:=\mathcal{L}_{t}+\mathcal{H}_{t}=(1/t)\sum_{s\leq
t}L(a_{s},c_{s}),$ another standard version of the logarithmic scoring rule,
instead of $\mathcal{L}_{t}$).

The following properties of the log scores, completely parallel to those of
the quadratic scores, are now easy to see from (\ref{eq:log-scores}) and (%
\ref{eq:log-scores-2}):

\begin{itemize}
\item $\mathcal{L}_{t}\geq 0,$ and $\mathcal{L}_{t}=0$ if and only if $%
c_{s}=a_{s}$ for all $1\leq s\leq t$ (i.e., $\mathcal{B}_{t}=0$: every
forecast is equal to the action).

\item $\mathcal{R}_{t}^{\mathcal{L}}\geq 0,$ with equality if and only if $%
a_{s}=\bar{a}_{t}(c_{s})$ for all $1\leq s\leq t$ (i.e., $\mathcal{R}_{t}=0$%
: all actions in the same forecasting bin are the same).

\item $\mathcal{K}_{t}^{\mathcal{L}}\geq 0,$ with equality if and only if $%
c_{s}=\bar{a}_{t}(c_{s})$ for all $1\leq s\leq t$ (i.e., $\mathcal{K}_{t}=0$%
: every forecast is equal to the bin average action).

\item $\mathcal{L}_{t}=\mathcal{R}_{t}^{\mathcal{L}}+\mathcal{K}_{t}^{%
\mathcal{L}}$ (by (\ref{eq:log-scores-2}) and $L(y,x)=H(y)+\mathcal{D}%
(y\parallel x)$ for every $x$ and $y=\bar{a}_{t}(x)$).

\item $\mathcal{R}_{t}^{\mathcal{L}}=\min_{\phi }\mathcal{L}_{t}^{\phi (%
\mathbf{c})}$, where the minimum is taken over all bin relabelings $\phi
:C\rightarrow C$, and is attained when $\phi (x)=\bar{a}_{t}(x)$ for every $%
x,$ i.e., when the label of the $x$-bin is changed to the average $\bar{a}%
_{t}(x)$ of the bin (cf. (\ref{eq:R=minB})).

\item $\mathcal{L}_{t}$ is a strictly proper scoring rule (like the Brier
score; see Section \ref{sus:scores}).
\end{itemize}

\noindent Thus, the log-refinement score $\mathcal{R}^{\mathcal{L}}$ depends
only on the binning and not on the bin labels, whereas the log-calibration
score $\mathcal{K}^{\mathcal{L}}$ depends only on the bin averages and
labels.

We will now show that our results for the quadratic scores hold for the log
scores as well. Specifically, we will prove Theorem \ref{th:beat-1-log}, the
basic log-calibeating result (the parallel of Theorem \ref{th:beat-1}) and
Theorem \ref{th:log-calibration}, the log-calibration as
self-log-calibeating (the parallel of Theorem \ref{th:calibration}), from
which the other results follow, just as in the quadratic case.

For this we need the online version of the log-refinement score, where each
offline average $\bar{a}_{t}(c_{s})$ at time $t$ is replaced with the online
average $\bar{a}_{s-1}(c_{s})$ at time $s-1$. In order to avoid infinite
values, we use a \textquotedblleft regularized" version $\bar{a}%
_{s-1}^{\prime }(c_{s})$ instead, where for every $t\geq 1$ and $x\in \Delta 
$ we define%
\begin{equation*}
\bar{a}_{t}^{\prime }(x)%
{\;:=\;}%
\frac{1}{n_{t}(x)+1}\left( \sum_{1\leq s\leq t:c_{s}=x}a_{s}+g_{0}\right)
\end{equation*}%
with $g_{0}:=(1/m,...,1/m)$; thus, $\bar{a}_{t}^{\prime }(x)\gg 0$ and $\bar{%
a}_{t}^{\prime }(x)-\bar{a}_{t}(x)\rightarrow 0$ as $n_{t}(x)\rightarrow
\infty $. This amounts to starting each bin with an initial strictly
positive element $g_{0}$. We then define the (regularized) \emph{online
refinement score} by%
\begin{equation*}
\widetilde{\mathcal{R}}_{t}^{\mathcal{L}}%
{\;:=\;}%
\frac{1}{t}\sum_{s=1}^{t}\mathcal{D}(a_{s}\parallel \bar{a}_{s-1}^{\prime
}(c_{s})).
\end{equation*}

Assuming for simplicity that there are finitely many forecasts---and thus
bins---the parallel result to Proposition \ref{p:online-R} is

\begin{proposition}
\label{p:online-R-log}Let $(a_{t})_{t\geq 1}$ be a sequence of actions in $%
\Delta $, and let $(c_{t})_{t\geq 1}$ be a sequence of forecasts in a finite
set $D\subset \Delta $. Then, as $t\rightarrow \infty $, we have 
\begin{equation*}
\sup_{\mathbf{a}_{t}\in \Delta ^{t},\mathbf{c}_{t}\in D^{t}}\left[ 
\widetilde{\mathcal{R}}_{t}^{\mathcal{L}}-\mathcal{R}_{t}^{\mathcal{L}}%
\right] \leq O\left( \frac{\log t}{t}\right) .
\end{equation*}
\end{proposition}

As in Section \ref{s:online R}, we first prove the result for a single
bin---Proposition \ref{p:lt}, the parallel of Proposition \ref{p:var}---from
which it easily extends to finitely many bins.

Let thus $(x_{t})_{t\geq 1}$ be a sequence of points in $\Delta $; for every 
$t\geq 0$ put%
\begin{eqnarray*}
X_{t} &%
{\;:=\;}%
&\sum_{s=1}^{t}x_{s}, \\
\bar{x}_{t} &%
{\;:=\;}%
&\frac{1}{t}X_{t}, \\
\bar{x}_{t}^{\prime } &%
{\;:=\;}%
&\frac{1}{t+1}(X_{t}+g_{0}),\text{\ \ and} \\
\ell _{t} &%
{\;:=\;}%
&\frac{1}{t}\sum_{s=1}^{t}L(x_{s},\bar{x}_{s-1}^{\prime }).
\end{eqnarray*}

\begin{proposition}
\label{p:lt}As $t\rightarrow \infty $ we have%
\begin{equation*}
\sup_{\mathbf{x}_{t}\in \Delta ^{t}}\left[ \ell _{t}-L(\bar{x}_{t},\bar{x}%
_{t})\right] \leq O\left( \frac{\log t}{t}\right) .
\end{equation*}
\end{proposition}

\begin{proof}
First, we have 
\begin{equation}
L(\bar{x}_{t},\bar{x}_{t}^{\prime })-L(\bar{x}_{t},\bar{x}_{t})\leq \frac{1}{%
t}.  \label{eq:x-bar}
\end{equation}%
Indeed, $(t+1)\bar{x}_{t}^{\prime }\geq X_{t}=t\bar{x}_{t},$ and so $\log 
\bar{x}_{i,t}^{\prime }\geq \log \bar{x}_{i,t}+\log (t/(1+t))\geq \log \bar{x%
}_{i,t}-1/t$ for every $i$; multiplying by $\bar{x}_{i,t}$ and summing over $%
i$ yields (\ref{eq:x-bar}) (use $\sum_{i}\bar{x}_{i,t}=1$).

Second, we will show that%
\begin{equation}
\ell _{t}-L(\bar{x}_{t},\bar{x}_{t}^{\prime })\leq O\left( \frac{\log t}{t}%
\right) ;  \label{eq:lt}
\end{equation}%
together with (\ref{eq:x-bar}) it yields the result.

To prove (\ref{eq:lt}), we start with the identity%
\begin{eqnarray*}
L(x_{s},\bar{x}_{s-1}^{\prime }) &=&L(X_{s},\bar{x}_{s-1}^{\prime
})-L(X_{s-1},\bar{x}_{s-1}^{\prime }) \\
&=&[L(X_{s},\bar{x}_{s}^{\prime })-L(X_{s-1},\bar{x}_{s-1}^{\prime
})]+[L(X_{s},\bar{x}_{s-1}^{\prime })-L(X_{s},\bar{x}_{s}^{\prime })].
\end{eqnarray*}%
The sum over $s=1,...,t$ of the first terms, $L(X_{s},\bar{x}_{s}^{\prime
})-L(X_{s-1},\bar{x}_{s-1}^{\prime })$, telescopes to $L(X_{t},\bar{x}%
_{t}^{\prime })$ (because $X_{0}=0$ and $\bar{x}_{0}^{\prime }=g_{0}$, and
so $L(X_{0},\bar{x}_{0}^{\prime })=0$), and hence, putting%
\begin{equation*}
\Lambda _{t}%
{\;:=\;}%
\sum_{s=1}^{t}[L(X_{s},\bar{x}_{s-1}^{\prime })-L(X_{s},\bar{x}_{s}^{\prime
})]
\end{equation*}%
for the sum of the second terms, $L(X_{s},\bar{x}_{s-1}^{\prime })-L(X_{s},%
\bar{x}_{s}^{\prime })$, we get%
\begin{equation*}
t\ell _{t}=L(X_{t},\bar{x}_{t}^{\prime })+\Lambda _{t}=tL(\bar{x}_{t},\bar{x}%
_{t}^{\prime })+\Lambda _{t}.
\end{equation*}%
We will now prove that the maximum of $\Lambda _{t}$ over all possible
sequences $x_{1},...,x_{t}$ is $O(\log t)$.

Put%
\begin{equation*}
\lambda _{i,s}%
{\;:=\;}%
X_{i,s}\left[ \log (X_{i,s}+\alpha )-\log (X_{i,s-1}+\alpha )\right] \text{\
\ \ and\ \ \ }\lambda _{s}%
{\;:=\;}%
\sum_{i=1}^{m}\lambda _{i,s},
\end{equation*}%
where $\alpha :=1/m$ and we write $x_{i,s}$ for the $i$th coordinate of $%
x_{s}$ (and similarly for the other vectors); then $X_{i,s}[\log \bar{x}%
_{i,s}^{\prime }-\log \bar{x}_{i,s-1}^{\prime }]=\lambda _{i,s}-X_{i,s}[\log
(s+1)-\log s],$ and so, summing over $i=1,...,m$, we get%
\begin{equation*}
L(X_{s},\bar{x}_{s-1}^{\prime })-L(X_{s},\bar{x}_{s}^{\prime })=\lambda
_{s}-s[\log (s+1)-\log s]
\end{equation*}%
(because $x_{t}\in \Delta $ for every $t,$ and so $\sum_{i}X_{i,s}=s$).
Summing over $s=1,...,t$ yields%
\begin{equation*}
\Lambda _{t}=\sum_{s=1}^{t}\lambda _{s}-\sum_{s=1}^{t}s[\log (s+1)-\log s].
\end{equation*}

Take $1\leq r\leq t$; the function $\lambda _{i,s}$ is a convex function of $%
x_{i,r}$, and thus $\lambda _{s}=\sum_{i=1}^{m}\lambda _{i,s}$ is a convex
function of the vector $x_{r}$. Therefore $\sum_{s=1}^{t}\lambda _{s}$ is a
convex function of $x_{r}$, from which it follows that $\Lambda _{t}$ is
maximal when $x_{r}$ is a unit vector. When all the $x_{r}$ for $1\leq r\leq
t$ are unit vectors we get%
\begin{equation*}
\sum_{s=1}^{t}\lambda _{i,s}=\sum_{k=1}^{n_{i}}k\left[ \log \left( k+\alpha
\right) -\log \left( k-1+\alpha \right) \right] ,
\end{equation*}%
where $n_{i}$ ($=X_{i,t}$) is the number of times up to $t$ that $x_{s}$
equals the $i$th unit vector (these are the times when the sequence $%
(X_{i,s})_{s\geq 0}$ increases, by $1$).

The function $\xi \mapsto \xi \lbrack \log (\xi +\alpha )-\log (\xi
-1+\alpha )]$ is decreasing in $\xi ,$ and so, in order for $%
\sum_{s=1}^{t}\lambda _{s}$, and thus $\Lambda _{t}$, to be maximal, the $%
n_{i}$-s should be as close to equal as possible, i.e., $|n_{i}-n_{j}|\leq 1$
for all $i,j$ (if, say, $n_{1}\geq n_{2}+2$ then replacing in the sequence $%
(x_{s})_{1\leq s\leq t}$ one instance of $v_{1}$ with $v_{2}$ will increase $%
\Lambda _{t}$). Let $t=mr$ for simplicity, then $n_{i}=r$ for every $i$, and
so%
\begin{equation*}
\Lambda _{mr}=m\sum_{k=1}^{r}k[\log \left( k+\alpha \right) -\log \left(
k-1+\alpha \right) ]-\sum_{s=1}^{mr}s[\log (s+1)-\log s].
\end{equation*}%
Consider the second sum; opening the square brackets yields%
\begin{equation}
\sum_{s=1}^{mr}s[\log (s+1)-\log s]=mr\log (mr+1)-\sum_{s=1}^{mr}\log s.
\label{eq:sum2}
\end{equation}%
Similarly, the first sum is%
\begin{eqnarray}
&\,&m\sum_{k=1}^{r}k[\log \left( k+\alpha \right) -\log \left( k-1+\alpha
\right) ]=mr\log (r+\alpha )-m\sum_{k=0}^{r-1}\log (k+\alpha )  \notag \\
&\;&\;\;\;\;\;\;\;\;\;\;\;\;\;\;\;\;\;\;\leq mr\log
(mr+1)-m\sum_{k=1}^{r-1}\log k-(mr-m)\log m,  \label{eq:sum1}
\end{eqnarray}%
where we have used $\log (r+\alpha )=\log (mr+1)-\log m$ and $%
\sum_{k=0}^{r-1}\log (k+\alpha )=\log \alpha +\sum_{k=1}^{r-1}\log (k+\alpha
)\geq -\log m+\sum_{k=1}^{r-1}\log k$. Subtracting (\ref{eq:sum2}) from (\ref%
{eq:sum1}) gives%
\begin{equation}
\Lambda _{mr}\leq \sum_{s=1}^{mr}\log s-m\sum_{k=1}^{r-1}\log k-(mr-m)\log m.
\label{eq:mr}
\end{equation}%
We will show that the expression on the right-hand side of (\ref{eq:mr}) is $%
O(\log r)$ as $r\rightarrow \infty $ (the dimension $m$ is fixed). The
Stirling approximation, $n!\sim \sqrt{2\pi }n^{n+1/2}e^{-n}$, gives%
\begin{equation*}
\sum_{s=1}^{n}\log s=\log (n!)=\left( n+\frac{1}{2}\right) \log n-n+\log
\left( \sqrt{2\pi }\right) +o(1);
\end{equation*}%
using it in (\ref{eq:mr}) yields%
\begin{eqnarray*}
\Lambda _{mr} &\leq &\left[ \left( mr+\frac{1}{2}\right) \log (mr)-mr\right]
-m\left[ \left( r-\frac{1}{2}\right) \log (r-1)-(r-1)\right] \\
&&-(mr-m)\log m+O(1).
\end{eqnarray*}%
Since $\log (mr)=\log m+\log r\leq \log m+\log (r-1)+1/(r-1)$, we finally get%
\begin{equation*}
\Lambda _{mr}\leq \frac{m+1}{2}\log (r-1)+O(1).
\end{equation*}%
Thus, $\Lambda _{mr}\leq O(\log r)=O(\log (mr))$, which completes the proof
of (\ref{eq:lt}), and hence of the proposition.
\end{proof}

\bigskip

\begin{proof}[Proof of Proposition \protect\ref{p:online-R-log}]
We have%
\begin{eqnarray*}
\widetilde{\mathcal{R}}_{t}^{\mathcal{L}}-\mathcal{R}_{t}^{\mathcal{L}}
&=&\sum_{x\in D}\left( \frac{n_{t}(x)}{t}\right) \left[ \frac{1}{n_{t}(x)}%
\sum_{s\leq t:c_{s}=x}L(a_{s},\bar{a}_{s-1}^{\prime }(x))-L(\bar{a}_{t}(x),%
\bar{a}_{t}(x))\right] \\
&\leq &\sum_{x\in D}\left( \frac{n_{t}(x)}{t}\right) O\left( \frac{\log
n_{t}(x)}{n_{t}(x)}\right) \leq O\left( \frac{\log t}{t}\right) ,
\end{eqnarray*}%
where the first inequality is by Proposition \ref{p:lt} applied to each $x$%
-bin separately, and the second inequality is by the concavity of the $\log $
function.
\end{proof}

\bigskip

We thus get a simple way to log-calibeat, which is the parallel of Theorem %
\ref{th:beat-1}.

\begin{theorem}
\label{th:beat-1-log}Let $B$ be a finite set, and let $\zeta $ be the
deterministic $\mathbf{b}$-based forecasting procedure given by 
\begin{equation*}
c_{t}=\bar{a}_{t-1}^{\prime \,\mathbf{b}}(b_{t})
\end{equation*}%
for every time $t\geq 1$. Then $\zeta $ is $B$-log-calibeating; specifically,%
\begin{equation*}
\sup_{\mathbf{a}_{t}\in \Delta ^{t},\mathbf{b}_{t}\in B^{t}}\left[ \mathcal{L%
}_{t}^{\mathbf{c}}-\mathcal{R}_{t}^{\mathcal{L},\mathbf{b}}\right] \leq
O\left( \frac{\log t}{t}\right)
\end{equation*}%
as $t\rightarrow \infty $.
\end{theorem}

\begin{proof}
By Proposition \ref{p:online-R-log}, since our choice of $c_{t}=\bar{a}%
_{t-1}^{\prime \,\mathbf{b}}(b_{t})$ makes $\mathcal{L}_{t}^{\mathbf{c}}=%
\widetilde{\mathcal{R}}_{t}^{\mathcal{L},\mathbf{b}}$.
\end{proof}

\bigskip

Next, we obtain log-calibration by self-log-calibeating. Theorem \ref%
{th:log-outgoing-MM} below will provide the needed tool: an appropriate
minimax (\textquotedblleft stochastic fixed point") result, similar to
Theorem \ref{th:outgoing} (S). Let $\delta >0$; a set $D\subseteq \Delta $
is a $\delta $\emph{-log-grid} of $\Delta $ if for every $x\in \Delta $
there is $d\in D$ such that $\mathcal{D}(x\parallel d)<\delta $. Finite $%
\delta $-log-grids are, for instance, finite $\delta ^{\prime }$-grids of $%
\{x\in \Delta :x_{i}\geq \delta ^{\prime }$ for all $i\}$ (one needs to stay
away from the boundary where the slope of $-\log $ becomes infinite) for
appropriate $\delta ^{\prime }>0$.

\begin{theorem}
\label{th:log-outgoing-MM}Let $D\subset \Delta $ be a finite $\delta $%
-log-grid of $\Delta $, and let $g:D\rightarrow \Delta $ be an arbitrary
function. Then there exists a distribution $\eta $ on $D$ that is of type MM
such that%
\begin{equation*}
\mathbb{E}_{c\sim \eta }\left[ \mathcal{D}(a\parallel c)-\mathcal{D}%
(a\parallel g(c))\right] =\mathbb{E}_{c\sim \eta }\left[ L(a,c)-L(a,g(c))%
\right] \leq \delta
\end{equation*}%
for every $a\in \Delta $.
\end{theorem}

\begin{proof}
Consider the finite two-person zero-sum game where the maximizer chooses $%
a\in D$, the minimizer chooses $c\in D$, and the payoff is $%
L(a,c)-L(a,g(c)). $ For every mixed strategy $\nu \in \Delta (D)$ of the
maximizer, let $\bar{a}:=\mathbb{E}_{a\sim \nu }\left[ a\right] $ be its
expectation; then%
\begin{equation*}
\mathbb{E}_{a\sim \nu }\left[ L(a,c)-L(a,g(c))\right] =L(\bar{a},c)-L(\bar{a}%
,g(c))\leq \mathcal{D}(\bar{a}\parallel c)
\end{equation*}%
(the inequality is by $L(\bar{a},\cdot )\geq H(\bar{a})$), and so the
minimizer can choose $c$ in the $\delta $-log-grid $D$ that makes the payoff 
$<\delta $. By the minimax theorem, there is therefore a mixed strategy $%
\eta \in \Delta (D)$ of the minimizer that \emph{guarantees} that the payoff
is $\leq \delta $.
\end{proof}

\bigskip

The parallel result of Theorem \ref{th:calibration} is

\begin{theorem}
\label{th:log-calibration}Let $\delta >0$ and let $D\subset \Delta $ be a
finite $\delta $-log-grid of $\Delta .$ Then there exists a stochastic $D$%
-forecasting procedure $\sigma $ that is $\delta $-log-calibrated.
\end{theorem}

\begin{proof}
For every history $h_{t-1}=(\mathbf{a}_{t-1},\mathbf{c}_{t-1}),$ applying
Theorem \ref{th:log-outgoing-MM} to the function $g(c)=\bar{a}_{t-1}^{\prime
}(c)$ yields a probability distribution on $D$ such that, by using it as the
distribution $\sigma (h_{t-1})$ of the forecast $c_{t},$ we have%
\begin{equation*}
\mathbb{E}_{t-1}\left[ L(a_{t},c_{t})-L(a_{t},\bar{a}_{t-1}^{\prime }(c_{t}))%
\right] \leq \delta
\end{equation*}%
for every $a_{t}\in \Delta .$ Taking overall expectation and averaging over $%
t$ yields $\mathbb{E}\left[ \mathcal{L}_{t}-\widetilde{\mathcal{R}}_{t}^{%
\mathcal{L}}\right] \leq \delta $; together with Proposition \ref%
{p:online-R-log} we get $\mathbb{E}\left[ \mathcal{K}_{t}^{\mathcal{L}}%
\right] =\mathbb{E}\left[ \mathcal{L}_{t}-\mathcal{R}_{t}^{\mathcal{L}}%
\right] \leq \delta +O\left( \log t/t\right) .$
\end{proof}

\bigskip

The other results in the paper extend in a similar way to the log scores.

\subsection{Multi-Calibeating Is Stronger Than the Stronger Expert\label%
{s-a:stronger expert}}

We show here that calibeating is a stronger notion than the so-called
\textquotedblleft stronger expert.\textquotedblright \footnote{%
Expands on the reply given to a question asked at a lecture at the Workshop
on Learning in Games, Toulouse, July 2024.}

Let $C=\Delta (A)$; as we have seen at the end of Section \ref{s:setup} (see
(\ref{eq:R=minB})), the refinement score is the minimal Brier score over all
relabelings of the bins; i.e.,%
\begin{equation*}
\mathcal{R}_{t}=\min_{\phi }\mathcal{B}_{t}^{\phi (\mathbf{c})},
\end{equation*}%
where the minimum is taken over all functions $\phi :\Delta (A)\rightarrow
\Delta (A)$ (from current labels to new labels), and we write $\mathcal{B}%
_{t}^{\phi (\mathbf{c})}$ for the Brier score where the sequence $\mathbf{c}$
is replaced by $\phi (\mathbf{c})=(\phi (c_{s}))_{s=1,2,...}$ .

Therefore, we have:

\begin{itemize}
\item $\mathbf{c}$ is \emph{multi-calibeating} $\mathbf{b}_{1},...,\mathbf{b}%
_{N}$ if 
\begin{equation*}
\mathcal{B}_{t}^{\mathbf{c}}\leq \min_{\phi }\mathcal{B}_{t}^{\phi (\mathbf{b%
}_{1},...,\mathbf{b}_{N})}+o(1)
\end{equation*}%
as\footnote{%
For simplicity we ignore here the uniformity requirement on the sequences $%
\mathbf{a},\mathbf{b}_{1},...,\mathbf{b}_{N}$.} $t\rightarrow \infty $,
where the minimum is taken over all functions $\phi :\Pi
_{n=1}^{N}B^{n}\rightarrow \Delta (A)$.
\end{itemize}

By comparison, when all forecasts are probability distributions on $A,$
i.e., all the sets $B^{n}$ are subsets of $\Delta (A)$, Foster (1991)
defines:\footnote{%
In the expansive literature on experts, these notions are referred to as
\textquotedblleft prediction with no regret"; see Appendix \ref{s-a:log} for
the parallel results with the logarithmic scoring rule.}

\begin{itemize}
\item $\mathbf{c}$ is \emph{as strong as} $\mathbf{b}_{1},...,\mathbf{b}_{N}$
if 
\begin{equation*}
\mathcal{B}_{t}^{\mathbf{c}}\leq \min_{1\leq n\leq N}\mathcal{B}_{t}^{%
\mathbf{b}_{n}}+o(1),
\end{equation*}
and

\item $\mathbf{c}$ is \emph{as strong as the convex hull} of $\mathbf{b}%
_{1},...,\mathbf{b}_{N}$ if 
\begin{equation*}
\mathcal{B}_{t}^{\mathbf{c}}\leq \min_{w}\mathcal{B}_{t}^{w_{1}\mathbf{b}%
_{1}+...w_{N}\mathbf{b}_{N}}+o(1),
\end{equation*}%
where the minimum is taken over all $w=(w_{1},...,w_{N})\geq 0$ with $%
\sum_{n=1}^{N}w_{n}=1$ (i.e., over all convex combinations of $\mathbf{b}%
_{1},...,\mathbf{b}_{N}$.
\end{itemize}

Whereas multi-calibeating takes into account \emph{all} bin relabelings,
stronger-expert concepts do not go beyond linear-combination relabelings.
Therefore, as claimed, \emph{multi-calibeating is a stronger notion than the
\textquotedblleft stronger expert."}


\begin{thebibliography}{99}
\bibitem{} Azoury, K. S. and M. K. Warmuth (2001), \textquotedblleft
Relative Loss Bounds for On-Line Density Estimation with the Exponential
Family of Distributions," \emph{Machine Learning} 43, 211--246.

\bibitem{} Blackwell, D. (1956), \textquotedblleft An Analog of the Minimax
Theorem for Vector Payoffs," \emph{Pacific Journal of Mathematics} 6, 1--8.

\bibitem{} Brier, G. W. (1950), \textquotedblleft Verification of Forecasts
Expressed in Terms of Probability," \emph{Monthly Weather Review} 78, 1--3.

\bibitem{} Brouwer, L. E. J. (1912), \textquotedblleft \"{U}ber Abbildung
von Mannigfaltigkeiten," \emph{Mathematische Annalen} 71, 97--115.

\bibitem{} Cesa-Bianchi, N. and G. Lugosi (2006), \emph{Prediction,
Learning, and Games}, Cambridge University Press.

\bibitem{} Dawid, A. (1982), \textquotedblleft The Well-Calibrated
Bayesian," \emph{Journal of the American Statistical Association} 77,
605--613.

\bibitem{} Forster, J. (1999), \textquotedblleft On Relative Loss Bounds in
Generalized Linear Regression," in \emph{12th International Symposium on
Fundamentals of Computation Theory (FCT '99)}, 269--280.

\bibitem{} Foster, D. P. (1991), \textquotedblleft Prediction in the Worst
Case," \emph{The Annals of Statistics} 19, 1084--1090.

\bibitem{} Foster, D. P. (1999), \textquotedblleft A Proof of Calibration
via Blackwell's Approachability Theorem," \emph{Games and Economic Behavior }%
29, 73--78.

\bibitem{} Foster, D. P. and S. Hart (2018), \textquotedblleft Smooth
Calibration, Leaky Forecasts, Finite Recall, and Nash Dynamics," \emph{Games
and Economic Behavior }109, 271--293.

\bibitem{} Foster, D. P. and S. Hart (2021), \textquotedblleft Forecast
Hedging and Calibration," \emph{Journal of Political Economy} 129,
3447--3490. doi.org/10.1086/716559. \texttt{%
http://www.ma.huji.ac.il/hart/publ.html\#calib-int}

\bibitem{} Foster, D. P. and S. Hart (2023), \textquotedblleft \thinspace
`Calibeating': Beating Forecasters at Their Own Game," \emph{Theoretical
Economics }18, 1441--1474.

\bibitem{} Foster, D. P. and S. Hart (2024), \textquotedblleft Addendum:
`Calibeating': Beating Forecasters at Their Own Game," \texttt{%
http://www.ma.huji.ac.il/hart/papers/calib-beat-add.pdf}

\bibitem{} Foster, D. P. and S. Hart (2026), \textquotedblleft Errata:
`Calibeating': Beating Forecasters at Their Own Game," \texttt{%
http://www.ma.huji.ac.il/hart/papers/calib-beat-errata.pdf}

\bibitem{} Foster, D. P. and R. V. Vohra (1998), \textquotedblleft
Asymptotic Calibration," \emph{Biometrika} 85, 379--390.

\bibitem{} Hart, S. (1992), \textquotedblleft Games in Extensive and
Strategic Forms," in \emph{Handbook of Game Theory, with Economic
Applications}, R. J. Aumann and S. Hart (editors), North-Holland, Vol. 1,
Chapter 2, 19--40.

\bibitem{} Hart, S. (2021), \textquotedblleft Calibrated Forecasts: The
Minimax Proof," Center for Rationality DP-744, The Hebrew University of
Jerusalem. \texttt{http://www.ma.huji.ac.il/hart/publ.html\#calib-minmax}

\bibitem{} Johnson, N. L., A. W. Kemp, and S. Kotz (2005), \emph{Univariate
Discrete Distributions}, 3rd edition, Wiley.

\bibitem{} Jung, H. (1901), \textquotedblleft Ueber die kleinste Kugel, die
eine r\"{a}umliche Figur einschliesst," \emph{Journal f\"{u}r die Reine und
Angewandte Mathematik} 123, 241--257.

\bibitem{} Kuhn, H. W. (1953), \textquotedblleft Extensive Games and the
Problem of Information,\textquotedblright\ in \emph{Contributions to the
Theory of Games, Vol. II}, H. W. Kuhn and A. W. Tucker (editors), \emph{%
Annals of Mathematics Studies} 28, Princeton University Press, 193--216.

\bibitem{} Lo\`{e}ve, M. (1978), \emph{Probability Theory, Vol. II}, 4th
edition, Springer.

\bibitem{} Murphy, A. H. (1972), \textquotedblleft Scalar and Vector
Partitions of the Probability Score. Part I: Two-State Situation," \emph{%
Journal of Applied Meteorology} 11, 273--282.

\bibitem{} Oakes, D. (1985), \textquotedblleft Self-Calibrating Priors Do
Not Exist," \emph{Journal of the American Statistical Association} 80, 339.

\bibitem{} Olszewski, W. (2015), \textquotedblleft Calibration and Expert
Testing," in \emph{Handbook of Game Theory, Vol. 4}, H. P. Young and S.
Zamir (editors), Springer, 949--984.

\bibitem{} Olszewski, W. and A. Sandroni (2008), \textquotedblleft
Manipulability of Future-Independent Tests," \emph{Econometrica} 76,
1437--1466.

\bibitem{} Sanders, F. (1963), \textquotedblleft On Subjective Probability
Forecasting," \emph{Journal of Applied Meteorology} 2, 191--201.

\bibitem{} Sandroni, A. (2003), \textquotedblleft The Reproducible
Properties of Correct Forecasts," \emph{International Journal of Game Theory}
32, 151--159.

\bibitem{} Shmaya, E. (2008), \textquotedblleft Many Inspections are
Manipulable," \emph{Theorerical Economics} 3, 367--382.

\bibitem{} Welford, B. P. (1962), \textquotedblleft Note on a Method for
Calculating Corrected Sums of Squares and Products,\textquotedblright\ \emph{%
Technometrics} 4, 419--420.

\bibitem{} Vovk, V. (2001), \textquotedblleft Competitive On-Line
Statistics," \emph{International Statistical Review} 69, 213--248.

\bibitem{} von Neumann, J. (1928), \textquotedblleft Zur Theorie der
Gesellschaftsspiele, " \emph{Mathematische Annalen} 100, 295--320.
\end{thebibliography}
\end{document}


\title{Errata: \textquotedblleft Calibeating\textquotedblright : Beating
Forecasters at Their Own Game\thanks{%
Errata to the published journal version (\emph{Theoretical Economics}, 2023)
and to \texttt{arXiv} version 2 (2022).}}
\author{Dean P. Foster\thanks{%
Department of Statistics, Wharton, University of Pennsylvania, Philadelphia,
and Amazon, New York. \emph{e-mail}: \texttt{dean@foster.net} \ \emph{web
page}: \texttt{http://deanfoster.net}} \and Sergiu Hart\thanks{%
Institute of Mathematics, Department of Economics, and Federmann Center for
the Study of Rationality, The Hebrew University of Jerusalem. \emph{e-mail}: 
\texttt{hart@huji.ac.il} \ \emph{web page}: \texttt{%
http://www.ma.huji.ac.il/hart}}}
\maketitle

\def\@biblabel#1{#1\hfill}
\def\thebibliography#1{\section*{References}
\addcontentsline{toc}{section}{References}
\list
{}{
\labelwidth 0pt
\leftmargin 1.8em
\itemindent -1.8em
\usecounter{enumi}}
\def\newblock{\hskip .11em plus .33em minus .07em}
\sloppy\clubpenalty4000\widowpenalty4000
\sfcode`\.=1000\relax\def\baselinestretch{1}\large \normalsize}
\let\endthebibliography=\endlist
\newcommand{\shfrac}[2]{\ensuremath{{}^{#1} \hspace{-0.04in}/_{\hspace{-0.03in}#2}}}
\newcommand\T{\rule{0pt}{2.6ex}}
\newcommand\B{\rule[-1.2ex]{0pt}{0pt}}%

While the result of Theorem 6 (in Section 6) in our paper, Foster and Hart
(2023)---namely, the existence of a deterministic procedure that is
calibeating and continuously calibrated---is correct, its more detailed
version, Theorem 12 (in Appendix A.6\footnote{%
In the full paper in \texttt{arXiv}, the numbering of the relevant
appendices is shifted up by one; Appendix A.6 is thus Appendix A.7 in the
full paper.}), is not: the proof that we gave contains an error. We provide
here a more refined statement of Theorem 12, together with the
corresponding, more extensive, proof.

The error was caused by using the same notation for two distinct concepts of
refinement: one is the average within-bin variance of the \emph{actions} $%
a_{t}$ (Section 2.1), the other is the average within-bin variance of the 
\emph{differences} between actions and forecasts $a_{t}-c_{t}$ (Appendix
A.6). When each bin contains forecasts with the same value, i.e., $c_{t}$ is
constant within the bin, the variance of $a_{t}-c_{t}$ equals the variance
of $a_{t}$, and the two definitions coincide. However, this is not the case
for fractional binnings, where a bin usually contains different forecasts.
See the end of this errata for details regarding the specific error.

We provide here the necessary corrections for the paper. They are all within
Appendix A.6, and consist of a sharpened notation and a revised statement of
Theorem 12, together with its expanded proof; the remainder of Appendix
A.6---namely, Proposition 13 and its proof---needs no correction beyond
using the new notation.

\bigskip

\noindent \textbf{Notation}. When we deal with the variance of the
differences $a_{t}-c_{t}$ rather than the variance of the actions $a_{t}$ we
will add a superscript $\#$. Thus, for the fractional binning $(\mathbf{b}%
,\Pi )$, the variance denoted by $v_{t}(b,i)$ in the paper is now $%
v_{t}^{\#}(b,i)$, i.e., 
\begin{eqnarray*}
v_{t}^{\#}(b,i) &%
{\;:=\;}%
&\sum_{s=1}^{t}\left( \frac{\lambda _{s}(b,i)}{n_{t}(b,i)}\right) \left\Vert
a_{s}-c_{s}-e_{t}(b,i)\right\Vert ^{2} \\
&=&\sum_{s=1}^{t}\left( \frac{\lambda _{s}(b,i)}{n_{t}(b,i)}\right)
\left\Vert a_{s}-c_{s}-\left( \bar{a}_{t}(b,i)-\bar{c}_{t}(b,i)\right)
\right\Vert ^{2},
\end{eqnarray*}%
where 
\begin{eqnarray*}
\bar{a}_{t}(b,i) &%
{\;:=\;}%
&\sum_{s=1}^{t}\left( \frac{\lambda _{s}(b,i)}{n_{t}(b,i)}\right) a_{s}\text{%
\ \ \ and} \\
\bar{c}_{t}(b,i) &%
{\;:=\;}%
&\sum_{s=1}^{t}\left( \frac{\lambda _{s}(b,i)}{n_{t}(b,i)}\right) c_{s}
\end{eqnarray*}%
are the average action and average forecast in the $(b,i)$-bin, and $%
e_{t}(b,i)=\bar{a}_{t}(b,i)-\bar{c}_{t}(b,i)$ is the average difference. The
refinement score denoted by $\mathcal{R}_{t}^{\mathbf{b},\Pi }$ in the paper
is now $\mathcal{R}_{t}^{\#;\mathbf{b},\Pi }$, i.e., 
\begin{equation*}
\mathcal{R}_{t}^{\#;\mathbf{b},\Pi }%
{\;:=\;}%
\sum_{(b,i)\in B\times I}\left( \frac{n_{t}(b,i)}{t}\right)
v_{t}^{\#}(b,i)=\sum_{(b,i)\in B\times I}\sum_{s=1}^{t}\left( \frac{\lambda
_{s}(b,i)}{t}\right) v_{t}^{\#}(b,i),
\end{equation*}%
and the online refinement score $\widetilde{\mathcal{R}}_{t}^{\mathbf{b,}\Pi
}$ is now $\widetilde{\mathcal{R}}_{t}^{\#;\mathbf{b,}\Pi }$, i.e., 
\begin{equation*}
\widetilde{\mathcal{R}}_{t}^{\#;\mathbf{b,}\Pi }%
{\;:=\;}%
\sum_{(b,i)\in B\times I}\sum_{s=1}^{t}\left( \frac{\lambda _{s}(b,i)}{t}%
\right) \left\Vert a_{s}-c_{s}-e_{s-1}(b,i)\right\Vert ^{2}.
\end{equation*}%
The same notation is of course used when there is no $\mathbf{b}$, i.e., for 
$\Pi $ alone. 

The result of Proposition 13 is now written as 
\begin{equation*}
0\leq \widetilde{\mathcal{R}}_{t}^{\#;\mathbf{b,}\Pi }-\mathcal{R}_{t}^{\#;%
\mathbf{b,}\Pi }\leq o(1).
\end{equation*}
The result of Appendix A.4, namely, the refinement score being monotonically
decreasing with respect to refining the binning, is stated for an arbitrary
sequence $z_{t}$, and so holds for the two kinds of refinement scores; thus,
for instance, $\mathcal{R}_{t}^{\#;\mathbf{b,}\Pi }\leq \mathcal{R}%
_{t}^{\#;\Pi }$ and $\mathcal{R}_{t}^{\mathbf{b,}\Pi }\leq \mathcal{R}_{t}^{%
\mathbf{b}}$. The general Brier score decomposition of Appendix A.5 for
fractional binnings is now written as $\mathcal{B}_{t}^{\mathbf{c}}=\mathcal{%
R}_{t}^{\#;\Pi }+\mathcal{K}_{t}^{\Pi }$.

\bigskip

We now state the corrected version of Theorem 12.%
\setcounter{savedtheorem}{\value{theorem}}
\setcounter{theorem}{11}%

\begin{theorem}[corrected]
Let $B$ be a finite set. Then there exists a \emph{deterministic} $\mathbf{b}
$-based forecasting procedure $\zeta $ that is $B$-calibeating and is
continuously calibrated. Specifically: first, for every continuous binning $%
\Pi $ there is a deterministic $\mathbf{b}$-based forecasting procedure $%
\zeta \equiv \zeta _{\Pi }$ such that%
\begin{equation}
\mathcal{K}_{t}^{\mathbf{b,}\Pi }=\mathcal{B}_{t}^{\mathbf{c}}-\mathcal{R}%
_{t}^{\#;\mathbf{b,}\Pi }\leq o(1);  \tag{34}
\end{equation}%
and second, there is a continuous binning $\Pi ^{\ast }$ for which (34)
implies that the corresponding procedure $\zeta _{\Pi ^{\ast }}$ is $B$%
-calibeating, i.e., 
\begin{equation*}
\mathcal{B}_{t}^{\mathbf{c}}\leq \mathcal{R}_{t}^{\mathbf{b}}+o(1),
\end{equation*}%
and is continuously calibrated. All the above inequalities hold as $%
t\rightarrow \infty $ uniformly over all sequences $\mathbf{a}$ and $\mathbf{%
b}.$
\end{theorem}

\setcounter{theorem}{\value{savedtheorem}}%

The statement of Theorem 12 in Foster and Hart (2023) claimed incorrectly
that the result holds for $\Pi ^{\ast }$ equal to the continuous binning $%
\Pi _{0}$ given by Proposition 3 in Foster and Hart (2021); as we will show
below, we need $\Pi ^{\ast }$ to be a certain \textquotedblleft expansion"
of $\Pi _{0}$.

The proof relies on a result that shows that the two kinds of refinements, $%
\mathcal{R}^{\#}$ and $\mathcal{R}$, are close when the binning is
\textquotedblleft local," in the sense that the forecasts within each bin
are almost constant.

A weight function $w:C\rightarrow \lbrack 0,1]$ is $\delta $\emph{-local}
(for $\delta >0$) if all points with positive weight (the \emph{support} of $%
w$) lie in an open ball of radius $\delta $; i.e., $\{c\in
C:w(c)>0\}\subseteq B(y;\delta )$ for some $y\in C$. A fractional binning $%
\Pi =(w_{i})_{i\in I}$ is $\delta $\emph{-local} if all the weight functions 
$w_{i}$ are $\delta $-local. Such a fractional binning $\Pi \equiv \Pi
_{\delta }$ can be obtained, for instance, by using the so-called
\textquotedblleft $\delta $-tent functions" based on a $\delta $-grid: let $%
C_{\delta }=\{y^{i}\}_{i\in I}$ be a finite $\delta $-grid of $C$, and for
each $i\in I$ let $w_{i}(c):=\Lambda (c,y^{i})/\sum_{j\in I}\Lambda
(c,y_{j}),$ where $\Lambda (c,y):=[\delta -\left\Vert c-y\right\Vert ]_{+}$;
moreover, this construction yields a continuous binning (because the
functions $\Lambda (\cdot ,y^{i})$ are continuous, and their sum is
positive, and thus bounded away from $0$ on the compact set $C$).

\bigskip 

\noindent \textbf{Lemma} \emph{Let }$\Pi $\emph{\ be a }$\delta $\emph{-local%
} \emph{fractional binning. Then}

\begin{equation*}
\left\vert \mathcal{R}_{t}^{\#;\Pi }-\mathcal{R}_{t}^{\Pi }\right\vert
<2\delta +\delta ^{2}.
\end{equation*}

\begin{proof}
Let $\Pi =(w_{i})_{i\in I},$ and for each $i\in I$ let $y^{i}\in C$ be such
that $B(y^{i};\delta )$ contains the support of $w_{i}$. Thus, all forecasts 
$c_{s}$ counted in bin $i$ (i.e., with $\lambda _{s}(i)=w_{i}(c_{s})>0$)
satisfy $\left\Vert c_{s}-y^{i}\right\Vert <\delta $.

Given random variables $X$ and $Y$ such that $\left\Vert X\right\Vert \leq 1$
and $\left\Vert Y\right\Vert <\delta $, and hence $\mathbb{V}ar[X]\leq 1$
and $\mathbb{V}ar[Y]<\delta ^{2}$, we have%
\begin{equation*}
\left\vert \mathbb{V}ar[X-Y]-\mathbb{V}ar[X]\right\vert \leq 2\sqrt{\mathbb{V%
}ar[X]\,\mathbb{V}ar[Y]}+\mathbb{V}ar[Y]<2\delta +\delta ^{2}.
\end{equation*}%
Applying this to $X=a_{s}$ and $Y=c_{s}-y^{i}$ (since $y^{i}$ is a constant
it does not affect the variance) yields%
\begin{equation*}
\left\vert v_{t}^{\#}(i)-v_{t}(i)\right\vert <2\delta +\delta ^{2},
\end{equation*}%
where 
\begin{eqnarray*}
v_{t}^{\#}(i) &=&\sum_{s=1}^{t}\left( \frac{\lambda _{s}(i)}{n_{t}(i)}%
\right) \left\Vert a_{s}-c_{s}-(\bar{a}_{t}(i)-\bar{c}_{t}(i))\right\Vert
^{2}\;\;\;\text{and} \\
v_{t}(i) &=&\sum_{s=1}^{t}\left( \frac{\lambda _{s}(i)}{n_{t}(i)}\right)
\left\Vert a_{s}-\bar{a}_{t}(i)\right\Vert ^{2}
\end{eqnarray*}%
are the variances in bin $i$ of the differences $a_{s}-c_{s}$, and of the
actions $a_{s}$, respectively. Averaging over $i\in I$, with weights $%
n_{t}(i)/t$, gives the result.
\end{proof}

\bigskip

We now prove the corrected Theorem 12. Part (i) is the original proof (for
an arbitrary $\Pi $), and part (ii) is new.

\bigskip 

\begin{proof}[Proof of Theorem 12 (corrected)]
(i) Let $\Pi =(w_{i})_{i\in I}$ be a continuous binning. At time $t$, given $%
\mathbf{a}_{t-1},\mathbf{c}_{t-1}$, and $\mathbf{b}_{t}$, applying the
outgoing fixed point result, specifically, Theorem 10 (D), to the continuous
function $c\longmapsto c+\sum_{i\in I}w_{i}(c)e_{t-1}(b_{t},i)$, yields $%
c_{t}\in C$ such that 
\begin{eqnarray*}
\left\Vert a_{t}-c_{t}\right\Vert ^{2} &\leq &\left\Vert
a_{t}-c_{t}-\sum_{i\in I}w_{i}(c_{t})e_{t-1}(b_{t},i)\right\Vert ^{2} \\
&\leq &\sum_{i\in I}w_{i}(c_{t})\left\Vert
a_{t}-c_{t}-e_{t-1}(b_{t},i)\right\Vert ^{2}
\end{eqnarray*}%
for every $a_{t}\in A$ (the second inequality is by the convexity of $%
\left\Vert \cdot \right\Vert ^{2}$). Averaging over $t$ gives $\mathcal{B}%
_{t}^{\mathbf{c}}\leq \widetilde{\mathcal{R}}_{t}^{\#;\mathbf{b},\Pi }$, and
thus (34) by Proposition 13 together with the decomposition of Appendix A.5
for the fractional binning $(\mathbf{b},\Pi )$.

(ii) We construct $\Pi ^{\ast }$ as follows. Let $\Pi _{0}=(w_{i}^{0})_{i\in
I_{0}}$ be the continuous binning given by Proposition 3 in Foster and Hart
(2021). For each $n\geq 1$ let $\Pi _{n}=(w_{i}^{n})_{i\in I_{n}}$ be a $%
\delta _{n}$-local continuous binning, where $0<\delta _{n}<1$ and $\delta
_{n}\rightarrow 0$ as $n\rightarrow \infty $; assume that the indexing sets $%
I_{n}$ for all $n\geq 0$ are taken to be disjoint. The collection of all
these weight functions, with $w_{i}^{n}$ rescaled by a factor of $1/2^{n+1}$%
, yields a continuous binning, which we denote by $\Pi ^{\ast }$; i.e., $\Pi
^{\ast }=(w_{i}^{\ast })_{i\in I^{\ast }}$ with $I^{\ast }:=\cup _{n\geq
0}I_{n}$ and $w_{i}^{\ast }:=(1/2^{n+1})w_{i}^{n}$ for each $i\in I_{n}$ and 
$n\geq 0$ (indeed, every $w_{i}^{\ast }$ is continuous, and $\sum_{i\in
I}w_{i}^{\ast }=\sum_{n\geq 0}\sum_{i\in
I_{n}}(1/2^{n+1})w_{i}^{n}=\sum_{n\geq 0}\left( 1/2^{n+1}\right) \mathbf{1}=%
\mathbf{1},$ because $\sum_{i\in I_{n}}w_{i}^{n}=\mathbf{1}$ for each $n$).

Consider now the joint fractional binning $(\mathbf{b},\Pi ^{\ast })$, and
the corresponding refinement score $\mathcal{R}_{t}^{\#;\mathbf{b},\Pi
^{\ast }}$. Separating the sum over all $i\in I$ into sums over $i\in I_{n}$
for all $n\geq 0$, and noting that rescaling a weight function does not
affect the within-bin \emph{relative} weights $\lambda _{s}(b,i)/n_{t}(b,i)=%
\mathbf{1}_{b}(b_{s})\cdot w_{i}(c_{s})/\sum_{r\leq t}w_{i}(c_{r})$ that are
used to compute the bin averages and variances, we get%
\begin{equation*}
\mathcal{R}_{t}^{\#;\mathbf{b},\Pi ^{\ast }}=\sum_{n\geq 0}\frac{1}{2^{n+1}}%
\mathcal{R}_{t}^{\#;\mathbf{b},\Pi _{n}}.
\end{equation*}%
The procedure $\zeta _{\Pi ^{\ast }}$ constructed in (i) for $\Pi ^{\ast }$
yields, by (34),%
\begin{equation*}
\mathcal{B}_{t}^{\mathbf{c}}-\mathcal{R}_{t}^{\#;\mathbf{b},\Pi ^{\ast
}}=\sum_{n\geq 0}\frac{1}{2^{n+1}}\left( \mathcal{B}_{t}^{\mathbf{c}}-%
\mathcal{R}_{t}^{\#;\mathbf{b},\Pi _{n}}\right) \leq o(1).
\end{equation*}%
By the decomposition of Appendix A.5 for $(\mathbf{b},\Pi _{n})$, each term
in the above sum equals $\mathcal{K}_{t}^{\mathbf{b},\Pi _{n}}$, and is thus
nonnegative, and so%
\begin{equation}
\mathcal{K}_{t}^{\mathbf{b},\Pi _{n}}=\mathcal{B}_{t}^{\mathbf{c}}-\mathcal{R%
}_{t}^{\#;\mathbf{b,}\Pi _{n}}\leq o(1)\text{\ \ \ for every }n\geq 0. 
\tag{E1}  \label{eq:Pi-n}
\end{equation}

For $n=0$, (\ref{eq:Pi-n}) implies that 
\begin{equation*}
\mathcal{K}_{t}^{\Pi _{0}}=\mathcal{B}_{t}^{\mathbf{c}}-\mathcal{R}%
_{t}^{\#;\Pi _{0}}\leq \mathcal{B}_{t}^{\mathbf{c}}-\mathcal{R}_{t}^{\#;%
\mathbf{b,}\Pi _{0}}\leq o(1),
\end{equation*}%
where we have used the decomposition of Appendix A.5 for $\Pi _{0}$, and $%
\mathcal{R}_{t}^{\#;\mathbf{b,}\Pi _{0}}\leq \mathcal{R}_{t}^{\#;\Pi _{0}}$
by the refining monotonicity of the refinement score (Appendix A.4). Thus,
by Proposition 3 in Foster and Hart (2021), the procedure is continuously
calibrated.

Next, for each $n\geq 1$, (\ref{eq:Pi-n}) implies that%
\begin{equation*}
\mathcal{B}_{t}^{\mathbf{c}}\leq \mathcal{R}_{t}^{\#;\mathbf{b},\Pi
_{n}}+o(1)\leq \mathcal{R}_{t}^{\mathbf{b,}\Pi _{n}}+3\delta _{n}+o(1)\leq 
\mathcal{R}_{t}^{\mathbf{b}}+3\delta _{n}+o(1),
\end{equation*}%
where we have used $\mathcal{R}_{t}^{\#;\mathbf{b},\Pi _{n}}\leq \mathcal{R}%
_{t}^{\mathbf{b},\Pi _{n}}+2\delta _{n}+\delta _{n}^{2}<\mathcal{R}_{t}^{%
\mathbf{b},\Pi _{n}}+3\delta _{n}$ by the above lemma applied to the
fractional binning $(\mathbf{b},\Pi _{n})$, which is $\delta _{n}$-local,
and $\mathcal{R}_{t}^{\mathbf{b},\Pi _{n}}\leq \mathcal{R}_{t}^{\mathbf{b}}$
(again by Appendix A.4). Therefore $\mathcal{B}_{t}^{\mathbf{c}}\leq 
\mathcal{R}_{t}^{\mathbf{b}}+4\delta _{n}$ for all $t$ large enough; since $%
\delta _{n}\rightarrow 0$, this proves that $\mathcal{B}_{t}^{\mathbf{c}}-%
\mathcal{R}_{t}^{\mathbf{b}}\leq o(1).$
\end{proof}

\bigskip

We can now pinpoint the error: it lies in the inequality $\mathcal{R}_{t}^{%
\mathbf{b},\Pi _{0}}\leq \mathcal{R}_{t}^{\mathbf{b}}$ in the paper, which
we took to mean (in the current notation) $\mathcal{R}_{t}^{\#;\mathbf{b}%
,\Pi _{0}}\leq \mathcal{R}_{t}^{\mathbf{b}}$, whereas what we showed was $%
\mathcal{R}_{t}^{\#;\mathbf{b},\Pi _{0}}\leq \mathcal{R}_{t}^{\#;\mathbf{b}}.
$ We have resolved this by supplementing $\Pi _{0}$ with continuous binnings 
$\Pi _{\delta }$ for which all forecasts within the same bin are close to
one another, and so the variance of the differences $a_{t}-c_{t}$ is close
to the variance of the actions $a_{t}$, and thus $\mathcal{R}_{t}^{\mathbf{b}%
,\Pi _{\delta }}$ is close to $\mathcal{R}_{t}^{\#;\mathbf{b},\Pi _{\delta
}}.$

\bigskip

\subsection*{Misprints}

\begin{itemize}
\item Page 1441, 5th line from the bottom: replace \textquotedblleft preset
paper" with \textquotedblleft present paper."

\item Page 1450, 7th line after formula (6): replace \textquotedblleft $\phi
:B\rightarrow B$\textquotedblright\ with \textquotedblleft $\phi
:B\rightarrow \Delta (A).$\textquotedblright 

\item Page 1457, Theorem 7 (iii): replace \textquotedblleft $(\delta ,B_{n})$%
-calibeating\textquotedblright\ with \textquotedblleft $B_{n}$%
-calibeating.\textquotedblright 

\item Page 1472, displayed formula in Remark (b): replace $\bar{a}_{t-1}^{n}$
with $\bar{a}_{i,t-1}^{n}$ (add the missing subscript $i$).
\end{itemize}